\documentclass[times,preprint]{article}
\usepackage[numbers]{natbib}
\usepackage{framed,multirow}

\usepackage{amssymb}
\usepackage{latexsym}

\usepackage{url}
\usepackage{xcolor}
\definecolor{newcolor}{rgb}{.8,.349,.1}

\usepackage{lmodern}
\usepackage[T1]{fontenc}
\usepackage{setspace}

\usepackage{pgfplots}
\pgfplotsset{compat=1.11}
\usepackage{tikz,tkz-euclide}
\usetikzlibrary{arrows,calc,patterns}

\usepackage{rotating}
\usepackage{amsmath,mathtools,bm,cancel,relsize}
\usepackage{bbm}
\usepackage{sobolev}
\usepackage{graphics,psfrag,subfigure}
\usepackage{empheq}
\usepackage{colortbl}
\usepackage{hyperref}
\usepackage{fancyref}
\usepackage{enumitem}
\usepackage{amsthm}
\usepackage{authblk}

\newtheorem{thm}{Theorem}[section]

\newtheorem{lem}[thm]{Lemma}

\numberwithin{equation}{section}

\DeclareMathOperator\supp{supp} % support of a function

\graphicspath{{./}{../figures/}}
\title{An arbitrary-order Cell Method with block-diagonal mass-matrices for the time-dependent 2D Maxwell equations}

\author[1]{Bernard Kapidani}
\author[2]{Lorenzo Codecasa}
\author[1]{Joachim Sch{\"o}berl}
\affil[1]{Institute for Analysis and Scientific Computing, Technische Universit{\"a}t Wien, A-1040, Vienna, Austria.\authorcr
    \tt bernard.kapidani@tuwien.ac.at, joachim.schoeberl@tuwien.ac.at }
\affil[2]{Dipartimento di Elettronica, Informatica e Bioingegneria, Politecnico di Milano, I-20133 Milano, Italy \authorcr
    \tt lorenzo.codecasa@polimi.it }

\date{\today}
\begin{document}

\maketitle
\begin{abstract}
  We introduce a new numerical method for the time-dependent Maxwell equations on unstructured meshes in two space dimensions. This relies on the introduction of a new mesh, which is the barycentric-dual cellular complex of the starting simplicial mesh, and on approximating two unknown fields with integral quantities on geometric entities of the two dual complexes. A careful choice of basis-functions yields cheaply invertible block-diagonal system matrices for the discrete time-stepping scheme.
  The main novelty of the present contribution lies in incorporating arbitrary polynomial degree in the approximating functional spaces, defined through a new reference cell. The presented method, albeit a kind of Discontinuous Galerkin approach, requires neither the introduction of user-tuned penalty parameters for the tangential jump of the fields, nor numerical dissipation to achieve stability. In fact an exact electromagnetic energy conservation law for the semi-discrete scheme is proved and it is shown on several numerical tests that the resulting algorithm provides spurious-free solutions with the expected order of convergence.
\end{abstract}
%\begin{keyword}
%Maxwell Equations\sep 
%Cell Method \sep 
%Discontinuous Galerkin \sep 
%Dual Grids \sep 
%Covariant Mapping \sep
%High-Order Finite Elements
%\end{keyword}

%\end{frontmatter}
%\linenumbers
\section{Introduction}\label{sec:intro}
James Clerk Maxwell \cite{maxwell1} showed in 1861 that the electric and magnetic fields are not separate phenomena: they instead exchange energy as their amplitudes oscillate in wave patterns, which propagate through space at the speed of light. The resulting celebrated Maxwell equations have withstood the revolutions of the modern physics' world and, to the present day, are always needed to accurately describe radio-frequency devices in industry or to explain experimental findings in electromagnetism.
A hundred years after Maxwell's original theory, later succinctly recast by Heaviside \cite{heaviside1} in the language of vector calculus, Yee showed in \cite{yee} how a Maxwell initial boundary value problem (MIBVP) in 3+1 dimensions of space and time can be solved efficiently on computers, by appropriately choosing the points where fields and their derivatives are to be approximated by finite difference equations on two staggered and uniformly spaced Cartesian--orthogonal grids. Since then Yee's algorithm has slowly become ubiquitous (see \cite{taflove,meep}), yet, a plethora of other methods has consequently also been proposed, analysed and tested to account for its various shortcomings: ineffectiveness in the case of material discontinuities which cannot be aligned with the Cartesian axes and fixed $\mathcal{O}(h^2 + \tau^2)$ order of convergence, where $h$ and $\tau$ are the discrete steps in the spatial and temporal grids, respectively.
Without any pretence of being exhaustive, we mention in this introductory section some families of approaches which try to mend these drawbacks.

There are approaches based on conforming finite elements spaces (see \cite{jin,monk} and references therein), which work on unstructured space grids and present (tangentially continuous) piecewise-polynomials vector basis-functions of arbitrary degree (mainly the ones introduced by Nedelec in \cite{nedelec1980}). Unfortunately these approaches lose the efficiency inherent in the Yee algorithm, since the system matrices\footnote{usually the mass-matrices.} which need to be inverted at every time-step are banded but not (block-)diagonal. This amenable structure can be retrieved if mass-lumping techniques are employed (e.g. \cite{white}), where basis-functions are strongly tied to inexact numerical integration rules and need to be completely re-computed (or are simply unavailable) if the order of approximation needs to be increased.
Later developments led instead to the adoption of Discontinuous Galerkin (DG) Finite Element Method (FEM) approaches, which ignore the conformity constraint on the basis-functions and use orthonormal bases (which in principle lead to spectral convergence rates) compactly supported inside each finite element in the spatial discretisation of the domain. This choice, of course, destroys the geometry of the continuous Maxwell system, introducing spurious numerical solutions which do not converge to physical ones as the mesh size $h$ tends to zero, and the presence of which can be easily detected by applying the same discretisation method for solving the Maxwell eigenvalue problem (MEP) instead of the MIBVP \cite{spurious_modes}. Counter-measures can be taken, in the form of penalization terms for the tangential jumps in the approximated solutions: for example, using up-wind fluxes (as in \cite{warburton_jcp}) eliminates spurious solutions by introducing numerical energy dissipation in the scheme, which fact can become unacceptable when long-time behaviours of electromagnetic systems have to be studied. On the other hand, symmetric-interior-penalty (SIP) schemes (see \cite{grote,christophs,huang_sip}) preserve the hyperbolic nature of the system by introducing more unknowns which live on the skeleton of the mesh and do not approximate any physical quantity. Furthermore, a positive definite scalar penalty parameter, which must be tuned by the algorithm's user in accordance with $h$ and the maximum polynomial degree in the chosen bases, must also be inserted in the formulation.

There is a third class of mutually related methods which mimic more closely Yee's original algorithm: the Finite Integration Technique of \cite{weiland,matsuo}, which recasts the equations in integral form to apply the Yee algorithm to general staggered cuboidal elements but does not improve the accuracy of the original method otherwise (although we note that higher order versions of the method restricted to Cartesian-orthogonal grids do exist, e.g. \cite{chung}), the cell method (CM) of \cite{tonti,marrone,codecasa_politi,pigitd,dgatap}, which is also developed on two spatial grids in the more general setting of unstructured meshes, where a dual mesh is obtained either by barycentric subdivision (a procedure we will review in the present contribution) or by the circumcentric one of the primal mesh. These methods can be theoretically studied in a wider framework (see also \cite{Stern2015,teixeira_lattice}) of approaches particularly fitting for Maxwell's equations (since they encode the so-called De Rham complex), in which differential operators are discretised using only topological information about the input mesh and all the metric information is instead encapsulated in the mass-matrix (which is in this context much rather seen as a discrete Hodge-star operator, e.g. \cite{hiptmair2001, kettunen_hodge,bossavit_yee,auchmann}). The structure-preserving nature of these methods comes at the price of not being able to extend their convergence order to asymptotics steeper than $\mathcal{O}(h)$ (or $\mathcal{O}(h^2)$ at best if strict conditions on the mesh are imposed). This elusive higher order approximation remains a much desired property, since, far from material discontinuities, solutions of the MIBVP are smooth and oscillatory.

In the present paper we are strongly inspired by this latter framework: we start from the set of basis-functions introduced by Codecasa and co-authors in \cite{codecasa_politi,dgatap}, and more recently studied in \cite{dga_as_dg}, where an equivalence between their formulation and a peculiar DG one using two barycentric-dual unstructured meshes and piecewise-constant basis-functions was proven by some of the present authors. Building on this result, we show how to extend the method to arbitrary degree in the local polynomial basis-functions. To make the present work as self-contained as possible, we use Section \ref{sec:maxwell} to review the continuous problem and the associated notation and Section \ref{sec:mesh} to review concepts related to barycentric-dual cellular complexes. In Section \ref{sec:funspaces} the abstract setting in terms of involved functional spaces for the new algorithm is introduced (which fact also gives new and valuable mathematical background for the cell method), followed by an explicit construction of the bases for finite-dimensional (arbitrary-order) approximations of the newly introduced spaces. A proof is given for the electromagnetic energy conservation property of the ensuing semi-discrete scheme.
Section \ref{sec:basis} provides some insight on the relationship between the new arbitrary-order scheme and known lowest order ones, in so far they present the same explicit splitting of topological and geometric operators. Some details on the optimization of a possible computer implementation are also given. Section \ref{sec:num} provides numerical experiments to validate the correctness and performance of the proposed method: particular focus is devoted to showing the spectral correctness of the method (which is paramount for practical high bandwidth applications). Some general remarks, open questions, and directions for future work conclude the paper in Section \ref{sec:conclusio}.
\section{The Maxwell system of equations}\label{sec:maxwell}
In two space dimensions, the most general form for the MIBVP is
\begin{align}
  & \partial_t\bm{D}\left(\bm{r},t\right) = \bm{curl}(H\left(\bm{r},t\right)) -\bm{J}(\bm{r},t), \label{eq:ampmax}\\
  & \partial_t B\left(\bm{r},t\right) = - curl(\bm{E}\left(\bm{r},t\right)), \label{eq:faraday}\\
  & div(\bm{D}(\bm{r},t)) = \rho_c (\bm{r},t), \label{eq:egauss}\\
  & div(B(\bm{r},t)) = 0, \label{eq:mgauss}
\end{align}
\noindent to be solved $\forall t \!\in\! [0,+\infty[$ and for all $\bm{r}(x,y)$ in the bounded domain $\Omega \subset \mathbb{R}^{2}$. The fields $\bm{E}(\bm{r},t)$ and $\bm{D}(\bm{r},t)$ go by the names of electric field and electric displacement field, respectively, while the fields $H(\bm{r},t)$ and $B(\bm{r},t)$ are called magnetic field and magnetic induction field. Fields $\bm{J}(\bm{r},t)$ (the convective electric current) and $\rho_c(\bm{r},t)$ (the free electric charge) are source-terms which cause the dynamics of electromagnetic fields, i.e. they are the true right-hand side (r.h.s.) in the system of partial differential equations.

Since we set ourselves in the $\mathbb{R}^2$ ambient space, we denote only some of the unknown fields in bold-face: $\bm{E}(\bm{r},t)$ is in fact a (polar) vector field living in the Cartesian plane, while $H(\bm{r},t)$ is a pseudo-vector aligned with the $z$-axis: the true vector field would live in $\mathbb{R}^3$, with the condition $\bm{H}(\bm{r},t) = \left(0,\;0,\;H(\bm{r},t)\right)^{\mathrm{T}}$ (where the $(\cdot)^{\mathrm{T}}$ superscript denotes vector or matrix transposition).
In the applied jargon of microwave engineers this is the so-called \emph{Transverse-Magnetic} (TM) field.
We have accordingly used the appropriate $curl$ and $div$ (for divergence) operators for any vector field $\bm{v}(\bm{r},t) = ( v_x (\bm{r},t),\; v_y (\bm{r},t) )^{\mathrm{T}}$, defined (in Cartesian coordinates) as
\begin{align*}
  & curl(\bm{v}(\bm{r},t)) = \partial_x v_y(\bm{r},t)-\partial_y v_x(\bm{r},t),\\
  & div(\bm{v}(\bm{r},t)) = \partial_x v_x(\bm{r},t) + \partial_y v_y(\bm{r},t),
\end{align*}
\noindent as well as the $\bm{curl}$ and $div$ operators for any pseudo-vector $u(\bm{r},t)$, defined as
\begin{align*}
  & \bm{curl}(u(\bm{r},t)) = \left(\partial_y u(\bm{r},t), \; -\partial_x u (\bm{r},t)\right)^\mathrm{T},\\
  & div(u(\bm{r},t)) = \partial_z u(\bm{r},t) = 0,
\end{align*}
\noindent all valid for suitably differentiable components of $\bm{v}$, $u$. We will also make use of the identities
\begin{align}
  & curl\left( u\bm{v}\right) = curl(\bm{v})\,u + \bm{curl}(u)\cdot\bm{v}, \label{eq:green_id}\\
  & \int_\Omega curl\left(\bm{v}\right)\,\mathrm{d}\bm{r} =  \oint_{\partial\Omega} \bm{v}\cdot\hat{\bm{t}}(\ell)\,\mathrm{d}\ell, \label{eq:green_th}
\end{align}
\noindent namely the product rule for partial derivatives and the Green theorem, again valid for suitably differentiable functions. 
The notation $\hat{\bm{t}}(\ell)$ will denote the tangential unit vector on a directed curve (for example the boundary of $\Omega$, denoted $\partial\Omega$) for which $\ell$ is the arc-length parameter.
Furthermore, we remark that the tangent unit vector is taken to always induce a counter-clockwise circulation on contours in accordance with the well-known cork-screw rule. 
One can promptly argue that equations (\ref{eq:egauss})--(\ref{eq:mgauss}) are not dynamical constraints but rather initial conditions. It is easy to see that, if (\ref{eq:egauss})--(\ref{eq:mgauss}) hold true for $t=0$, then they are satisfied for any $t$ with $0< t < +\infty$: it suffices taking the divergence of both sides in the remaining two equations and integrating them with respect to time from zero to the chosen instant. We are therefore left with two equations and four unknowns: to make the system meaningful again, (\ref{eq:ampmax})--(\ref{eq:faraday}) must be supplemented with the phenomenological\footnote{experimentally determined.} constitutive equations 
\begin{align}
  & \bm{D}(\bm{r},t) = \varepsilon(\bm{r},t) \bm{E}\left(\bm{r},t\right),\label{eq:const_DE}\\
  & B(\bm{r},t) = \mu(\bm{r},t) H(\bm{r},t),\label{eq:const_BH}
\end{align}
\noindent where $\varepsilon = \varepsilon_0\varepsilon_r$, $\mu = \mu_0\mu_r$ are respectively called dielectric permittivity and magnetic permeability, with $\mu_0$ and $\varepsilon_0$ also being experimental constants and $c_0 = \left(\mu_0\varepsilon_0\right)^{-\frac{1}{2}}$ being the speed of light (i.e. the wave-speed of electromagnetic radiation) in a vacuum.
 
We will now make some mildly restrictive assumptions: we consider, in all that follows, time-invariant materials (for which generalization to general dispersive ones is, as for all numerical methods, more involved and will be the object of future studies). We further assume the material parameters to be symmetric positive-definite (s.p.d.) tensors (of rank two for $\varepsilon$, rank one for $\mu$) with piecewise-smooth and point-wise bounded (in space) real coefficients.
Only for simplicity of presentation, we also consider the source-free equations, i.e. $\bm{J}=\bm{0}$, $\rho_c = 0$ (where generalization of the analysis to problems with sources is straightforward and will be employed in the numerical experiments in Section \ref{sec:num}). 
We finally assume the spatial domain $\Omega$ to be a bounded polygon and allow homogeneous Dirichlet boundary conditions on either field: $\bm{E}(\bm{r},t)\cdot\hat{\bm{t}}(\ell) = \boldsymbol{0}$, $\forall \bm{r}\!\in\!\partial\Omega$, i.e. perfect electric conductor (PEC) boundary conditions in the applied jargon, or $H(\bm{r},t) = 0$, $\forall \bm{r}\!\in\!\partial\Omega$, i.e. perfect magnetic conductor (PMC). It is easy to deduce from the system of equations that Dirichlet boundary conditions for any of the two fields imply Neumann ones for the remaining unknown, and vice-versa.

\section{Barycentric-dual complexes}\label{sec:mesh}
Having as a goal the numerical solution of (\ref{eq:ampmax})--(\ref{eq:faraday}), we assume a conforming (see \ref{sec:app1}) partition of $\Omega$ into triangles (a triangular mesh) to be available, which can be easily provided from any black-box mesher (e.g. \cite{netgen,distmesh}).
Rigorously speaking, said partition is a particular kind of simplicial complex. We define a simplicial complex for $\Omega$, denoted $\mathcal{C}^\Omega$, as a sequence of sets of simplexes in various dimensions
$$\mathcal{C}^\Omega = \{\mathcal{C}_k^\Omega\}_{k=0,1,\dots,d},$$
\noindent where $d=2$ is the ambient space dimension. Keeping in mind that a $k$--simplex is the convex hull of $k\!+\!1$ affinely independent points, $\mathcal{C}_2^{\Omega}$ will denote the set of triangles (2-simplexes), $\mathcal{C}_1^\Omega$ will denote the set of edges (1-simplexes), while $\mathcal{C}_0^\Omega$ will denote the set of vertices (0-simplexes) in the mesh. We also define the skeleton of a complex:
\begin{align}
  & \mathcal{S}\left(\mathcal{C}^\Omega\right) = \bigcup_{i=0}^{i=d-1} \mathcal{C}_k^\Omega , \label{eq:skeleton_primal}
\end{align}
\noindent i.e. the set of all simplexes of dimension smaller than the maximal one.
Furthermore, we assume that the vertices in $\mathcal{C}_0^\Omega$ can be ordered by virtue of an index set $\mathcal{I}$, $i\!\in\!\mathbb{N}^+, \forall i\in \mathcal{I}$. Consequently all edges in $\mathcal{C}_1^\Omega$ possess a global (in $\mathcal{C}^\Omega$) inner orientation induced by the ordering of vertices in their boundary.

A generic simplicial complex is itself a particular type of \emph{cellular} (or cell) complex, which is the more general structure one gets if they relax the requirement on geometric entities of $\mathcal{C}^\Omega$ from being simplexes to, for example, being generic polytopes (called $k$-cells instead of $k$-simplexes). Our starting mesh, as any given simplicial complex, possesses a dual complex, which we denote (in 2D) with $\tilde{\mathcal{C}}^\Omega = \{\tilde{\mathcal{C}}_0^\Omega,\tilde{\mathcal{C}}_1^\Omega,\tilde{\mathcal{C}}_2^\Omega\}$ and which is indeed a cellular complex \emph{but not a simplicial one}. The existence of a dual cellular complex hinges on a sequence of one-to-one mappings $\{D_k\}_{k=0,1,2}$ such that
\begin{align}
  &D_k: \mathcal{C}_k^\Omega \mapsto \tilde{\mathcal{C}}_{d-k}^\Omega.\label{eq:duality_mapping}
\end{align}

This mathematical concept originally arose in solutions of algebraic-topological problems \cite{munkres}, and the geometric realization (which is non-unique) of such a dual cellular complex is very often outside of the computational needs of topologists. On the contrary, for what follows, it is a fundamental choice to construct $\tilde{\mathcal{C}}_\Omega$ via the barycentric subdivision of $\mathcal{C}^\Omega$: each vertex $\tilde{\mathbf{v}} \in \tilde{\mathcal{C}}_0^\Omega$ is the centroid of some $\mathcal{T}\in\mathcal{C}_2^\Omega$, each $\tilde{e} \in \tilde{\mathcal{C}}_1^\Omega$ is a polyline obtained by joining the centroid of some $E\in\mathcal{C}_1^\Omega$ to the centroids of neighbouring triangles, while each $\tilde{\mathcal{T}} \in \tilde{\mathcal{C}}_2^\Omega$ is a (generally non-convex) polygon bounded by dual edges (elements of $\tilde{\mathcal{C}}_1^\Omega$) and containing exactly one vertex $\mathbf{v}\in\mathcal{C}_0^\Omega$. A depiction of one simplicial complex and its barycentric-dual companion is given in the two first leftmost panels of Fig.~\ref{fig:staggered_grids}, while the whole formal procedure is more thoroughly described in \ref{sec:app1}.

\begin{figure}[!h]
    \centering
    \begin{minipage}{0.3\textwidth}
        \centering
        \begin{tikzpicture}[thick,scale=3.5, every node/.style={scale=3.5}]
        \draw (0.0,0.0) -- (0.5,0.0) -- (0.5,0.5) -- (0.0,0.0);
        \draw[fill=gray] (0.0,0.0) -- (0.5,0.5) -- (0.0,0.5) -- (0.0,0.0);
        \draw (0.0,0.5) -- (0.5,0.5) -- (0.0,1.0) -- (0.0,0.5);
        \draw (0.5,0.5) -- (0.0,1.0) -- (0.5,1.0) -- (0.5,0.5);
        \draw (0.5,0.5) -- (0.5,0.0) -- (1.0,0.0) -- (0.5,0.5);
        \draw (0.5,0.5) -- (1.0,0.0) -- (1.0,0.5) -- (0.5,0.5);
        \draw (0.5,0.5) -- (1.0,0.5) -- (1.0,1.0) -- (0.5,0.5);
        \draw (0.5,0.5) -- (1.0,1.0) -- (0.5,1.0) -- (0.5,0.5);
        \draw[color=red] (1/2,1/2) -- (1.0,1/2);
        \node[scale=0.3] at (1/2+0.05,1/2+0.1) {$\mathbf{v}$};
        \node[scale=0.3,color=red,thick] at (3/4,1/2+0.05) {$E$};
        \node[scale=0.3] at (1/6,1/3) {$\mathcal{T}$};
        \node[circle,fill,scale=0.05] at (1/2,1/2) {};
        \end{tikzpicture}\end{minipage}
    %\hfill
    \begin{minipage}{0.3\textwidth}
        \centering
        \begin{tikzpicture}[thick,scale=3.5, every node/.style={scale=3.5}]
        \draw[color=black,fill=gray] (1/2,1/4) -- (2/3,1/6) -- (3/4,1/4) -- (5/6,1/3) -- (3/4,1/2)
        -- (5/6,2/3) -- (3/4,3/4) -- (2/3,5/6) -- (1/2,3/4) -- (1/3,5/6) -- (1/4,3/4)
        -- (1/6,2/3) -- (1/4,1/2) -- (1/6,1/3) -- (1/4,1/4) -- (1/3,1/6) -- (1/2,1/4);
        \draw[dashed] (0.0,0.0) -- (0.5,0.0) -- (0.5,0.5) -- (0.0,0.0);
        \draw[dashed] (0.0,0.0) -- (0.5,0.5) -- (0.0,0.5) -- (0.0,0.0);
        \draw[dashed] (0.0,0.5) -- (0.5,0.5) -- (0.0,1.0) -- (0.0,0.5);
        \draw[dashed] (0.5,0.5) -- (0.0,1.0) -- (0.5,1.0) -- (0.5,0.5);
        \draw[dashed] (0.5,0.5) -- (0.5,0.0) -- (1.0,0.0) -- (0.5,0.5);
        \draw[dashed] (0.5,0.5) -- (1.0,0.0) -- (1.0,0.5) -- (0.5,0.5);
        \draw[dashed] (0.5,0.5) -- (1.0,0.5) -- (1.0,1.0) -- (0.5,0.5);
        \draw[dashed] (0.5,0.5) -- (1.0,1.0) -- (0.5,1.0) -- (0.5,0.5);
        \draw[color=black] (1/4,0.0) -- (1/3,1/6);
        \draw[color=black] (3/4,0.0) -- (2/3,1/6);
        \draw[color=black] (1.0,1/4) -- (5/6,1/3);
        \draw[color=black] (1.0,3/4) -- (5/6,2/3);
        \draw[color=black] (3/4,1.0) -- (2/3,5/6);
        \draw[color=black] (1/4,1.0) -- (1/3,5/6);
        \draw[color=black] (0.0,3/4) -- (1/6,2/3);
        \draw[color=black] (0.0,1/4) -- (1/6,1/3);
        \draw[color=red] (5/6,1/3) -- (3/4,1/2) -- (5/6,2/3);
        \node[scale=0.3] at (1/2+0.05,1/2+0.1) {$\tilde{\mathcal{T}}$};
        \node[scale=0.3,thick,color=red] at (5/6,1/2+0.07) {$\tilde{E}$};
        \node[circle,fill,scale=0.05] at (1/6,1/3) {};
        \node[scale=0.3] at (1/6,1/4) {$\tilde{\mathbf{v}}$};
        \end{tikzpicture}\end{minipage}
    %\hfill
    \begin{minipage}{0.3\textwidth}
        \centering
        \begin{tikzpicture}[thick,scale=3.5, every node/.style={scale=3.5}]
        \draw[color=black] (1/2,1/4) -- (2/3,1/6) -- (3/4,1/4) -- (5/6,1/3) -- (3/4,1/2)
        -- (5/6,2/3) -- (3/4,3/4) -- (2/3,5/6) -- (1/2,3/4) -- (1/3,5/6) -- (1/4,3/4)
        -- (1/6,2/3) -- (1/4,1/2) -- (1/6,1/3) -- (1/4,1/4) -- (1/3,1/6) -- (1/2,1/4);
        \draw[color=black] (0.0,0.0) -- (0.5,0.0) -- (0.5,0.5) -- (0.0,0.0);
        \draw[color=black] (0.0,0.0) -- (0.5,0.5) -- (0.0,0.5) -- (0.0,0.0);
        \draw[color=black] (0.0,0.5) -- (0.5,0.5) -- (0.0,1.0) -- (0.0,0.5);
        \draw[color=black] (0.5,0.5) -- (0.0,1.0) -- (0.5,1.0) -- (0.5,0.5);
        \draw[color=black] (0.5,0.5) -- (0.5,0.0) -- (1.0,0.0) -- (0.5,0.5);
        \draw[color=black] (0.5,0.5) -- (1.0,0.0) -- (1.0,0.5) -- (0.5,0.5);
        \draw[color=black] (0.5,0.5) -- (1.0,0.5) -- (1.0,1.0) -- (0.5,0.5);
        \draw[color=black] (0.5,0.5) -- (1.0,1.0) -- (0.5,1.0) -- (0.5,0.5);
        \draw[color=black] (1/4,0.0) -- (1/3,1/6);
        \draw[color=black] (3/4,0.0) -- (2/3,1/6);
        \draw[color=black] (1.0,1/4) -- (5/6,1/3);
        \draw[color=black] (1.0,3/4) -- (5/6,2/3);
        \draw[color=black] (3/4,1.0) -- (2/3,5/6);
        \draw[color=black] (1/4,1.0) -- (1/3,5/6);
        \draw[color=black] (0.0,3/4) -- (1/6,2/3);
        \draw[color=black] (0.0,1/4) -- (1/6,1/3);
        \draw[dashed,fill=gray] (1/2,0.0) -- (3/4,0.0) -- (2/3,1/6) -- (1/2,1/4) -- (1/2,0.0);
        \draw[color=red] (3/4,0.0) -- (2/3,1/6);
        \draw[color=green] (1/2,1/4) -- (1/2,0.0);
        \node[scale=0.3] at (3/5,1/9) {$K$};
        \node[scale=0.3,thick,color=red] at (3/4+0.05,1/12) {$\tilde{e}$};
        \node[scale=0.3,thick,color=green] at (1/2-0.05,1/8) {${e}$};
        \end{tikzpicture}\end{minipage}
    \caption{The primal and dual complex: a glossary. On the left we mesh the unit square $\Omega=[0,1]\times[0,1]$ with the simplicial complex $\mathcal{C}_\Omega$ and we show $\mathbf{v}\in\mathcal{C}_0^\Omega$, $E\in\mathcal{C}_1^\Omega$, $\mathcal{T}\in\mathcal{C}_2^\Omega$. In the middle, where the primal complex is shown dashed, we have constructed the barycentric-dual complex: $\tilde{\mathbf{v}}\in\tilde{\mathcal{C}}_0^\Omega$ is dual to $\mathcal{T}$, $\tilde{E}\in\tilde{\mathcal{C}}_1^\Omega$ is dual to $E$, $\tilde{\mathcal{T}}\in\tilde{\mathcal{C}}_2^\Omega$ is dual to $\mathbf{v}$. On the right, we finally draw the resulting auxiliary complex $\mathcal{K}^\Omega$ and also emphasize a quadrilateral $K\in\mathcal{K}_2^\Omega$, an edge ${e}\in\{\mathcal{K}_1^\Omega\cap\mathcal{S}(\mathcal{C}^\Omega)\}$, and an edge $\tilde{e} \in \{\mathcal{K}_1^\Omega \cap \mathcal{S}( \tilde{\mathcal{C}}^\Omega )\}$.}
    \label{fig:staggered_grids}
\end{figure}
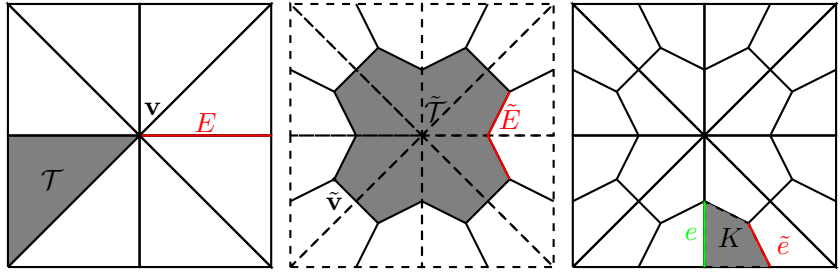

We will also need, for what lies ahead in the paper, to define an additional complex $\mathcal{K}^\Omega =\{\mathcal{K}_k^\Omega\}_{k=0}^2$, where
\begin{align}
  & \mathcal{K}_2^\Omega = \left\{ \emptyset \neq K=\mathcal{T}\cap\tilde{\mathcal{T}}, \;\; \forall\mathcal{T}\!\in\!\mathcal{C}_2^\Omega, \forall\tilde{\mathcal{T}}\!\in\!\tilde{\mathcal{C}}_2^\Omega  \right\},\label{eq:kite_complex}
\end{align}
\noindent and where we note that each $d$-dimensional simplex of the original mesh (and hence the whole of $\Omega$) is thus further partitioned into $d$+1 disjoint subsets $K\!\in\!\mathcal{K}_2^\Omega$ (see again Fig.~\ref{fig:staggered_grids}, rightmost panel). For any original triangle, we get three irregular quadrilaterals, which will be of utmost importance, and we will call \emph{fundamental 2-cells} (see the definition of \emph{micro-cell} in \cite{marrone} or see \cite{dga_as_dg}) and denote with $K$ in the rest of the paper. Definitions of lower dimensional sets $\mathcal{K}_1^\Omega$ and $\mathcal{K}_0^\Omega$ are intuitive, but we additionally provide here an explicit decomposition of $\mathcal{K}_1^\Omega$ into the two sets of segments
\begin{align*}
  & \mathcal{K}_1^\Omega\cap\mathcal{S}(\mathcal{C}^\Omega) = \{ \emptyset \neq {e} = E \cap \partial K \; s.t. \; E \in \mathcal{C}_1^\Omega, K\in\mathcal{K}_2^\Omega\},\\
  & \mathcal{K}_1^\Omega\cap\mathcal{S}(\tilde{\mathcal{C}}^\Omega) = \{ \emptyset \neq \tilde{{e}} =
                             \tilde{E} \cap \partial K \; s.t. \; \tilde{E} \in \tilde{\mathcal{C}}_1^\Omega, K\in\mathcal{K}_2^\Omega\},
\end{align*}
\noindent of which $\mathcal{K}_1^\Omega$ is the disjoint union. We furthermore note that every $K\in\mathcal{K}_2^\Omega$ is uniquely identified by a triangle--vertex pair $(\mathcal{T},\mathbf{v})$, for some triangle $\mathcal{T}\!\in\!\mathcal{C}_2^\Omega$ and some mesh vertex $\mathbf{v} \!\in\! \{\mathcal{C}_0^\Omega \cap \partial\mathcal{T}\}$.
Having constructed the appropriate dual complex we can refer hereinafter to the starting simplicial one as the primal complex.

We mention in passing the circumcentric-dual (see \cite{tonti}) as another popular construction employed in the literature, which has amenable properties for finite volumes schemes (see \cite{leveque_2002} and references therein), but requires the triangulation of $\Omega$ to be a Delaunay one, which is usually too restrictive or simply not satisfied by the meshing algorithm at hand.

We conclude the section by remarking that there is an equivalent definition of skeleton $\mathcal{S}^k(\tilde{\mathcal{C}}^\Omega)$ for the dual complex and noting that in the following the notation $|\mathcal{C}_k^\Omega|$ and $|\tilde{\mathcal{C}}_k^\Omega|$ will, as customary, denote the size of the argument set (e.g. the number of triangles in the primal complex is $|\mathcal{C}_2^\Omega|$, the number of edges in the primal complex is $|\mathcal{C}_1^\Omega|$, etc.).

%The arrows symbolize tangent vector components which must be continuous across interfaces (green for the electric field, blue for the magnetic field), making Eqs.(?) well defined.
\section{The new formulation: continuous and discrete}\label{sec:funspaces}
In the present section we will turn our attention to functions supported on these complexes and use ingredients from the theory of Sobolev spaces to develop a mathematical background for our method. We will thus finally jump back to the Maxwell system we want to solve and make use of all the machinery.

\subsection{Barycentric-dual discontinuous functional spaces}
For any bounded $\mathcal{D}\subset\mathbb{R}^2$, we recall the usual real Hilbert spaces:
\begin{align*}
  & L^2(\mathcal{D})  = \left\{ f:\mathcal{D}\mapsto\mathbb{R} \; s.t. \; \int_\mathcal{D} \vert f \vert^2\,\mathrm{d}\bm{r}  < +\infty \right\},\\
  & \bm{L}^2(\mathcal{D}) = \bigm\{ \bm{v} = \left(f,\;g\right)^{\mathrm{T}} : \mathcal{D}\mapsto\mathbb{R}^2 \; s.t. \; f,g \!\in\! L^2(\mathcal{D})\bigm\},
  %& H^1(\mathcal{D})  = \left\{ f\!\in\! L^2(\mathcal{D}) \; s.t. \; \nabla f \!\in\! \bm{L}^2(\mathcal{D}) \right\},\\
\end{align*}
\noindent from which we infer the standard inner product and its induced norm:
\begin{align*}
  & \left( f,g \right)_\mathcal{D} := \int_\mathcal{D} fg\,\mathrm{d}\bm{r},\;\;\;
    \Norm{f}[\mathcal{D}] = \left(\int_\mathcal{D} \vert f \vert^2\,\mathrm{d}\bm{r}\right)^{\frac{1}{2}} =  \left( f,f \right)_\mathcal{D}^{\frac{1}{2}}, 
\end{align*}
\noindent for all $f, g \!\in\! L^2(\mathcal{D})$. The following inner product and norm are also implied from the definition of $\bm{L}^2(\mathcal{D})$:
\begin{align*}
  & \left( \bm{v},\bm{w} \right)_\mathcal{D} := \int_\mathcal{D} \bm{v}\cdot\bm{w}\,\mathrm{d}\bm{r},\;\;\;
    \Norm{\bm{v}}[\mathcal{D}] = \left(\int_\mathcal{D} \vert \bm{v} \vert^2\,\mathrm{d}\bm{r}\right)^{\frac{1}{2}} = \left( \bm{v},\bm{v} \right)_\mathcal{D}^{\frac{1}{2}},
\end{align*}
\noindent for all $\bm{v}, \bm{w} \!\in\! \bm{L}^2(\mathcal{D})$.
To properly define the electromagnetic energy for a generic computational domain, we will often need a weighted norm which includes $\varepsilon$ and $\mu$, which we will denote by adding the appropriate symbol to the subscripts of standard $L^2$ or $\bm{L}^2$ norms: 
\begin{align*}
  & \Norm{\bm{v}}[\mathcal{D},\varepsilon] =
  \left( \varepsilon\bm{v},\bm{v}\right)_\mathcal{D}^{\frac{1}{2}},\;\;
    \Norm{u}[\mathcal{D},\mu] =
  \left( \mu u,u\right)_\mathcal{D}^{\frac{1}{2}},
\end{align*}
\noindent which are well-defined by virtue of the s.p.d. assumption on the material tensors. We finally introduce the following real Sobolev spaces
\begin{align*}
  & {\bm{H}}^{curl} (\mathcal{D}) = \left\{ \bm{v} \!\in\! \bm{L}^2(\mathcal{D}) \; s.t. \;
  curl(\bm{v}) \!\in\! L^2(\mathcal{D})\right\},\\
%  & {\bm{H}}_{0}^{curl} (\mathcal{D}) = \left\{ \bm{v} \!\in\! \bm{H}^{curl}(\mathcal{D}) \; s.t. \;
%     \bm{v}\cdot\hat{\bm{t}}(\ell) {\mid}_{\partial\mathcal{D}} = 0 \right\},\\
  & H^{\bm{curl}}(\mathcal{D})  = \left\{ u \!\in\! L^2(\mathcal{D}) \; s.t. \; \bm{curl}(u)\!\in\! \bm{L}^2(\mathcal{D})\right\},
\end{align*}
\noindent where all derivatives are now taken in the distributional sense. Since we assume that time and space are separable, the semi-weak solutions of the Maxwell system live in function spaces which are well-established in the literature:
\begin{align*}
  & \bm{E}(\bm{r},t) \in AC\!\left([0,\mathrm{T}]\right)\otimes \bm{H}^{curl}(\Omega),\\
  & {H}(\bm{r},t) \in AC\!\left([0,\mathrm{T}]\right) \otimes H^{\bm{curl}}(\Omega),
\end{align*}
for end-time $t=\mathrm{T}$ s.t. $0<\mathrm{T}<+\infty$, where $AC\left([0,\mathrm{T}]\right)$ denotes the space of absolutely continuous functions on $[0,\mathrm{T}]$. We keep the differentiability condition in the strong sense in the time variable, since we will be discretising it with finite differences (as in the Yee algorithm), and we postpone the inclusion of boundary conditions to a later point in the paper.
%We deviate now from theory in the sense that we seek solutions for the unknown fields in two slightly bigger (therefore non-conforming) spaces
%\begin{align*}
%  & \bm{E}(\bm{r},t) \in AC\!\left([0,\mathrm{T}]\right) \cap \bm{H}_0^{curl}(\tilde{\mathcal{C}}_2^\Omega),\\
%  & {H}(\bm{r},t) \in AC\!\left([0,\mathrm{T}]\right) \cap H^{\bm{curl}}(\mathcal{C}_2^\Omega),
%\end{align*}
%\noindent where we have used the same symbols for the solutions without risk of ambiguity and,
We define now, with reference to the complexes introduced in Section \ref{sec:mesh}, the new \emph{broken} Sobolev spaces
\begin{align}
  \bm{H}^{curl}(\tilde{\mathcal{C}}_2^\Omega) &= \left\{ \bm{v} \!\in\! \bm{L}^2(\Omega) \; s.t. \;
  \bm{v}|_{\tilde{\mathcal{T}}} \!\in\! \bm{H}^{curl}(\tilde{\mathcal{T}}),\,
  \forall\tilde{\mathcal{T}} \!\in\! \tilde{\mathcal{C}}_2^\Omega %\;\wedge\; \bm{v}\cdot\hat{\bm{t}}(\ell){\mid}_{\partial\Omega} = 0
  \right\},\label{eq:broken_vector_valued}\\
  H^{\bm{curl}}(\mathcal{C}_2^\Omega) &= \left\{ u \!\in\! L^2(\Omega) \; s.t. \;
  u|_\mathcal{T} \!\in\! H^{\bm{curl}}(\mathcal{T}),\,
  \forall \mathcal{T} \!\in\! \mathcal{C}_2^\Omega\right\}. \label{eq:broken_scalar_valued}
\end{align}

Informally speaking, these are locally conforming spaces which are globally non-conforming on $\Omega$, yet the non-conformity has a different support for the two spaces. Our next step is now to apply \emph{local} testing in space to equations (\ref{eq:ampmax}) and (\ref{eq:faraday}) with respect to the new broken spaces, that is
\begin{align}
& \sum_{\tilde{\mathcal{T}}\in\tilde{\mathcal{C}}_2^\Omega}
\left( \varepsilon \partial_t\bm{E},\bm{v}\right)_{\tilde{\mathcal{T}}} =
\sum_{\tilde{\mathcal{T}}\in\tilde{\mathcal{C}}_2^\Omega} \left( \bm{curl}({H}), \bm{v}\right)_{\tilde{\mathcal{T}}},
\;\; & \forall \bm{v}\in \bm{H}^{curl}(\tilde{\mathcal{C}}_2^\Omega),\label{eq:local_am}\\
& \sum_{\mathcal{T}\in\mathcal{C}_2^\Omega}
\left( \mu \partial_t{H},u \right)_\mathcal{T} =
-\sum_{\mathcal{T}\in\mathcal{C}_2^\Omega} \left( curl(\bm{E}),u \right)_\mathcal{T},
\;\; &\forall u\in H^{\bm{curl}}(\mathcal{C}_2^\Omega),\label{eq:local_fa}
\end{align}
where the constitutive equations (\ref{eq:const_DE})--(\ref{eq:const_BH}) have been used and we stress the different local integration domains $\mathcal{T}$ and $\tilde{\mathcal{T}}$. The interplay of the two dual complexes can be exploited by making the r.h.s. of (\ref{eq:local_am})--(\ref{eq:local_fa}) ultra-weak, i.e. performing the following formal integration by parts
\begin{align*}
& \sum_{\tilde{\mathcal{T}}\in\tilde{\mathcal{C}}_2^\Omega}
\left( \varepsilon \partial_t\bm{E},\bm{v}\right)_{\tilde{\mathcal{T}}} =
\sum_{\tilde{\mathcal{T}}\in\tilde{\mathcal{C}}_2^\Omega} \left(
\int_{\partial\tilde{\mathcal{T}}} \hspace{-2.5mm} H\bm{v}\cdot\hat{\bm{t}}(\ell)\,\mathrm{d}\ell
-\left( {H}, curl\left(\bm{v}\right)\right)_{\tilde{\mathcal{T}}} \right),
\;\; & \forall \bm{v}\in \bm{H}_{0}^{curl}(\tilde{\mathcal{C}}_2^\Omega),\\
& \sum_{\mathcal{T}\in\mathcal{C}_2^\Omega}
\left( \mu \partial_t{H},u \right)_\mathcal{T} =
\sum_{\mathcal{T}\in\mathcal{C}_2^\Omega} \left(
\int_{\partial\mathcal{T}} \hspace{-2.5mm}u\bm{E}\cdot\hat{\bm{t}}(\ell)\,\mathrm{d}\ell
-\left( \bm{E},\bm{curl}(u) \right)_\mathcal{T}\right),
\;\; &\forall u\in H^{\bm{curl}}(\mathcal{C}_2^\Omega),
\end{align*}
\noindent where boundary terms arise from the tangential discontinuity of test-functions. This latter step proves to be a crucial part of the novel derivation: as we will show in the following subsection, the tangential traces of solutions appearing in the line-integral terms will all remain single-valued even when the trial-spaces for $\bm{E}$ and $H$ in our Galerkin approximation will be broken in the same manner as the test-spaces.

\subsection{Finite-dimensional approximation}
%We now stray away from the somewhat abstract setting and devote ourselves to constructively showing the existence of basis-functions for finite-dimensional approximations of the broken spaces in (\ref{eq:broken_vector_valued})--(\ref{eq:broken_scalar_valued}).
In non-conforming DG methods, once a mesh is available, the equations are independently tested on each triangle against some polynomial basis (or some other kind of locally smooth functions, if the DG FEM is combined with spectral or Trefftz approaches, e.g. \cite{egger}). This gives birth, once the solution is approximated within the same finite-dimensional basis, to block-diagonal (hence easily invertible) mass-matrices on the left-hand side (l.h.s) of the weak formulation of (\ref{eq:ampmax})--(\ref{eq:faraday}). We can here generate a similar block-diagonal structure by virtue of the two newly defined broken spaces. This will be done by using basis-functions of finite-dimensional subspaces for $\bm{H}^{curl}(\tilde{\mathcal{C}}_2^\Omega)$ and $H^{\bm{curl}}(\mathcal{C}_2^\Omega)$ with compact support limited to some $\tilde{\mathcal{T}}\in\tilde{\mathcal{C}}_2^\Omega$ and $\mathcal{T}\in\mathcal{C}_2^\Omega$, respectively.
The basis-functions will be as usual piecewise--polynomial (vectors) up to some fixed degree $p\geq0$.
Nevertheless, the procedure is far from equivalent to existing literature, since we have decided to use two different partitions of $\Omega$ which overlap and must be forced to exchange information. This is not a drawback, since it allows us to avoid introducing numerical fluxes (and handle all their consequences), as instead common in all popular DG approaches.

Let us start by defining local Cartesian--orthogonal coordinates $(\xi_1,\xi_2)$ and denote with $\hat{\mathcal{T}}$ the reference (or master) triangle, i.e. the convex hull of the point-set $\{(0,0)^\mathrm{T},(1,0)^\mathrm{T},(0,1)^\mathrm{T}\}$ in the given coordinates. We will denote with $\hat{\bm{r}} = \hat{\bm{r}}(\xi_1,\xi_2)$ position vectors on $\hat{\mathcal{T}}$.
This is a standard domain for FEM practitioners, as the usual procedure consists in defining local ``shape-functions'' on $\hat{\mathcal{T}}$ and subsequently using a family of continuous and invertible mappings $\varphi_\mathcal{T}$, which map $\hat{\mathcal{T}}$ to each physical triangle $\mathcal{T}\in\mathcal{C}_2^\Omega$, to ``patch-up'' global basis-functions on the whole of $\Omega$.
However, we note that there are, for each $\mathcal{T}\in\mathcal{C}_2^\Omega$, actually three different choices for affine transformations which map vertices of $\hat{\mathcal{T}}$ to vertices of $\mathcal{T}$ (up to reversal of orientation for the triangle), and they are in the form:
\begin{align*}
%\begin{split}
&\bm{r} = \varphi_{\mathcal{T},i}(\hat{\bm{r}}) := \mathbf{A}_{\mathcal{T},i} \hat{\bm{r}} + \mathbf{b}_{\mathcal{T},i},%\\
%&= \left( \mathbf{e}_j \mathbf{e}_k \right) \hat{\bm{r}} + \mathbf{v}_i,
%\end{split}
\end{align*}
\noindent where $i\in\{1,2,3\}$, $\hat{\bm{r}}\in\hat{\mathcal{T}}$, $\bm{r}\in \mathcal{T}$, $\mathbf{A}_{\mathcal{T},i} \in \mathbb{R}^{2\!\times\!2}$, $\mathbf{b}_{\mathcal{T},i}\in\mathbb{R}^2$.  
If we take any triangle $\mathcal{T}\in\mathcal{C}_2^\Omega$, denote with $\mathbf{v}_{\mathcal{T},1}, \mathbf{v}_{\mathcal{T},2}, \mathbf{v}_{\mathcal{T},3}$ the Euclidean vectors (now in the global mesh coordinates) for the three vertices in the set $\{\partial\mathcal{T} \cap \mathcal{C}_0^\Omega\}$,
%Similarly we denote with $\mathbf{e}_{\mathcal{T},1},\mathbf{e}_{\mathcal{T},2},\mathbf{e}_{\mathcal{T},3}$ the Euclidean vectors of the three segments defining the set $\{\partial\mathcal{T} \cap \mathcal{C}_1^\Omega\}$. We can then write
%\begin{align*}
%  & \mathbf{e}_{\mathcal{T},1} = \mathbf{v}_{\mathcal{T},3} - \mathbf{v}_{\mathcal{T},2},\\
%  & \mathbf{e}_{\mathcal{T},2} = \mathbf{v}_{\mathcal{T},3} - \mathbf{v}_{\mathcal{T},1},\\
%  & \mathbf{e}_{\mathcal{T},3} = \mathbf{v}_{\mathcal{T},2} - \mathbf{v}_{\mathcal{T},1},
%\end{align*}
%without any loss of generality.
%It is straightforward to see that the three appropriate mappings are given by setting
%\begin{align*}
%  &  \mathbf{A}_{\mathcal{T},1} = \begin{pmatrix*}[r]\;\;\;\mathbf{e}_{\mathcal{T},2} & \;\;\;\mathbf{e}_{\mathcal{T},3} \end{pmatrix*},
%  & \mathbf{b}_{\mathcal{T},1} = \mathbf{v}_{\mathcal{T},1},\\
%  &  \mathbf{A}_{\mathcal{T},2} = \begin{pmatrix*}[r] -\mathbf{e}_{\mathcal{T},3} & \;\;\;\mathbf{e}_{\mathcal{T},1} \end{pmatrix*},
%  & \mathbf{b}_{\mathcal{T},2} = \mathbf{v}_{\mathcal{T},2},\\
%  & \mathbf{A}_{\mathcal{T},3} = \begin{pmatrix*}[r] -\mathbf{e}_{\mathcal{T},1} & -\mathbf{e}_{\mathcal{T},2} \end{pmatrix*},
%  & \mathbf{b}_{\mathcal{T},3} = \mathbf{v}_{\mathcal{T},3},
%\end{align*}
%\noindent where column vectors have been grouped to form $2\!\times\!2$ matrices. 
and we recall that (as already remarked) each pair $(T,\mathbf{v}_{\mathcal{T},i})$ uniquely identifies a quadrilateral $K\in\mathcal{K}_2^\Omega$, we can make the notation less cumbersome by writing $\varphi_K$ (and $\mathbf{A}_K$, $\mathbf{b}_K$ as well) instead of using two subscripts. In connection with this, the following result additionally holds:
\begin{lem}\label{thm:kmap}
  For each $\mathcal{T}\in\mathcal{C}_2^\Omega$, $\{\mathbf{v}_{\mathcal{T},i}\}_{i=1,2,3}$ (defined as above), the affine mapping $\varphi_{\mathcal{T},i} := \varphi_K$ is invertible, and the inverse $\varphi_K^{-1}$ maps $K\in\mathcal{K}_2^\Omega$ to the kite\footnote{a quadrilateral where two disjoint pairs of adjacent sides are equal.}-cell (KC), denoted with $\hat{K}$ and defined as
  \begin{align*}\hat{K}=\mathrm{Conv}\left\{\left(0,0\right)^\mathrm{T}, \left(1/2,0\right)^\mathrm{T}, \left(1/3,1/3\right)^\mathrm{T}, \left(0,1/2\right)^\mathrm{T}\right\}.\end{align*}
  \noindent where $\mathrm{Conv}\{\cdot,\dots,\cdot\}$ denotes the convex hull of its arguments.
\end{lem}\qed
%\begin{proof}
%  Invertibility is ensured by the linear independence of the columns of $\mathbf{A}_K$ matrices, which is in turn given by the non-degeneracy of triangles in $\mathcal{C}_2^\Omega$. The affinity of the mappings furthermore ensures that the centroids of sub-simplexes of $\mathcal{T}$ are mapped to centroids of sub-simplexes of $\hat{\mathcal{T}}$. Noticing that each $\varphi_{\mathcal{T},i}$ clearly maps the origin of the local coordinate system to $\mathbf{v}_{\mathcal{T},i}$ completes the proof.
%\end{proof}

Here lies in fact our biggest departure from the classical FEM approach: we work on a proper subset of the reference triangle $\hat{\mathcal{T}}$, namely $\hat{K}$. Both $\hat{\mathcal{T}}$ and $\hat{K}$ are depicted in Fig.~\ref{fig:ref_kite}.

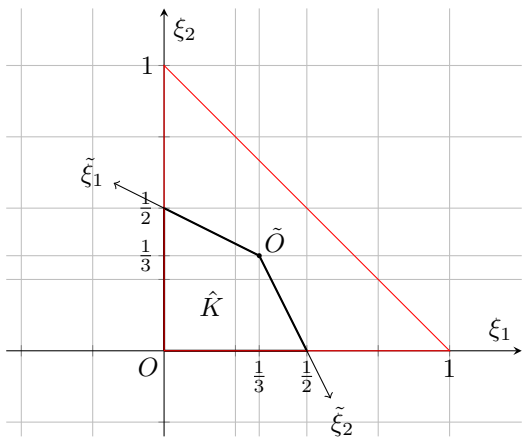
\begin{figure}[th]
\centering
\begin{tikzpicture}
  \begin{axis}[axis lines=middle,axis equal,
      grid=both,xlabel=$\xi_1$,ylabel=$\xi_2$,
      ymin=-0.3,ymax=1.2,xmax=1,xmin=-0.3,
      xtick={-1.0,-0.75,...,0.75,1.0},
      ytick={-1.0,-0.75,...,0.75,1.0},
      extra x ticks = {0.333333333333333333}, extra y ticks = {0.333333333333333333},
      xticklabel=\empty,yticklabel=\empty]
  \addplot[color=black,thick=true] coordinates{(0,0) (1/2,0) (1/3,1/3) (0,1/2) (0,0)};
  \addplot[color=red] coordinates{(0,0) (1,0) (0,1) (0,0)};
  \node[fill,circle,scale=0.2] (a) at (1/3,1/3) {};
  \node[] at (1/3+1/18,1/3+1/18) {$\tilde{O}$};
  \node[] at (-1/18,-1/18) {$O$};
  \node[] at (1/6,1/6) {$\hat{K}$};
    \node (b) at (-1/4,5/8) {$\tilde{\xi}_1$};
    \node (c) at (5/8,-1/4) {$\tilde{\xi}_2$};
    \draw[->] (a) to node {}  (b);
    \draw[->] (a) to node {}  (c);
    \node[scale=0.1,label={181:{$\frac{1}{2}$}},inner sep=2pt] at (axis cs:0,1/2) {};
    \node[scale=0.1,label={181:{$\frac{1}{3}$}},inner sep=2pt] at (axis cs:0,1/3) {};
    \node[scale=0.1,label={270:{$\frac{1}{2}$}},inner sep=2pt] at (axis cs:1/2,0) {};
    \node[scale=0.1,label={270:{$\frac{1}{3}$}},inner sep=2pt] at (axis cs:1/3,0) {};
    \node[scale=0.1,label={180:{$1$}},inner sep=2pt] at (axis cs:0,1) {};
    \node[scale=0.1,label={270:{$1$}},inner sep=2pt] at (axis cs:1,0) {};
  \end{axis}
\end{tikzpicture}
\caption{The reference kite-cell and the reference triangle, from which it is derived.} \label{fig:ref_kite}
\end{figure}

\noindent With the introduction of a reference fundamental 2-cell $\hat{K}$ we want to develop a new (semi-conforming) finite element, where we can still work with the exact same local set of coordinates $(\xi_1,\xi_2)$ and define the following shape-functions' set:
%% \begin{align}
%%   & \bm{w}_{\hat{K}}^{ijk}(\hat{\bm{r}})    = \alpha_{ij} \xi_1^i\xi_2^j\nabla\xi_k,\;\;
%%  % & \bm{Y}_{\hat{K}}^{ijk}(\hat{\bm{r}})    = \chi_{\hat{K}}\begin{pmatrix*}[c] \delta_{1,k}\xi_1^i\xi_2^j \\ \delta_{2,k}\xi_1^j \xi_2^i  \end{pmatrix*},\;\;
%%   0 \leq i+j \leq p, k \in \{1,2\},\label{eq:ref_shapefuns}
%% \end{align}
%% \noindent where also $i,j\in\mathbb{N}_0$ and 
%% \begin{align}
%%   \alpha_{ij} =
%%   \begin{cases} (2^{i+j+1}) \chi_{\hat{K}} & \text{if } i=0 \lor j=0,\\
%%     (3^{i+j}) \chi_{\hat{K}}   & \text{otherwise},\end{cases}\label{eq:alpha_coeff}
%% \end{align}
%% \noindent where $\chi_{\hat{K}}$ is the characteristic function of the domain $\hat{K}$\footnote{$\chi_{\hat{K}}$ takes value 1 $\forall \hat{\bm{r}}\in \hat{K}$ and 0 elsewhere.} and the values of the scaling constants in (\ref{eq:alpha_coeff}) ensure that all shape-functions take values in $[0,1]^2$ for all $\hat{\bm{r}}\in\hat{K}$. We remark that all functions $\bm{w}_{\hat{K}}^{ij2}$ with $i\neq=0$ and all functions $\bm{w}_{\hat{K}}^{ij1}$ with $j\neq=0$  are \emph{bulk} functions in the sense that they have vanishing tangential trace on lines $\xi_1=0$ and $\xi_2=0$. On the other hand, functions $\bm{w}_{\hat{K}}^{i01}$ and $\bm{w}_{\hat{K}}^{0j2}$ are edge functions: they yield monomials in arc-length when their tangential traces are computed on lines $\xi_1=0$ and $\xi_2=0$ respectively.
\begin{align}
  &\hat{\bm{w}}_{l}^{ij} (\hat{\bm{r}}) = C_{ijl}\,(\xi_{l})^i(\xi_{3-l})^j \hat{\nabla}\xi_{l},\;\;\;
  &l\in\{1,2\},\;i,j\geq0,\;i+j \leq p,\label{eq:ref_vecfuns}
\end{align}
\noindent where $i$,$j$ are integers, and $\hat{\nabla}$ denotes the gradient operator in the local coordinates.
The values of scaling factors $C_{ijl} \in\mathbb{R}^+$ ensure that shape-functions take all values in $[0,1]^2$ for some $\hat{\bm{r}}\in\hat{K}$.

We remark that local shape-functions defined in (\ref{eq:ref_vecfuns}) are of two kinds. For example, by setting $j=0$ we get ``edge'' functions, in the following sense: the selected shape-functions yield monomials in arc-length when their tangential trace is computed on the line $\xi_{l}=0$ and yield zero when their tangential trace is computed on $\xi_{3-l}=0$. 
This is a useful property when mapping vector-valued functions back to the \emph{physical} element $K\in\mathcal{K}_2^\Omega$. To do so we have to digress shortly on the index $l$, which is in fact a function of two additional indices: we can write (with some harmless abuse of notation in identifying sets with their indexing) $l=l(e,K)$, for any $e \in \mathcal{K}_1^\Omega \cap \mathcal{S}(\mathcal{C}^\Omega)$ and any $K\in\mathcal{K}_2^\Omega$ s.t. ${e}\subset\partial{K}$. 
This completely specifies which one of the local coordinates $(\xi_1, \xi_2)$ provides an arc-length parametrization for the image of segment ${e}$ (under the appropriate mapping $\varphi_K^{-1}$) and allows us to introduce the set of functions:
\begin{align}
  & \bm{w}_{e}^i(\bm{r}) :=
  \begin{cases}\mathbf{A}_K^{-\mathrm{T}} \hat{\bm{w}}_{l}^{i0} (\,\varphi_K^{-1}(\bm{r})\,),
  & \forall\bm{r}\in K,\,\forall K \in\mathcal{K}_2^\Omega \;\;s.t.\;\; 
  {e} \subset \{\partial{K}\},\, l = l(e,K),\\
               0 & \text{otherwise,}\end{cases} \label{eq:phys_edgefuns}%% \\
\end{align}
\noindent where $(\cdot)^{-\mathrm{T}}$ denotes the inverse-transpose matrix. In (\ref{eq:phys_edgefuns}) a (piecewise-)covariant transformation has been used, as it preserves tangential traces (see \cite{monk}) on two relevant boundary segments (while allowing fully discontinuous functions on the intersections of $\partial K$ with the skeleton $\mathcal{S}(\tilde{\mathcal{C}}^\Omega)$ of the dual complex).
%This is precisely what allows us, when constructing the global approximation space on the whole $\Omega$, to define conforming basis-functions for the broken space $\bm{H}^{curl}(\tilde{\mathcal{C}}_2^\Omega)$.

For fixed polynomial order $p$, (\ref{eq:phys_edgefuns}) is not sufficient for a complete basis: we must move back to $\hat{K}$ and take also local shape-functions in (\ref{eq:ref_vecfuns}) with $j\neq 0$. These are ``bulk'' basis-functions, as their tangential component vanishes now on both local coordinate axes. To preserve this feature onto the global mesh, we again use their covariantly mapped versions:
\begin{align}
  & \bm{w}_{K}^{ijl}(\bm{r}) =
  \begin{cases}\mathbf{A}_K^{-\mathrm{T}} \hat{\bm{w}}_{l}^{ij} (\,\varphi_K^{-1}(\bm{r})\,), & \forall\bm{r}\in K, j>0,\\
               0 & \text{otherwise,}\end{cases} \label{eq:phys_bulkfuns}%% \\
\end{align}
\noindent where we note the appearance of $K$ as a subscript index, rather than ${e}$, and we note that both admissible values of $l$ now produce bulk functions.
Summarizing, by grouping the $\bm{w}_{e}^i$ (for all $i$ s.t. $0\leq i \leq p$ and all ${e}\in\{\mathcal{K}_1^\Omega\cap\mathcal{S}(\mathcal{C}^\Omega)\}$) together with the $\bm{w}_{K}^{ijl}$ (for all admissible $\{i,j,l\}$ and all $K\in\mathcal{K}_2^\Omega$) into a new sequence $\{\bm{w}_n^p\}_{n=1}^{N}$, we achieve a complete set of basis-functions for the space 
\begin{align*}\bm{W}^p := Span\{\,\{\bm{w}_n^p\}_{n=1}^{N}\,\} = \bm{H}^{curl}(\tilde{\mathcal{C}}_2^\Omega) \cap \boldsymbol{P}^p(\mathcal{K}_2^\Omega;\mathbb{R}^2),\end{align*}
\noindent where $\boldsymbol{P}^p(\mathcal{K}_2^\Omega;\mathbb{R}^2)$ denotes the space of vector-valued functions whose components are piecewise-polynomials of degree at most $p$ on each $K\in\mathcal{K}_2^\Omega$.
It is not difficult to compute the dimension of this global space for a given mesh: the $n$ index runs from 1 to $N$, with
\begin{align}
  & N = (p+1)\left(2|\mathcal{C}_1^\Omega| + 3p|\mathcal{C}_2^\Omega|\right) =
        2|\mathcal{C}_1^\Omega| + p|\mathcal{K}_1^\Omega\cap\mathcal{S}(\mathcal{C}^\Omega)| + p(p+1)|\mathcal{K}_2^\Omega|,\label{eq:dimvec}
\end{align}
\noindent where the relationships between $\mathcal{K}^\Omega$ and $\mathcal{C}^\Omega$ have been used to make the splitting into lowest order, edge and bulk basis-functions manifest.

For the finite-dimensional space which will be approximating the pseudo-vector ${H}(\bm{r},t)$ instead, we proceed by first defining a new pair of oblique local coordinates $(\tilde{\xi}_1,\tilde{\xi}_2)$ on the KC element through an additional family of affine mappings $\tilde{\varphi}_K$ (and their inverses $\tilde{\varphi}_{K}^{-1}$),
%given by
%\begin{align*}
%  & \hat{\bm{r}} := \begin{pmatrix*}[c] \xi_1\\ \\ \xi_2\end{pmatrix*} =
%    \begin{pmatrix*}[r] -\frac{2}{3} & \frac{1}{3}\\ & \\ \frac{1}{3} & -\frac{2}{3}\end{pmatrix*}
%    \begin{pmatrix*}[c] \tilde{\xi}_1 \\ \\ \tilde{\xi}_2 \end{pmatrix*} +
%    \begin{pmatrix*}[c] \frac{1}{3} \\ \\\frac{1}{3} \end{pmatrix*}
%  \; \iff \;
%  \tilde{\bm{r}} := \begin{pmatrix*}[c] \tilde{\xi}_1\\ \\ \tilde{\xi}_2\end{pmatrix*} =
%  \begin{pmatrix*}[c] -2 & -1\\ & \\-1 & -2\end{pmatrix*}
%  \begin{pmatrix*}[c] \xi_1 \\ \\\xi_2\end{pmatrix*} +
%  \begin{pmatrix*}[c] 1 \\ \\1 \end{pmatrix*},
%\end{align*}
which we can construct by enforcing the origin in the associated oblique coordinates' system to coincide with the point $\tilde{O}=(1/3,1/3)$ (for its sketch, we refer the reader again to Fig.~\ref{fig:ref_kite}), and by enforcing $0\leq \tilde{\xi}_1,\tilde{\xi}_2 \leq 1$ on $\hat{K}$.
Thus, by denoting the position vector with $\tilde{\bm{r}}$ in the new coordinates' system, we can concisely give the expressions of scalar-valued local shape-functions on $\hat{K}$.
Namely, we introduce the monomials
\begin{align*}
  & \hat{\tilde{w}}_{\tilde{l}}^{ij}(\tilde{\bm{r}}) := (\tilde{\xi}_{\tilde{l}})^i(\tilde{\xi}_{3-{\tilde{l}}})^j,\;\;
  & \tilde{l}\in\{1,2\}, \;\; i>0, j\geq 0,\, i+j \leq p,
\end{align*}
\noindent where we stress the fact that $i$ is now a strictly positive integer. 
In this case, differently from the vector-valued setting, we have to consider segments $\tilde{e} \in \{\mathcal{K}_1^\Omega \cap \mathcal{S}(\tilde{\mathcal{C}}^\Omega)\}$ s.t. $\tilde{\xi}_1$ and $\tilde{\xi}_2$ provide arc-length parameters on them when moving back to any $K\in\mathcal{K}_2^\Omega$ in the physical mesh. Consequently we have introduced a different index $\tilde{l}=\tilde{l}(\tilde{e},K)$ . Setting $j=0$ yields a first subset of basis-functions for the global space
\begin{align}
  & \tilde{w}_{\tilde{e}}^{i}(\bm{r}) :=
  \begin{cases} \hat{\tilde{w}}_{\tilde{l}}^{i0} (\,\tilde{\varphi}_{K}^{-1}(\bm{r})\,) & \forall \bm{r}\in K \;\;s.t.\;\; 
  \tilde{e} \subset \partial{K}, \tilde{l}=\tilde{l}(\tilde{e},K),\\
                0 & \text{otherwise,}\end{cases} \label{eq:phys_scaledgefuns}
\end{align}
\noindent obtained by simple piecewise combinations of local shape-functions' pull-backs. The ones in (\ref{eq:phys_scaledgefuns}) are again edge functions (even if scalar-valued ones, their support being an edge-patch in $\mathcal{K}^\Omega$). A new set of bulk basis-functions is also present, 
%defined on $\hat{K}$ as
%\begin{align*}
%  & \tilde{w}_{\hat{K}}^{ij}(\tilde{\bm{r}}) =
%  \tilde{w}_{\hat{K}}^{ij}(\,\tilde{\varphi}_{\hat{K}}^{-1}(\hat{\bm{r}})\,) := \tilde{\xi}_1^i \tilde{\xi}_2^j,
%\end{align*}
%\noindent where $i,j>0$, $0 \leq i+j \leq p$, 
defined by setting $\tilde{l}=1$ (without loss of generality) and requiring $i>0$ and $j>0$ to hold simultaneously. The condition on $i$ and $j$ ensures that the associated shape-functions have vanishing trace on both $\tilde{\xi}_1=0$ and $\tilde{\xi}_2=0$ lines. 
Once more via pull-backs of local shape-functions onto the generic physical fundamental cell $K\in\mathcal{K}_2^\Omega$, we get
\begin{align}
  & \tilde{w}_{K}^{ij}(\bm{r}) =
  \begin{cases} \hat{\tilde{w}}_{1}^{ij} (\,\tilde{\varphi}_{K}^{-1}(\bm{r})\,) & \forall \bm{r}\in K,\, i,j>0,\\
                0 & \text{otherwise,}\end{cases} \label{eq:phys_scalbulkfuns}
\end{align}
\noindent again using $K$ as an index.
To complete the scalar-valued basis, a third set of functions is required, namely the set
\begin{align*}\tilde{w}_\mathcal{T}:=\frac{\mathbbm{1}_\mathcal{T}}{|\mathcal{T}|},\end{align*}
\noindent for all the triangles $\mathcal{T}\in\mathcal{C}_2^\Omega$, where $|\mathcal{T}|$ denotes the measure of $\mathcal{T}$ and $\mathbbm{1}_\mathcal{T}$ is the characteristic (or indicator) function of $\mathcal{T}$, i.e. the discontinuous function which takes value one for any $\bm{r}\in \mathcal{T}$ and zero elsewhere. Since the latter are piecewise-constant, scalar-valued functions, the mapping from reference to physical elements is trivial.

We can again group all the $\tilde{w}_{\tilde{e}}^i$, $\tilde{w}_K^{ij}$ and $\tilde{w}_\mathcal{T}$ in a new sequence of basis-functions $\{\tilde{w}_m^p\}_{m=1}^M$, which provides the basis of a finite-dimensional subspace $\tilde{W}^p \subset H^{\bm{curl}}(\mathcal{C}_2^\Omega)$, where again $\tilde{W}^p := Span\{\,\{\tilde{w}_m^p\}_{m=1}^M\,\}$. Namely, we have constructed a basis for the vector space 
\begin{align*}\tilde{W}^p = H^{\bm{curl}}(\mathcal{C}_2^\Omega) \cap {P}^p(\mathcal{K}_2^\Omega;\mathbb{R}),\end{align*} 
\noindent where ${P}^p(\mathcal{K}_2^\Omega;\mathbb{R})$ is the space of piecewise-polynomials of degree at most $p$ on each $K\in\mathcal{K}_2^\Omega$.
Once more, we can easily compute the dimension of $\tilde{W}^p$, which amounts to
\begin{align}
  & M = \left(1 + 3p + \frac{3}{2}p(p-1)\right)|\mathcal{C}_2^\Omega| =
        |\mathcal{C}_2^\Omega| + p|\mathcal{K}_1^\Omega\cap\mathcal{S}(\tilde{\mathcal{C}}^\Omega)| + \frac{p}{2}(p-1)|\mathcal{K}_2^\Omega|,\label{eq:dimscal}
\end{align}
\noindent where the contributions due to the three different flavours of basis-functions have been again manifestly split.

With the aid of $\bm{W}^p$ and $\tilde{W}^p$, we can finally approximate the unknown fields with a Galerkin method: we seek $\bm{E}^{h,p}(\bm{r},t) \in AC\!\left([0,\mathrm{T}]\right)\otimes \bm{W}^p$ and ${H}^{h,p}(\bm{r},t) \in AC\!\left([0,\mathrm{T}]\right) \otimes \tilde{W}^p$ such that
\begin{align}
& \sum_{\tilde{\mathcal{T}}\in\tilde{\mathcal{C}}_2^\Omega}
\left( \varepsilon \partial_t\bm{E}^{h,p},\bm{v}\right)_{\tilde{\mathcal{T}}} =
\sum_{\tilde{\mathcal{T}}\in\tilde{\mathcal{C}}_2^\Omega} \left(
\int_{\partial\tilde{\mathcal{T}}} \hspace{-2.5mm} {H}^{h,p}\bm{v}\cdot\hat{\bm{t}}(\ell)\,\mathrm{d}\ell
-\left( {H}^{h,p}, curl\left(\bm{v}\right)\right)_{\tilde{\mathcal{T}}} \right),\label{eq:weak_ampmax}\\
& \sum_{\mathcal{T}\in\mathcal{C}_2^\Omega}
\left( \mu \partial_t{H}^{h,p},u \right)_\mathcal{T} =
\sum_{\mathcal{T}\in\mathcal{C}_2^\Omega} \left(
\int_{\partial\mathcal{T}} \hspace{-2.5mm}u\bm{E}^{h,p}\cdot\hat{\bm{t}}(\ell)\,\mathrm{d}\ell
-\left( \bm{E}^{h,p},\bm{curl}(u) \right)_\mathcal{T}\right),\label{eq:weak_faraday}
\end{align}
\noindent hold $\forall\bm{v} \in \bm{W}^p$ and $\forall u\in \tilde{W}^p$ simultaneously.
Furthermore, the following assertion holds:
\begin{thm}\label{thm:lemma1}
    \textbf{ (Consistency and stability)} %If the test-functions $u \in H^{\bm{curl}}(\mathcal{C}_2^\Omega)$ and $\bm{v}\in \bm{H}^{curl}(\tilde{\mathcal{C}}_2^\Omega)$ are piecewise-smooth on each $K\in\mathcal{K}_2^\Omega$,
    The semi-discrete formulation (\ref{eq:weak_ampmax})--(\ref{eq:weak_faraday}) is consistent, meaning that it is satisfied by the true (conforming) weak solution of (\ref{eq:local_am})--(\ref{eq:local_fa}) in the limit $h\rightarrow 0$. Furthermore the semi-discrete electromagnetic energy $\mathcal{E}_K^{h,p}$ stored inside each $K\in\mathcal{K}_2^\Omega$ (and therefore in the whole of $\Omega$) is conserved through time:
    \begin{align}
    & \partial_t\mathcal{E}_K^{h,p} := \partial_t \left(\frac{1}{2}\Vert\bm{E}^{h,p}\Vert_{K,\varepsilon}^2 + \frac{1}{2}\Vert {H}^{h,p} \Vert_{K,\mu}^2\right) = 0, \; \forall t \in [0,\mathrm{T}],\label{eq:stability}
    \end{align}
    \noindent where $\varepsilon$ and $\mu$ are piecewise-smooth inside each $K\in\mathcal{K}_2^\Omega$.
\end{thm}
\begin{proof}
    Consistency is trivial, we prove (\ref{eq:stability}). We start by splitting all integrals into their contributions from each fundamental cell $K\in\mathcal{K}_2^\Omega$, which is straightforward for double integrals but requires some care for boundary terms. From (\ref{eq:weak_ampmax})--(\ref{eq:weak_faraday}) it ensues
    \begin{align*}
    & \sum_{ {\color{black}K \in \mathcal{K}_2^\Omega} }
    \left( \varepsilon \partial_t\bm{E}^{h,p},\bm{v}\right)_{ {\color{black}K} } \!=\!
    \sum_{ {\color{black}K \in \mathcal{K}_2^\Omega} } \left(
    \int_{ {\color{black}\partial{K}\cap\mathcal{S}(\tilde{\mathcal{C}}^\Omega)}} \hspace{-12.5mm} {H}^{h,p}\bm{v}\cdot\hat{\bm{t}}(\ell)\,\mathrm{d}\ell
    -\left( {H}^{h,p}, curl\left(\bm{v}\right)\right)_{{\color{black}K}} \right),
    &\!\forall \bm{v}\in \bm{W}^p,\\
    & \sum_{ {\color{black}K \in \mathcal{K}_2^\Omega} }
    \left( \mu \partial_t{H}^{h,p},u \right)_{{\color{black}K}} \!=\!
    \sum_{ {\color{black}K \in \mathcal{K}_2^\Omega} } \left(
    \int_{ {\color{black}\partial{K}\cap\mathcal{S}(\mathcal{C}^\Omega)} } \hspace{-12.5mm} u\bm{E}^{h,p}\cdot\hat{\bm{t}}(\ell)\,\mathrm{d}\ell
    -\left( \bm{E}^{h,p},\bm{curl}(u) \right)_{{\color{black}K}}\right),
    &\!\forall u\in \tilde{W}^p,
    \end{align*}
    where the definitions of sets $\mathcal{K}_1^\Omega\cap\mathcal{S}(\tilde{\mathcal{C}}^\Omega)$ and $\mathcal{K}_{1}^\Omega\cap\mathcal{S}(\mathcal{C}^\Omega)$ have been used to split line-integrals along the boundary of each $\tilde{\mathcal{T}}$ and $\mathcal{T}$ into local contributions. We now use the fact that our approximate solutions ${H}^{h,p}$ and $\bm{E}^{h,p}$ are themselves admissible test-functions (being linear combinations of the basis-functions) and plug them as such in the weak formulation:
    \begin{align*}
    & \sum_{ {\color{black}K \in \mathcal{K}_2^\Omega} }
    \left( \varepsilon \partial_t\bm{E}^{h,p},\bm{E}^{h,p}\right)_{ {\color{black}K} } =
    \sum_{ {\color{black}K \in \mathcal{K}_2^\Omega} } \left(
    \int_{ {\color{black}\partial{K}\cap\mathcal{S}(\tilde{\mathcal{C}}^\Omega)}} \hspace{-12.5mm} {H}^{h,p}\bm{E}^{h,p}\cdot\hat{\bm{t}}(\ell)\,\mathrm{d}\ell
    -\left( {H}^{h,p}, curl\left(\bm{E}^{h,p}\right)\right)_{{\color{black}K}} \right),\\
    & \sum_{ {\color{black}K \in \mathcal{K}_2^\Omega} }
    \left( \mu \partial_t{H}^{h,p},{H}^{h,p} \right)_{{\color{black}K}} =
    \sum_{ {\color{black}K \in \mathcal{K}_2^\Omega} } \left(
    \int_{ {\color{black}\partial{K}\cap\mathcal{S}(\mathcal{C}^\Omega)}} \hspace{-12.5mm}
    {H}^{h,p}\bm{E}^{h,p}\cdot\hat{\bm{t}}(\ell)\,\mathrm{d}\ell
    -\left( \bm{E}^{h,p},\bm{curl}({H}^{h,p}) \right)_{{\color{black}K}}\right).
    \end{align*}
    
    By adding the two equations together side-by-side, using the product rule for time derivatives on the l.h.s., while also using the Green theorem on the r.h.s., the assertion follows locally $\forall K\in\mathcal{K}_2^\Omega$.\end{proof}

The following remarks are in order: firstly, the definition of numerical fluxes is irrelevant as predicted (in a nutshell: when discretising the (ultra-)weak curls, wherever the test-functions present tangential trace jumps, trial-functions are tangentially continuous, and vice-versa).
On the other hand, the standard practice in Finite Element analysis is to set material tensors $\varepsilon$ and $\mu$ to a constant value on each triangle, since the primal complex is the one which is usually built (by some external tool) to resolve the geometry of discontinuities between materials. In this respect, the result of Theorem \ref{thm:lemma1} accommodates the output of any standard triangular mesher and at the same time suggests that well-behaved finite-dimensional approximations of the two new broken spaces should have local approximation properties not on whole primal and dual 2-cells, but on each $K\in\mathcal{K}_2^\Omega$, as is the case in our construction.

Lastly, we comment on boundary conditions: since $\partial \Omega$ is a subset of the skeleton $\mathcal{S}({\mathcal{C}}^\Omega)$ of the primal complex, the space $H^{\bm{curl}}(\mathcal{C}_2^\Omega)$ does not have a well-defined tangential trace on the boundary of the computational domain. In the above derivation this non-conformity is only apparently ignored: instead formal \emph{natural boundary conditions} have been employed, which, as can be deduced integrating by parts the r.h.s. of (\ref{eq:local_am}) on any $K\in\mathcal{K}_2^\Omega$ for which $\partial{K}\cap\partial\Omega\neq\emptyset$, amount to weakly enforcing $H|_{\partial\Omega}=0$.
On the other hand, the definition of $\bm{H}^{curl}(\tilde{\mathcal{C}}_2^\Omega)$ and, more precisely, the definitions of $\bm{w}_n^p$ basis-functions ensure that PEC boundary condition, if sought, can be enforced in the strong sense, since $\bm{W}^p$ possesses a tangential trace nearly everywhere on $\partial\Omega$. Additionally, we remark that both proposed bases are hierarchical by construction.

We finally note that (informally speaking) one can also \emph{swap the broken spaces}, i.e. we can define 
\begin{align}
\bm{H}^{curl}(\mathcal{C}_2^\Omega) &= \left\{ \bm{v} \!\in\! \bm{L}^2(\Omega) \; s.t. \;
\bm{v}|_\mathcal{T} \!\in\! \bm{H}^{curl}(\mathcal{T}),\,\forall \mathcal{T} \!\in\! \mathcal{C}_2^\Omega\right\},\label{eq:dual_broken_vector_valued}\\
H^{\bm{curl}}(\tilde{\mathcal{C}}_2^\Omega) &= \left\{ u \!\in\! L^2(\Omega) \; s.t. \;
u|_{\tilde{\mathcal{T}}} \!\in\! H^{\bm{curl}}(\tilde{\mathcal{T}}),\,\forall\tilde{\mathcal{T}} \!\in\! \tilde{\mathcal{C}}_2^\Omega\right\}, \label{eq:dual_broken_scalar_valued}
\end{align}
\noindent in the continuous setting, where a new formulation with an analogous weak form, appropriate sets of basis-functions, and an equivalent of Theorem \ref{thm:lemma1} can be easily deduced from our previous construction. The only key difference of such a formulation would lie in boundary conditions: if (\ref{eq:dual_broken_vector_valued})--(\ref{eq:dual_broken_scalar_valued}) were to be again approximated by finite-dimensional spaces, the PMC boundary condition ${H}^{h,p}|_{\partial\Omega}= 0$ would become an essential one and be incorporated in the strong sense.

\section{Implementing the fully discrete scheme}\label{sec:basis}
Owing to the explicit construction of Section \ref{sec:funspaces}, we can expand the approximated unknown fields as
\begin{align}
  & \bm{E}^{h,p}(\bm{r},t) = \sum\limits_{n=1}^{N} u_{n} (t) \bm{w}_{n}^p(\bm{r}), \;\;
  {H}^{h,p}(\bm{r},t) = \sum\limits_{m=1}^{M} f_{m} (t) \tilde{w}_{m}^p (\bm{r}), \label{eq:solution_expansion}
\end{align}
\noindent where the space-time separation of variables assumption on the solution is incorporated via the time dependence of coefficients in the linear combinations.
%% It is not difficult to compute the dimension of the spaces for a given mesh: $m$ and $n$ indexes run from 1 to
%% \begin{align}
%%   & N = (p+1)\left(2|\mathcal{C}_1^\Omega| + 3p|\mathcal{C}_2^\Omega|\right) =
%%         2|\mathcal{C}_1^\Omega| + p|\mathcal{K}_1^\Omega\cap\mathcal{S}(\mathcal{C}^\Omega)| + p(p+1)|\mathcal{K}_2^\Omega|,\label{eq:dimvec}\\
%%   & M = \left(1 + 3p + \frac{3}{2}p(p-1)\right)|\mathcal{C}_2^\Omega| =
%%         |\mathcal{C}_2^\Omega| + p|\mathcal{K}_1^\Omega\cap\mathcal{S}(\tilde{\mathcal{C}}^\Omega)| + \frac{p}{2}(p-1)|\mathcal{K}_2^\Omega|,\label{eq:dimscal}
%% \end{align}
%% \noindent respectively, where the relationships between $\mathcal{K}^\Omega$ and $\mathcal{C}^\Omega$ have been used to make the splitting into lowest order, ``edge'' and ``bulk'' basis-functions manifest.
The semi-discrete scheme, which we obtained by using (\ref{eq:solution_expansion}) and testing against the same basis-functions, has the following matrix-representation:
\begin{align}
& \begin{pmatrix*}[c]
    \tilde{\mathbf{M}}_p^\varepsilon & \mathbf{0} \\
    \mathbf{0} & \mathbf{M}_p^\mu
  \end{pmatrix*} \frac{\mathrm{d}}{\mathrm{d}{t}}
  \begin{pmatrix*}[c]
    \bm{u}(t) \\ \bm{f}(t)
  \end{pmatrix*} =
  \begin{pmatrix*}[c]
   \mathbf{0} & \mathbf{C}_p^\mathrm{T} \\
   -\mathbf{C}_p & \mathbf{0}
  \end{pmatrix*}
  \begin{pmatrix*}[c]
    \bm{u}(t) \\ \bm{f}(t)
  \end{pmatrix*}, \label{eq:semidiscrete_cmp}
\end{align}
\noindent where $\bm{f}$ is the column-vector containing semi-discrete magnetic field degrees of freedom (DoFs), $\bm{u}$ is the column-vector containing semi-discrete electric field DoFs, and where the proved energy conservation property is reflected by the r.h.s. skew-symmetry. We remark that the basis-functions can be appropriately re-ordered, by grouping members which have support contained into some common $\mathcal{T}\in\mathcal{C}_2^\Omega$ for the ${H}^{h,p}(\bm{r},t)$ field, and contained into a common $\tilde{\mathcal{T}}\in\tilde{\mathcal{C}}_2^\Omega$ for the $\bm{E}^{h,p}(\bm{r},t)$ field, respectively. This is not mandatory, but we thus stress the block-diagonal nature of mass-matrices $\tilde{\mathbf{M}}_p^\varepsilon$ and $\mathbf{M}_p^\mu$, which is achieved by breaking the standard Sobolev spaces. To discretise time, we use the well-known leap-frog scheme, which is the symplectic time-integrator used by Yee in his seminal paper. The search for symplectic integrators of arbitrary order which keep the time-stepping explicit is an active topic of research (see \cite{RuthForest, Titarev2002, jcp_symplectic}) which goes beyond the scope of the present contribution (yet provides also a further future research direction). For the fully discrete scheme it ensues
\begin{align}
  \begin{pmatrix*}[c]
    \mathbf{u}^{\left(n+1/2\right)\tau} \\ \mathbf{f}^{\left(n+1\right)\tau}
  \end{pmatrix*} &=
  \begin{pmatrix*}[c]
  \mathbf{u}^{\left(n-1/2\right) \tau} \\ \mathbf{f}^{n\tau}
  \end{pmatrix*} + \nonumber\\
  &+\tau\begin{pmatrix*}[c]
    (\tilde{\mathbf{M}}_p^\varepsilon)^{-1} \!&\! \mathbf{0} \\
    \mathbf{0} \!&\! (\mathbf{M}_p^\mu)^{-1}
  \end{pmatrix*} 
  \begin{pmatrix*}[c]
   \mathbf{0} & \mathbf{C}_p^\mathrm{T} \\
   -\mathbf{C}_p & \mathbf{0}
  \end{pmatrix*}
  \begin{pmatrix*}[c]
  \mathbf{u}^{\left(n+1/2\right)\tau} \\ \mathbf{f}^{n\tau}
  \end{pmatrix*}, \label{eq:discrete_cmp}
\end{align}
\noindent where $\mathbf{f}$ is the column-vector containing (now fully discrete) magnetic field DoFs, $\mathbf{u}$ is the column-vector containing electric field DoFs, $\tau\in\mathbb{R}^+$ is the discrete time-step (whose upper bound for a stable scheme can quickly be estimated by, e.g., a power-iteration algorithm) and $n=0,1,\dots,\lceil \mathrm{T}/\tau \rceil$. The inverses of mass-matrices are easily computed by solving very small local systems of equations, once-and-for-all and block-by-block (see sparsity patterns in Fig.~\ref{fig:sparsity_h_check} and Fig.~\ref{fig:sparsity_p_check}).

\begin{figure}[!h]
    \centering
    \begin{minipage}{0.1\textwidth}
        \vspace{-3cm}
        $p=0$
    \end{minipage}
    \begin{minipage}{0.3\textwidth}
        \vspace{-3cm}
        \includegraphics[width=\textwidth]{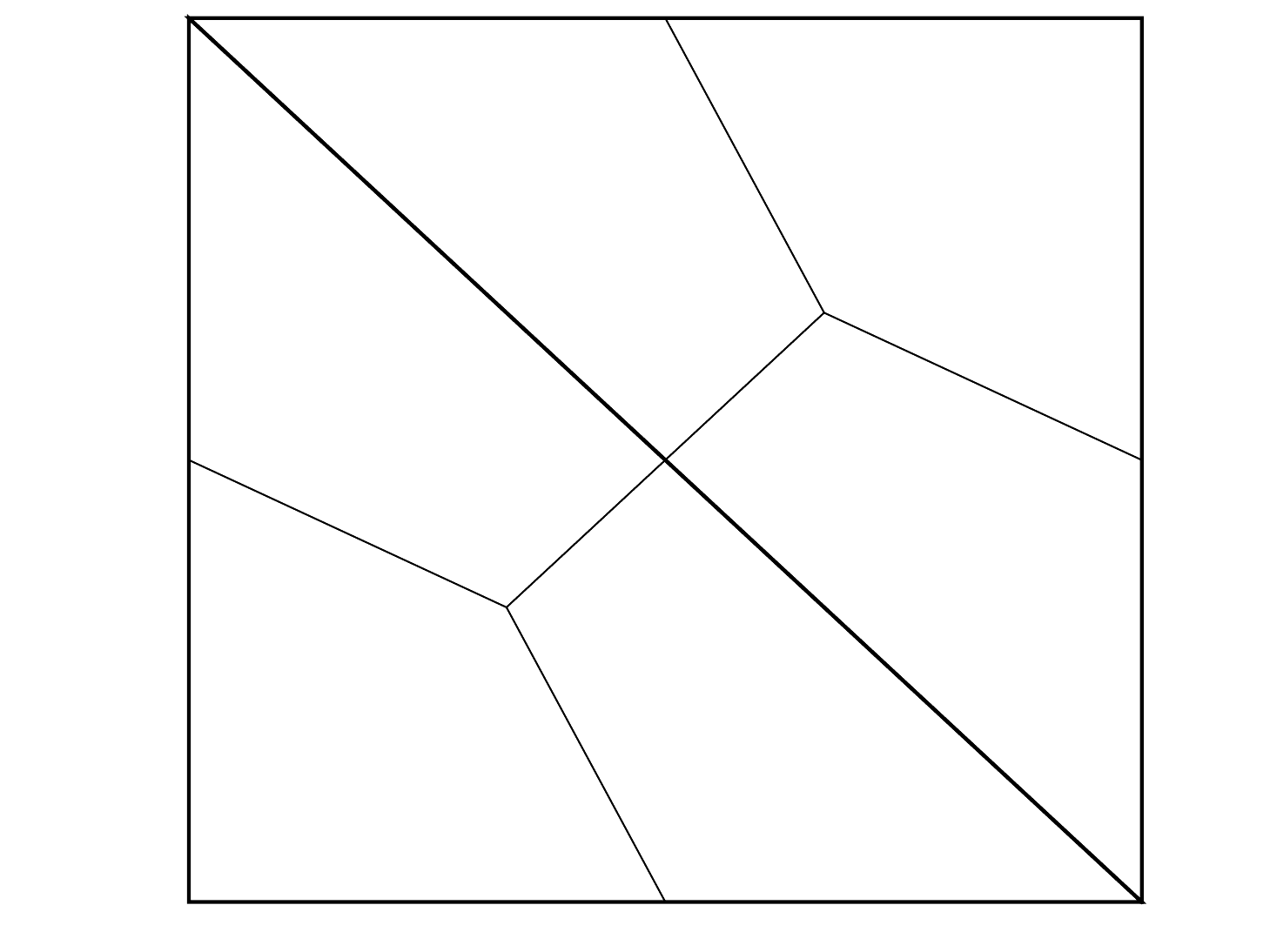}
    \end{minipage}
    \includegraphics[width=0.25\textwidth,height=0.25\textwidth]{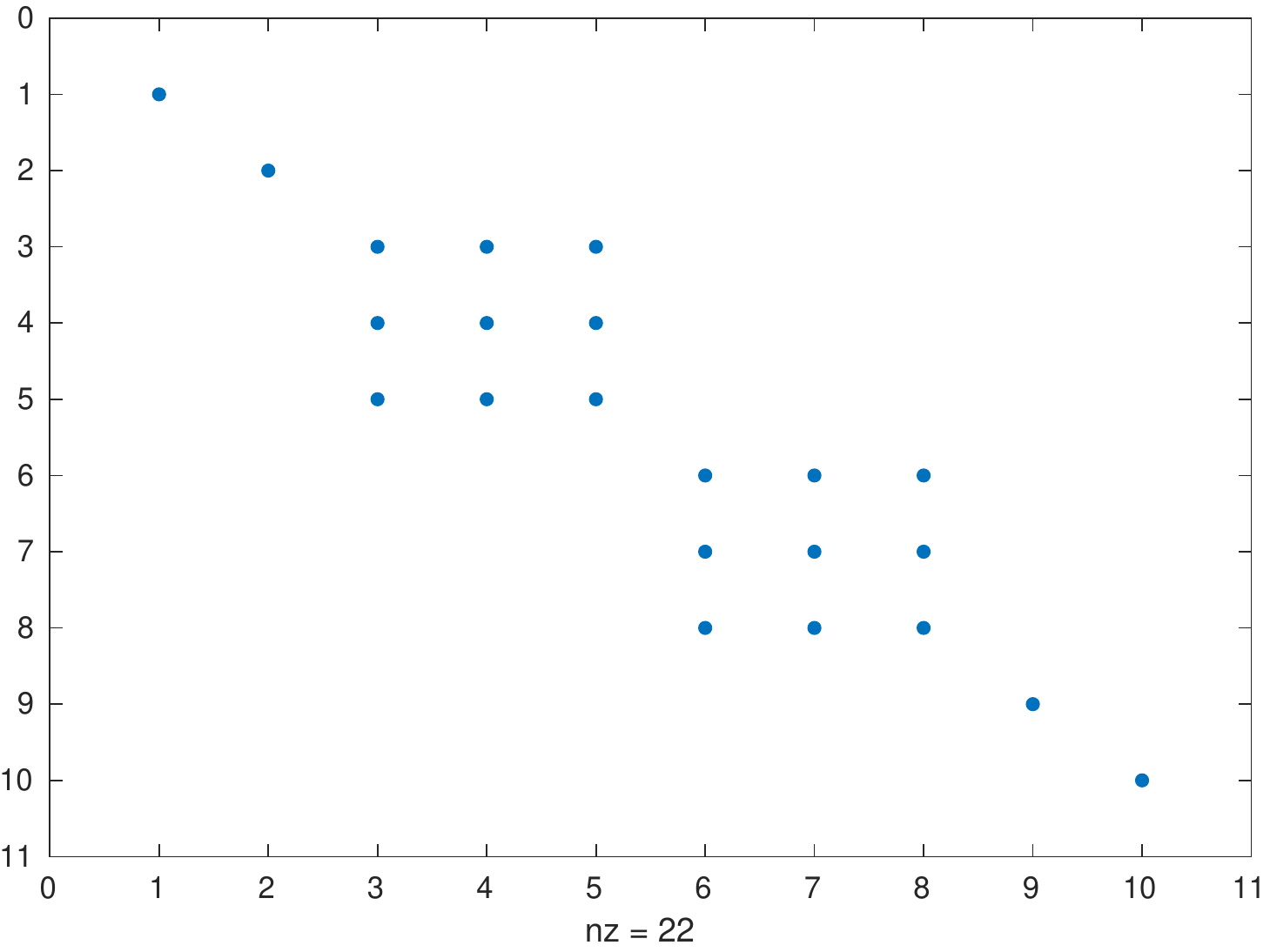}
    \includegraphics[width=0.25\textwidth,height=0.25\textwidth]{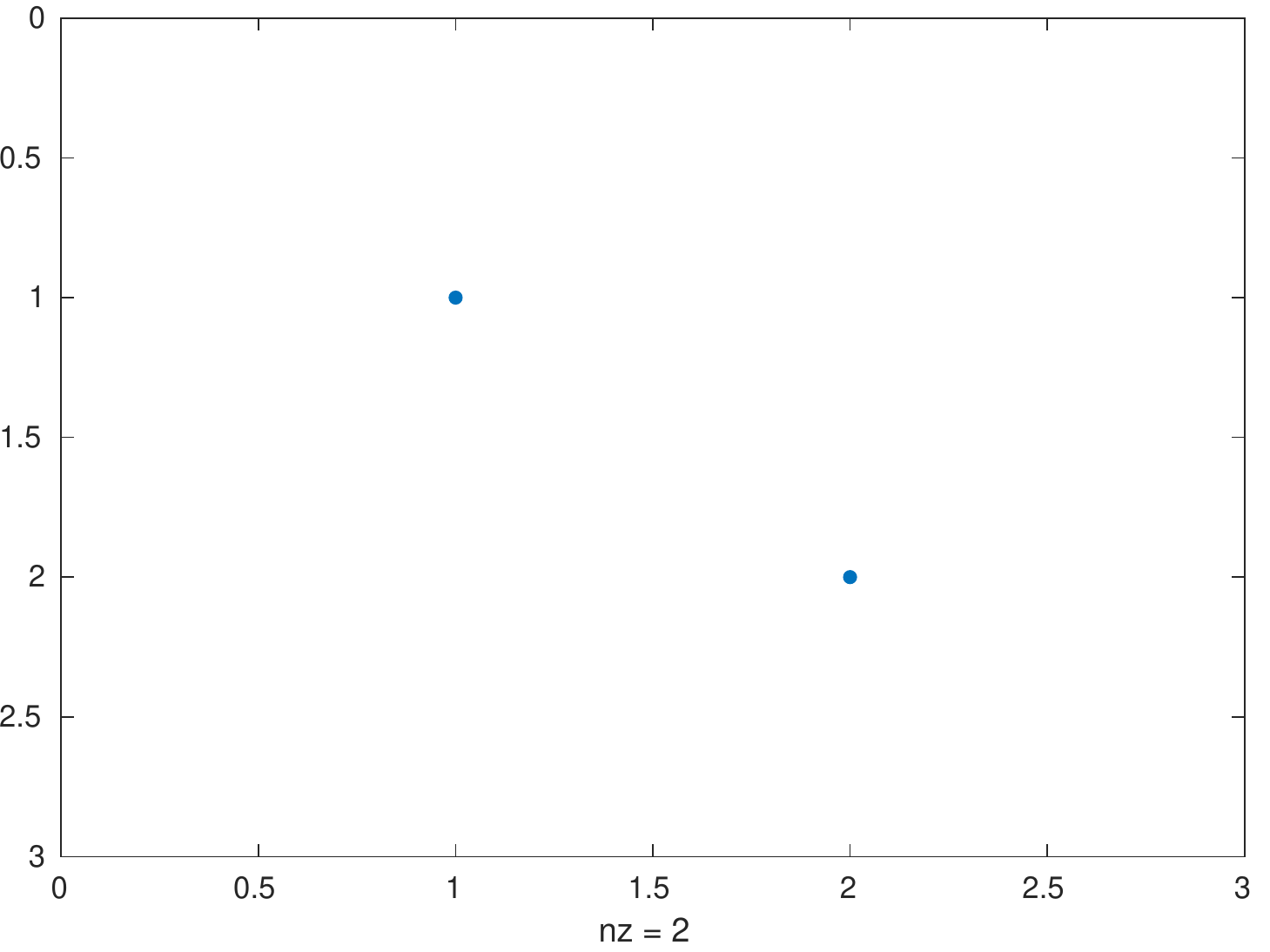}\\
    \centering
    \begin{minipage}{0.1\textwidth}
        \vspace{-3cm}
        $p=0$
    \end{minipage}
    \begin{minipage}{0.3\textwidth}
        \vspace{-3cm}
        \includegraphics[width=\textwidth]{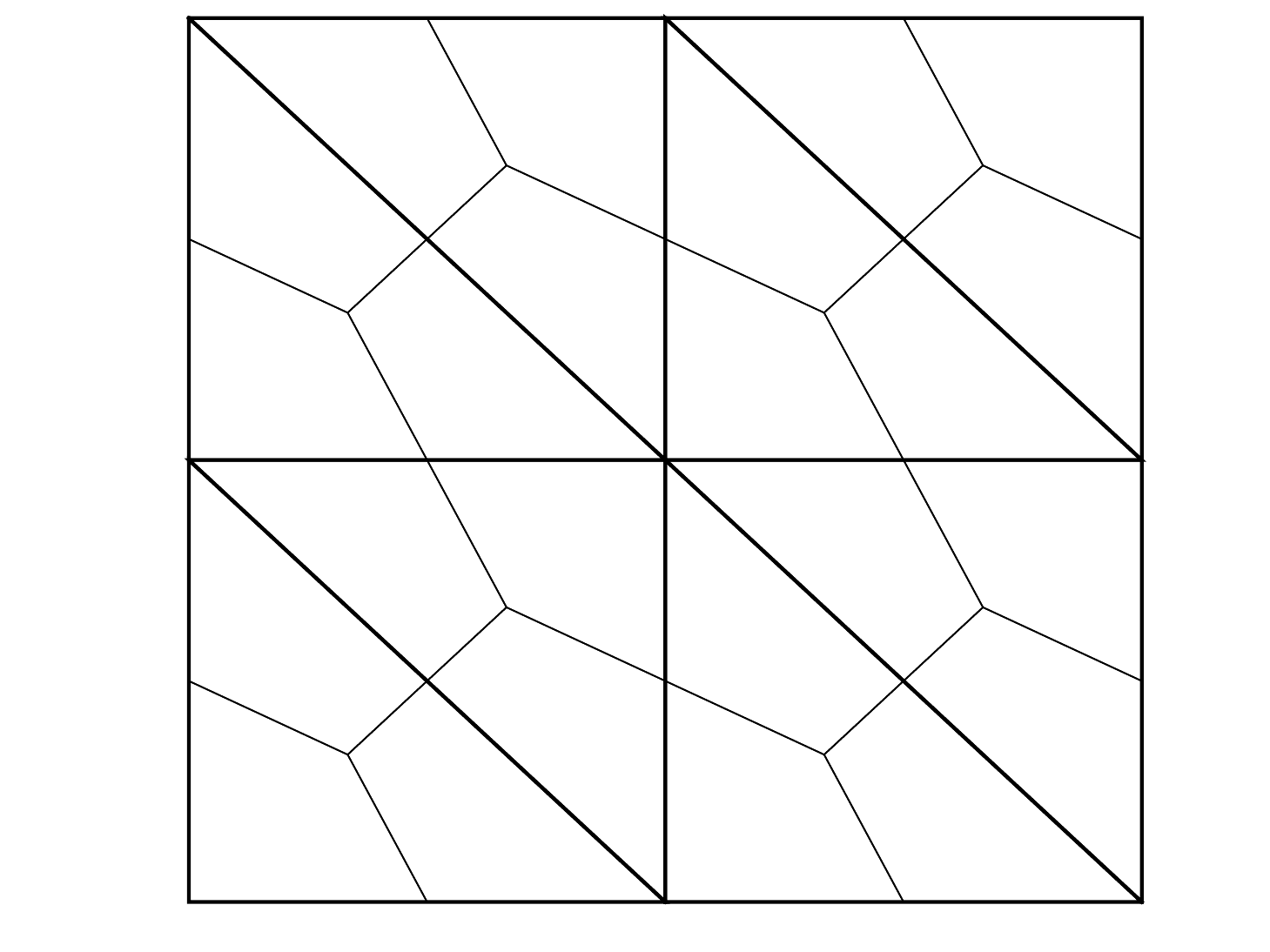}
    \end{minipage}
    \includegraphics[width=0.25\textwidth,height=0.25\textwidth]{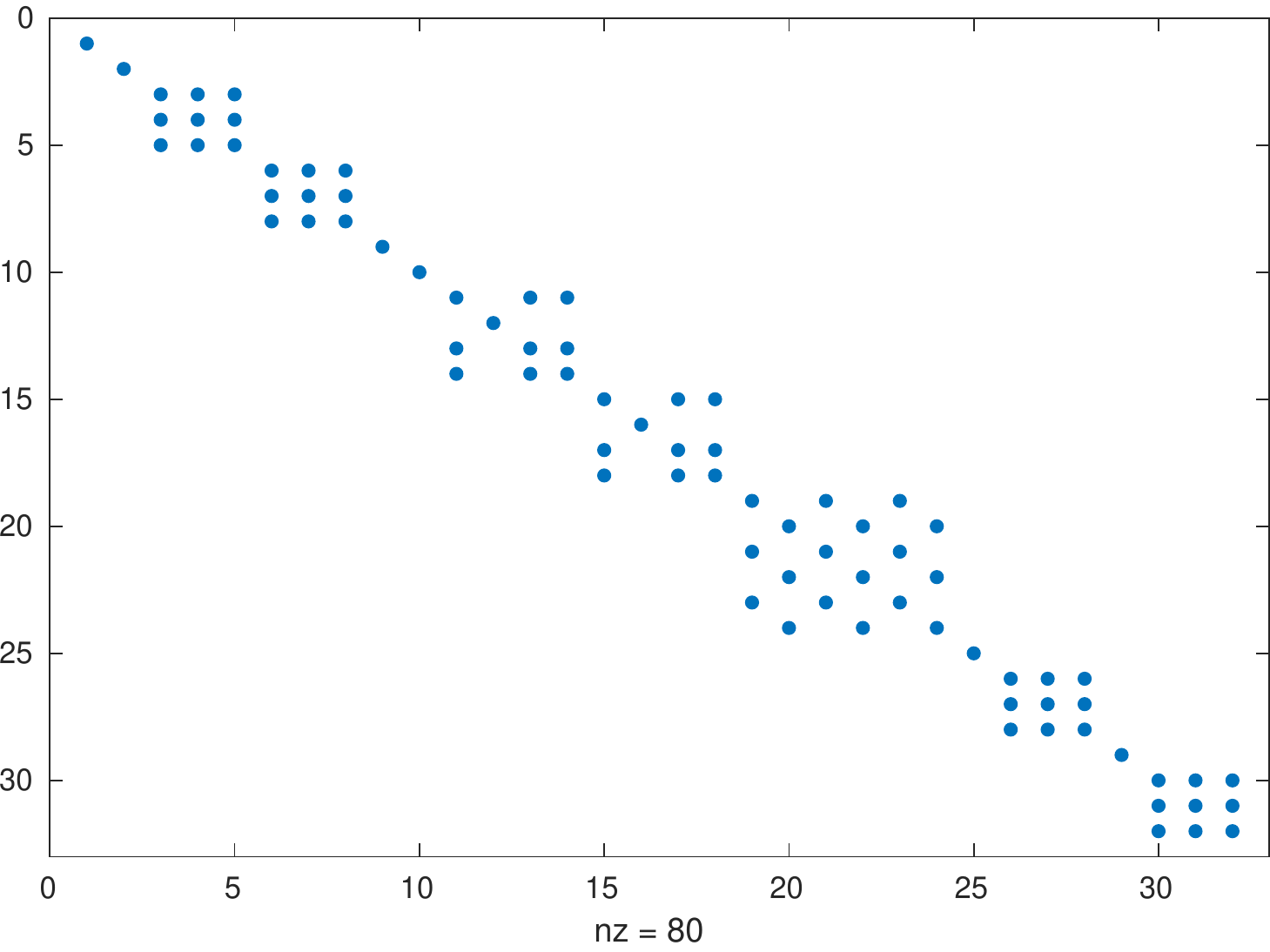}
    \includegraphics[width=0.25\textwidth,height=0.25\textwidth]{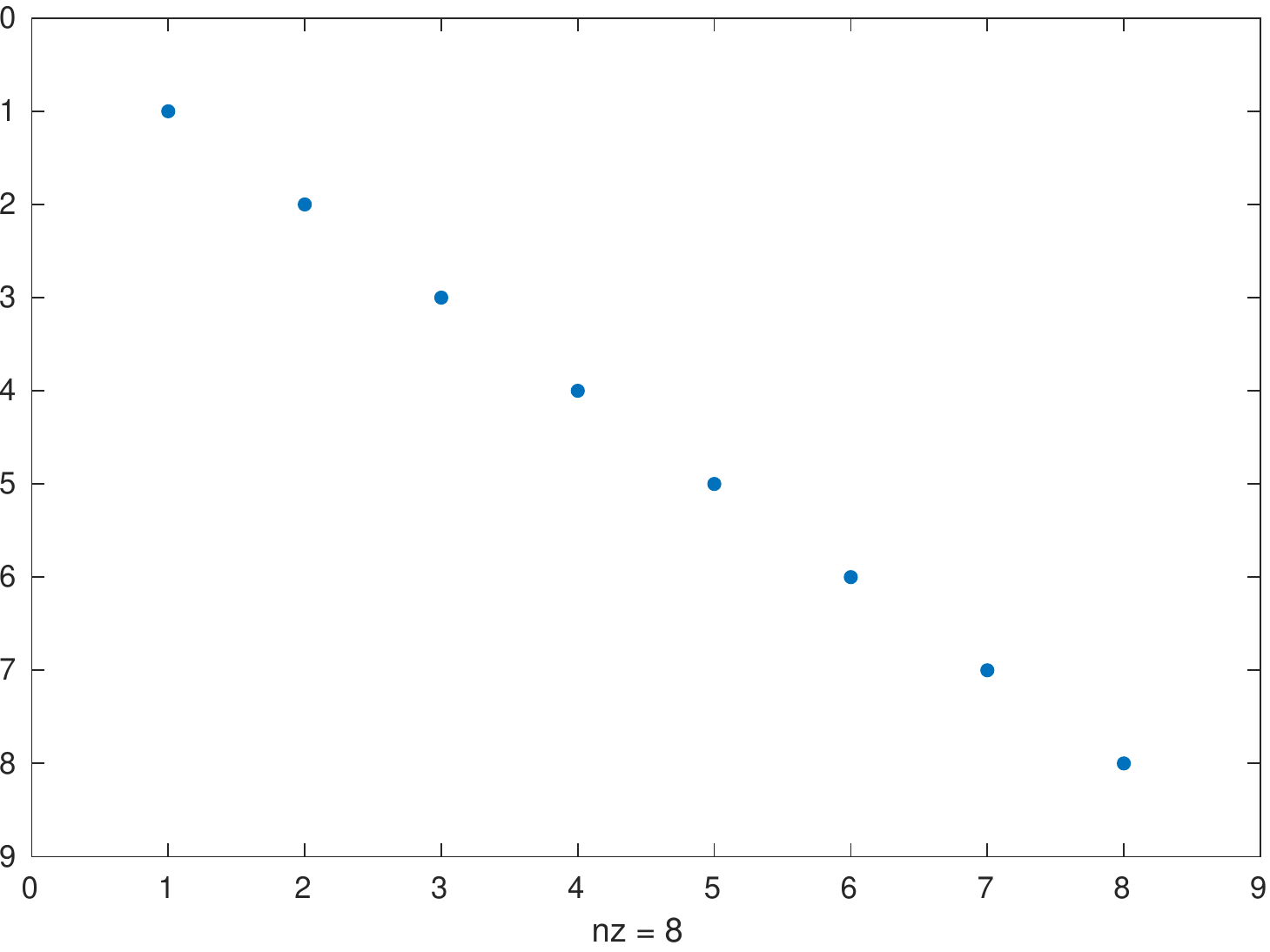}\\
    \centering
    \begin{minipage}{0.1\textwidth}
        \vspace{-3cm}
        $p=0$
    \end{minipage}
    \begin{minipage}{0.3\textwidth}
        \vspace{-3cm}
        \includegraphics[width=\textwidth]{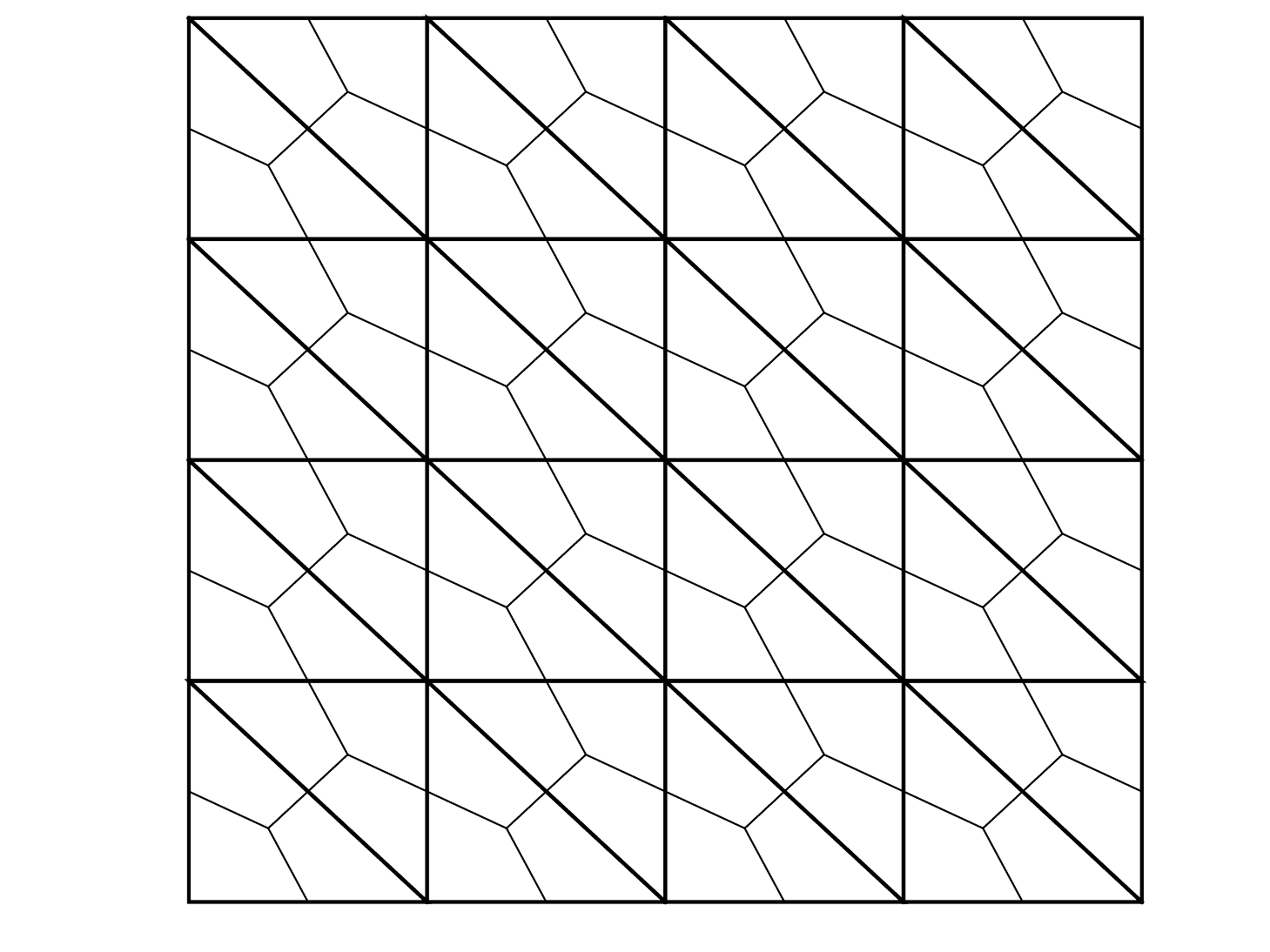}
    \end{minipage}
    \includegraphics[width=0.25\textwidth,height=0.25\textwidth]{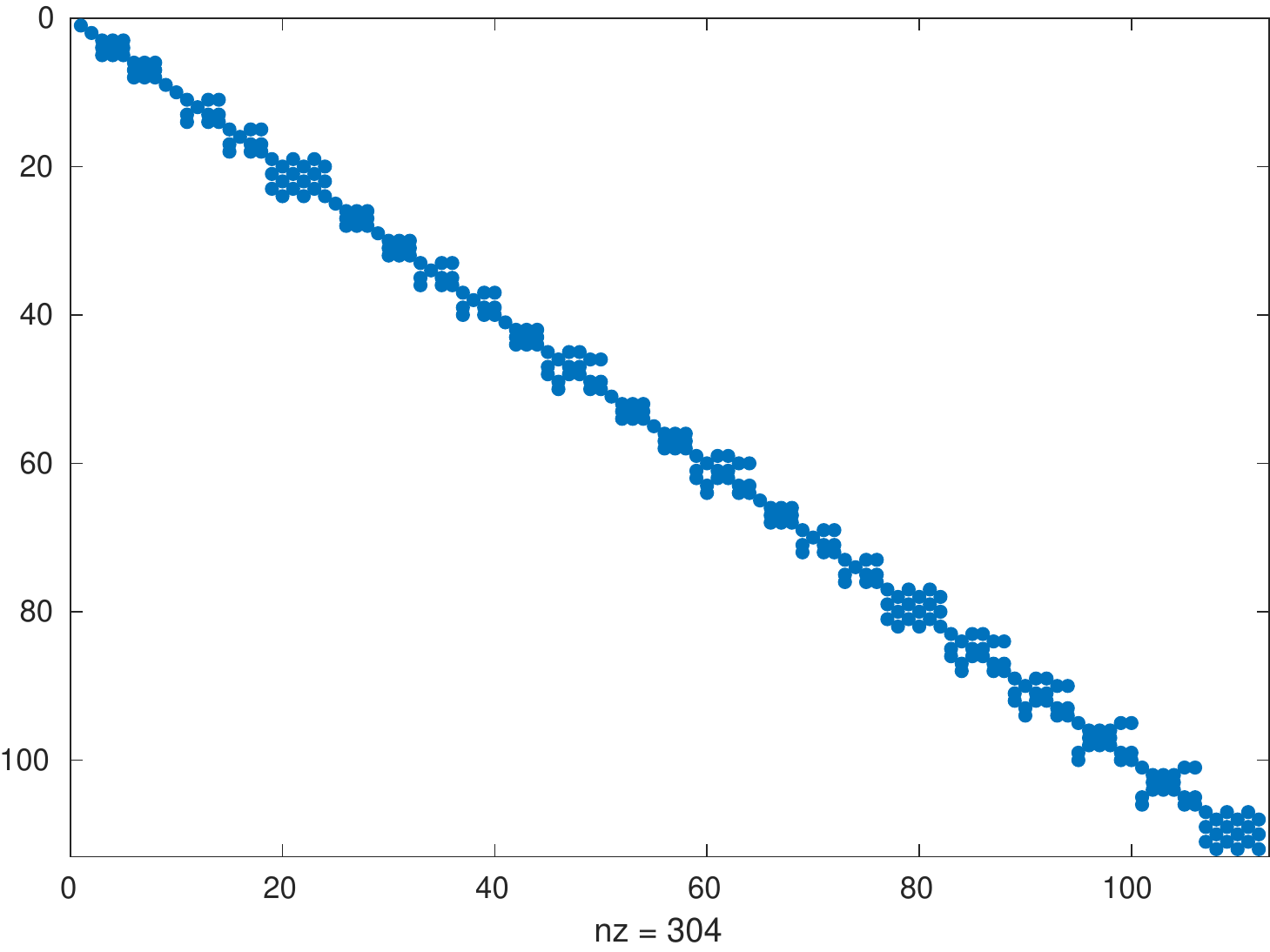}
    \includegraphics[width=0.25\textwidth,height=0.25\textwidth]{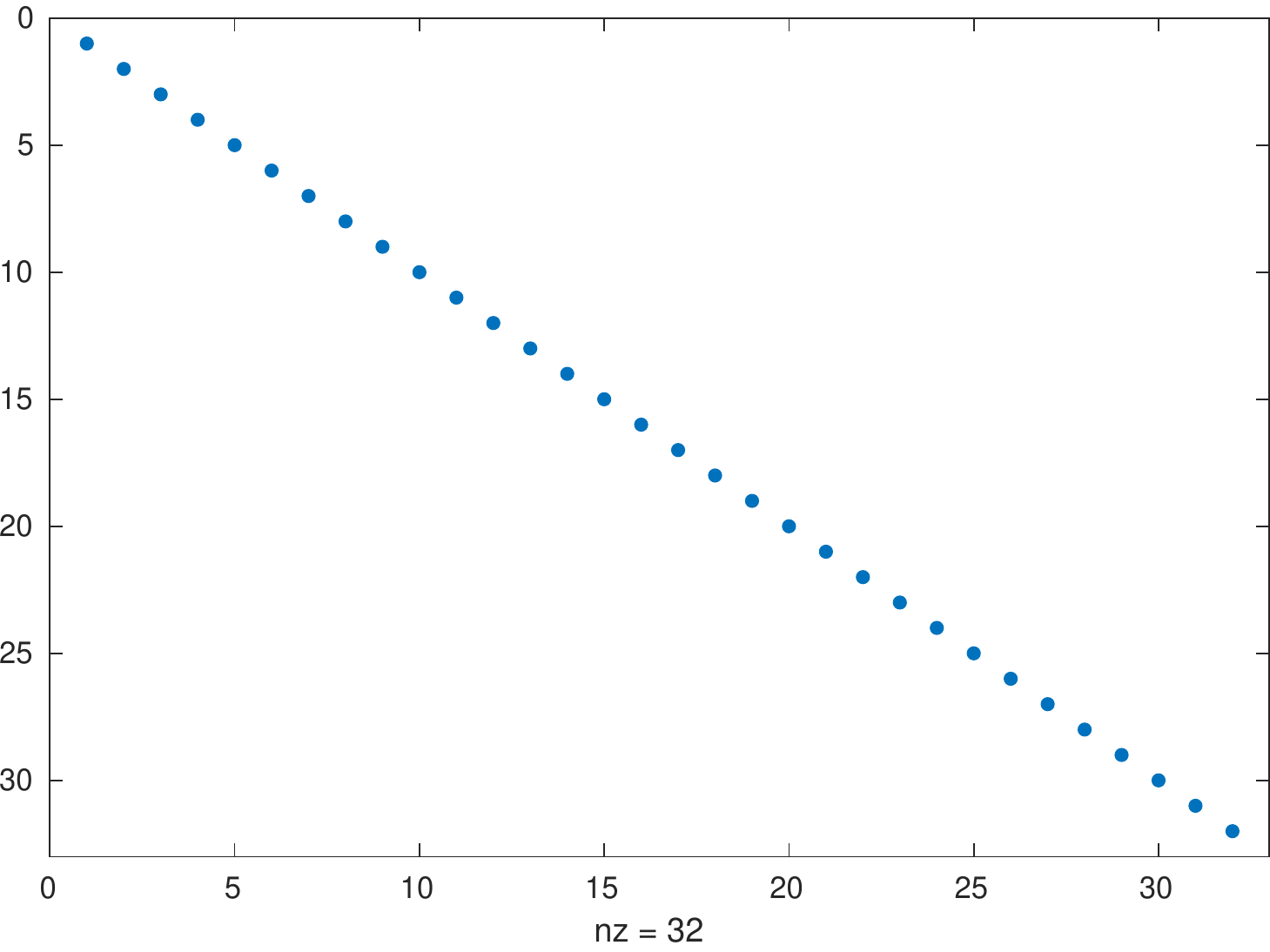}
    \caption{The sparsity pattern for the lowest order $\tilde{\mathbf{M}}_0^\varepsilon$ (second column) and $\mathbf{M}_0^\mu$ (third column) under uniform $h$-refinement, i.e. meshes in the first column are constructed by barycentric subdivision of uniform refinements (by means of edge-bisection) of the starting primal $\mathcal{C}^\Omega$ mesh. The label $\mathrm{nz}$ denotes the number of non-zero entries. Since for $p=0$ only $\tilde{w}_\mathcal{T}$ functions survive for the ${H}^{h,p}$ field, $\mathbf{M}_0^\mu$ is fully diagonal.}
   \label{fig:sparsity_h_check}
\end{figure}

Going now into greater depth with implementation-related details, we provide a procedure for the explicit computation of all matrix entries. 
%To fix the notation, we show the case of ``edge'' basis-functions of the $\bm{W}^p$ space, but equivalent derivations hold also for the remaining basis-functions.
It holds:
\begin{align}
\left(\tilde{\mathbf{M}}_p^\varepsilon\right)_{n,n'} &:=
\left(\varepsilon(\bm{r})\bm{w}_{n'}^p(\bm{r}),\bm{w}_{n}^p(\bm{r})\right)_\Omega
= \sum\limits_{K\in\mathcal{K}_2^\Omega} \left(\varepsilon(\bm{r})\bm{w}_{n'}^p(\bm{r}),\bm{w}_{n}^p(\bm{r})\right)_K = \nonumber \\
&=\sum\limits_{\substack{K\in \\ \{\mathcal{K}_2^\Omega \cap \,\supp(\bm{w}_n)\}}} \hspace{-5mm}
\left( J_K^{-1} \varepsilon(\varphi_K^{-1}(\hat{\bm{r}}))\mathbf{A}_K^{-\mathrm{T}}\hat{\bm{w}}_{l'}^{i'j'}(\hat{\bm{r}}),
\;\mathbf{A}_K^{-\mathrm{T}}\hat{\bm{w}}_{l}^{ij}(\hat{\bm{r}})\right)_{\hat{K}}= \nonumber\\
&=\sum\limits_{\substack{K\in \\ \{\mathcal{K}_2^\Omega\cap\,\supp(\bm{w}_n)\}}} \hspace{-5mm}
\left( J_K^{-1} \mathbf{A}_K^{-1}\varepsilon(\varphi_K^{-1}(\hat{\bm{r}}))\mathbf{A}_K^{-\mathrm{T}}\hat{\bm{w}}_{l'}^{i'j'}(\hat{\bm{r}}),
\;\hat{\bm{w}}_{l}^{ij}(\hat{\bm{r}})\right)_{\hat{K}}:= \nonumber\\
&:= \sum\limits_{\substack{K\in \\ \{\mathcal{K}_2^\Omega\cap\,\supp(\bm{w}_n)\}}} \hspace{-5mm}
\int_{\hat{K}} \left(\hat{\varepsilon}_K(\hat{\bm{r}})\hat{\bm{w}}_{l'}^{i'j'}(\hat{\bm{r}})\right) \cdot \hat{\bm{w}}_{l}^{ij}(\hat{\bm{r}})\,\mathrm{d}\hat{\bm{r}},\label{eq:mass_matrix_algebra}
\end{align}
\noindent where we reach line two (in which $\supp(\cdot)$ denotes the support of a function and $J_K$ is the Jacobian determinant of $\varphi_K$) by virtue of the transformation rules in (\ref{eq:phys_edgefuns}) and by using the local definition of shape-functions in (\ref{eq:ref_vecfuns}). We assume that local indices $i,j,l$ (in place of $n$) and $i',j',l'$ (in place of $n'$) exist such that the functional forms of some local shape-functions match the given $\bm{w}_{n'}^p$, $\bm{w}_{n}^p$. We remark that this is always true by construction of the space $\bm{W}^p$. Finally, (\ref{eq:mass_matrix_algebra}) is just a consistent re-definition where, for the sake of clarity, the modified material tensor $\hat{\varepsilon}_K := J_K^{-1} \mathbf{A}_K^{-1}\varepsilon\mathbf{A}_K^{-\mathrm{T}}$ has been introduced.

What (\ref{eq:mass_matrix_algebra}) means in practice is that, if the input mesh consists of straight-edged triangles and $\varepsilon$ is piecewise-constant on each $K$, all inner products in the mass-matrix con be computed (off-line with respect to the rest of computation) by working on the KC element, since the Jacobian (and hence $\hat{\varepsilon}_K$) is then piecewise-constant. These conditions are very often met in practical setups. A very similar procedure applies to the mass-matrix involving the scalar unknown, where a slightly different $\hat{\mu}_K:=J_K^{-1}\mu$ will arise, as the reader may also easily derive.

For the r.h.s. of (\ref{eq:semidiscrete_cmp}), we need to compute only half of the non-zero entries, as $\left(\mathbf{C}_p^{\mathrm{T}}\right)_{n,m} = \left(\mathbf{C}_p\right)_{m,n}$ (where $1\leq n \leq N$ and $1\leq m \leq M$). 
Since the involved algebra is a bit tedious, we omit in the following $\bm{r}$ and $\hat{\bm{r}}$ dependences to make the manipulations easier to read. Furthermore we assume that we are computing an entry related to a pair of trial- and test-functions which have non-empty intersection between their support (the matrix entry is trivially null otherwise). We compute
\begin{align}
  \left(\mathbf{C}_p\right)_{m,n} &:=
  \sum_{K\in \mathcal{K}_2^\Omega}\!\int_{\partial{K}\cap\mathcal{S}(\mathcal{C}^\Omega)} \hspace{-10mm}\tilde{w}_{m}^p \bm{w}_{n}^p\cdot\hat{\bm{t}}(\ell)\,\mathrm{d}\ell
    +\sum_{K\in \mathcal{K}_2^\Omega}\!\left( \bm{w}_{n}^p,\bm{curl}(\,\tilde{w}_{m}^p\,) \right)_K = \nonumber \\
  &= \int_{\substack{ {} \\ \partial{K}\cap\mathcal{S}(\mathcal{C}^\Omega)}} \hspace{-10mm}\tilde{w}_{m}^{p} \bm{w}_{n}^{p}\cdot\hat{\bm{t}}(\ell)\,\mathrm{d}\ell
    +\left( \bm{w}_{n}^{p},\bm{curl}(\,\tilde{w}_{m}^{p}\,) \right)_K = \nonumber \\
  &= \int_{\substack{ \\ \partial{\hat{K}}\cap\mathcal{S}(\hat{\mathcal{T}})}} \hspace{-10mm} 
  J_K\hat{\tilde{w}}_{\tilde{l}'}^{i'j'} \left(\mathbf{A}_K^{-\mathrm{T}}\hat{\bm{w}}_{l}^{ij}\right) \cdot \left(J_K^{-1}\mathbf{A}_K\hat{\bm{t}}(\hat{\ell})\right)\,\mathrm{d}\hat{\ell}
    \!+\nonumber\\
    &+\left( J_K \mathbf{A}_K^{-\mathrm{T}} \hat{\bm{w}}_{l}^{ij},J_K^{-1}\mathbf{A}_K\hat{\bm{curl}}(\, \hat{\tilde{w}}_{\tilde{l}'}^{i'j'}\,) \right)_{\hat{K}} = \nonumber \\
  &= \int_{\substack{ \\ \partial{\hat{K}}\cap\mathcal{S}(\hat{\mathcal{T}})}} \hspace{-10mm} \hat{\tilde{w}}_{\tilde{l}'}^{i'j'} \hat{\bm{w}}_{l}^{ij}\cdot\hat{\bm{t}}(\hat{\ell})\,\mathrm{d}\hat{\ell}
    \!+\left(\hat{\bm{w}}_{l}^{ij}, \hat{\bm{curl}}(\,\hat{\tilde{w}}_{\tilde{l}'}^{i'j'}\,) \right)_{\hat{K}} = \nonumber \\
  \begin{split}&=\!\int\limits_0^{\frac{1}{2}} \hat{\tilde{w}}_{\tilde{l}'}^{i'j'} \left(\hat{\bm{w}}_{l}^{ij}\cdot\hat{\bm{\xi_1}}\right)\!\big\rvert_{\xi_2=0}\,\mathrm{d}\xi_1
    \!-\!\int\limits_0^{\frac{1}{2}} \hat{\tilde{w}}_{\tilde{l}'}^{i'j'} \left(\hat{\bm{w}}_{l}^{ij}\cdot\hat{\bm{\xi_2}}\right)\!\big\rvert_{\xi_1=0}\,\mathrm{d}\xi_2 +\\
    &+ \left(\hat{\bm{w}}_{l}^{ij}, \hat{\bm{curl}}(\,\hat{\tilde{w}}_{\tilde{l}'}^{i'j'}\,) \right)_{\hat{K}},\end{split} \label{eq:weak_curl_algebra}
\end{align}
\noindent where the salient details are the following.
The disappearance of summation symbols on line two descends from the fact that, for any pair of basis-functions in $\tilde{W}^p \times \bm{W}^p$, there will be at most one $K\in\mathcal{K}_2^\Omega$ on which neither of the two identically vanishes, which we label precisely $K$. 
Consequently, the definition of global basis-functions again yields the existence of appropriate ``matching'' local shape-functions with respective indices $\{i,j,l\}$ and $\{i',j',\tilde{l}'\}$. 
Line three uses the chosen transformation rules for the local shape-functions, plus the fact that the $\mathbf{curl}$ of a covariant (pseudo-)vector field, under coordinate changes, transforms according to the contra-variant (also known as Piola) mapping, which is also the appropriate transformation rule for the tangent unit vector under the same change of coordinates (proofs of this standard facts can be found, e.g., in \cite{monk}).
The notation $\hat{\bm{curl}}$ is consequently introduced for the classical differential curl operator in the local (Cartesian) $\xi_1$ and $\xi_2$ coordinates. 
Line four is achieved by virtue of standard matrix-algebra simplifications. This yields the final result, in which the line- and double integrals are explicitly written in the local coordinates on $\hat{K}$.

\begin{figure}[!h]
    \centering
    \begin{minipage}{0.1\textwidth}
        \vspace{-3cm}
        $p=1$
    \end{minipage}
    \begin{minipage}{0.3\textwidth}
    \vspace{-3cm}
    \includegraphics[width=\textwidth]{spymesh1.pdf}
    \end{minipage}
    \includegraphics[width=0.25\textwidth,height=0.25\textwidth]{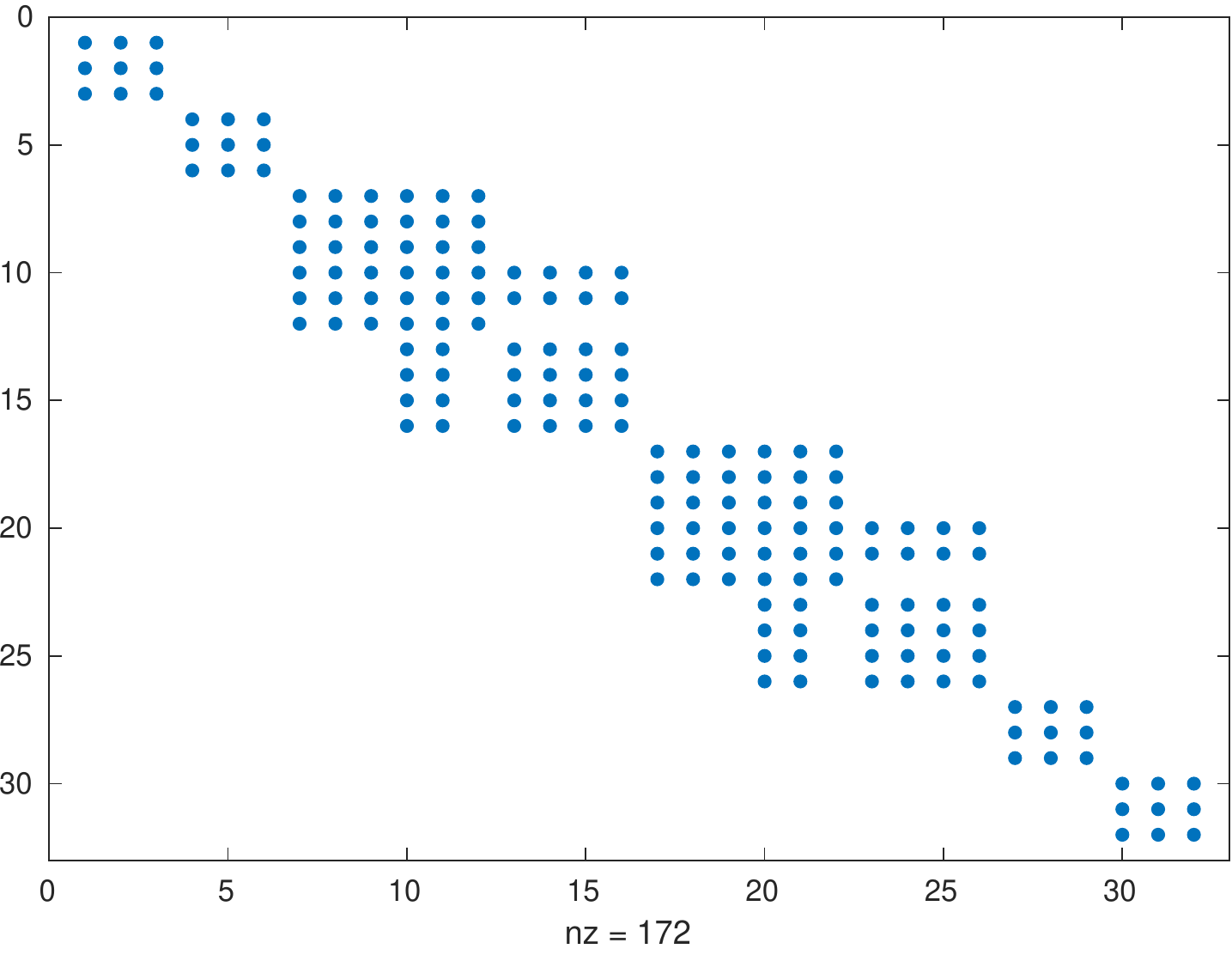}
    \includegraphics[width=0.25\textwidth,height=0.25\textwidth]{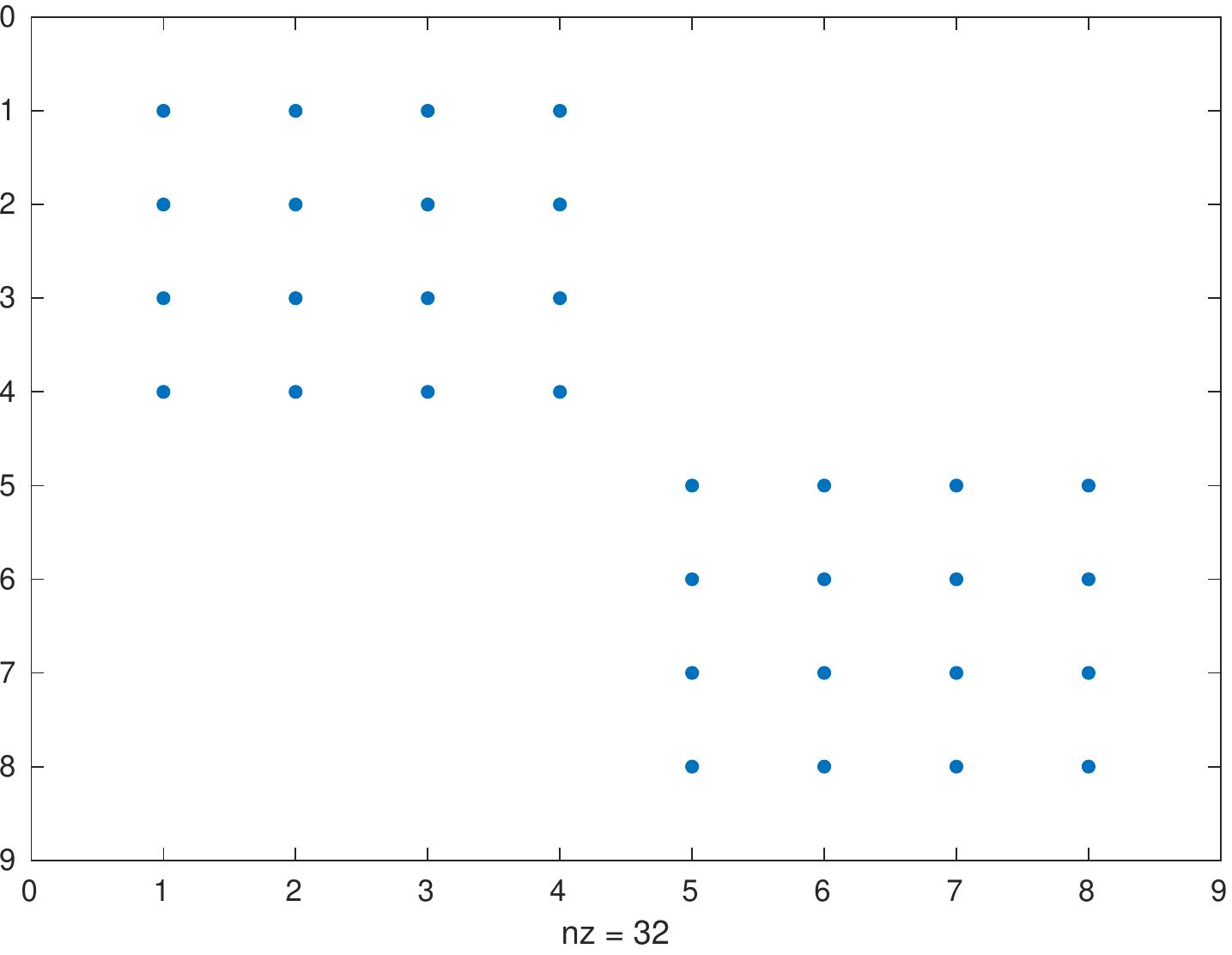}\\
    \centering
    \begin{minipage}{0.1\textwidth}
        \vspace{-3cm}
        $p=2$
    \end{minipage}
    \begin{minipage}{0.3\textwidth}
    \vspace{-3cm}
    \includegraphics[width=\textwidth]{spymesh1.pdf}
    \end{minipage}
    \includegraphics[width=0.25\textwidth,height=0.25\textwidth]{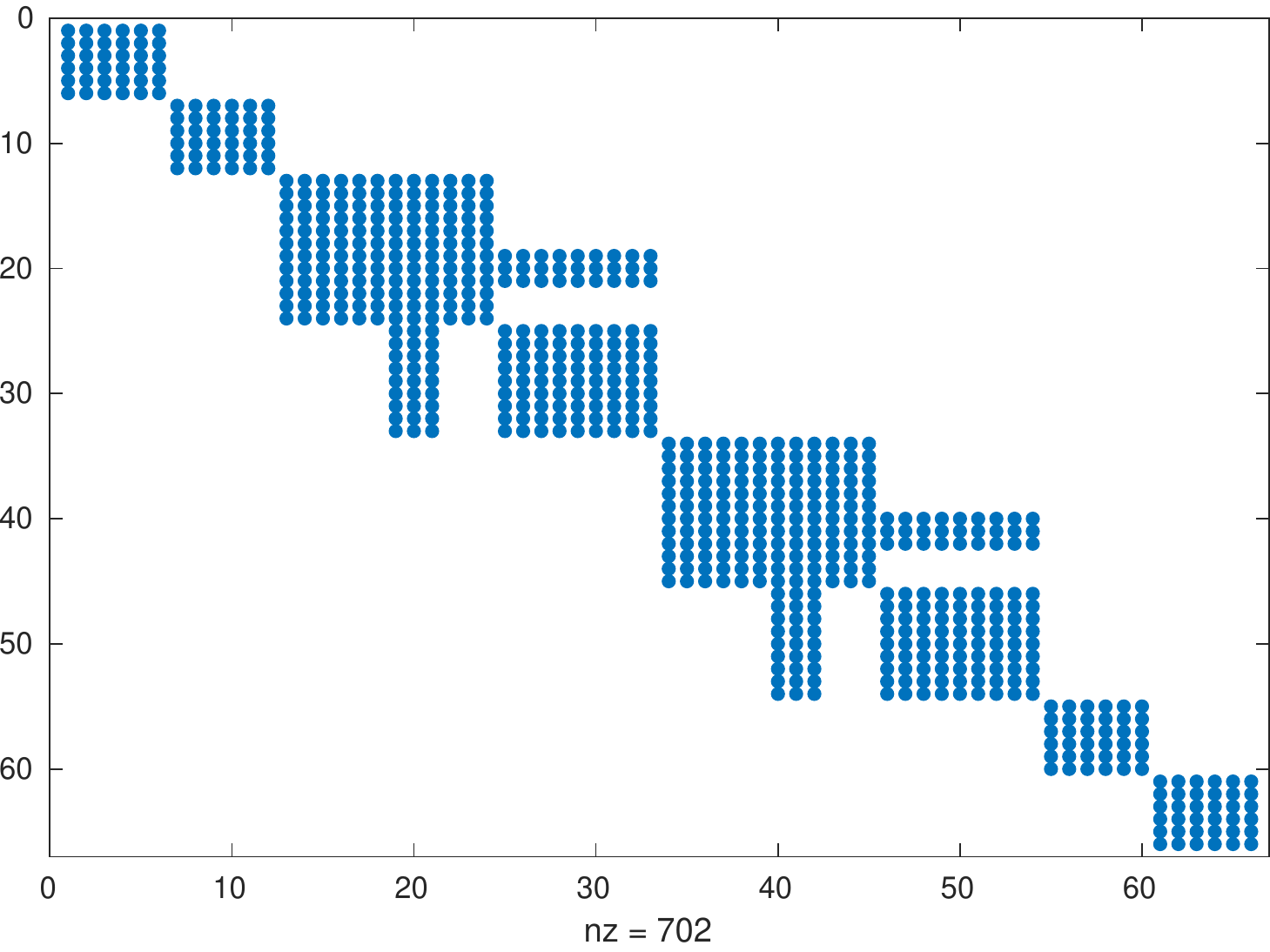}
    \includegraphics[width=0.25\textwidth,height=0.25\textwidth]{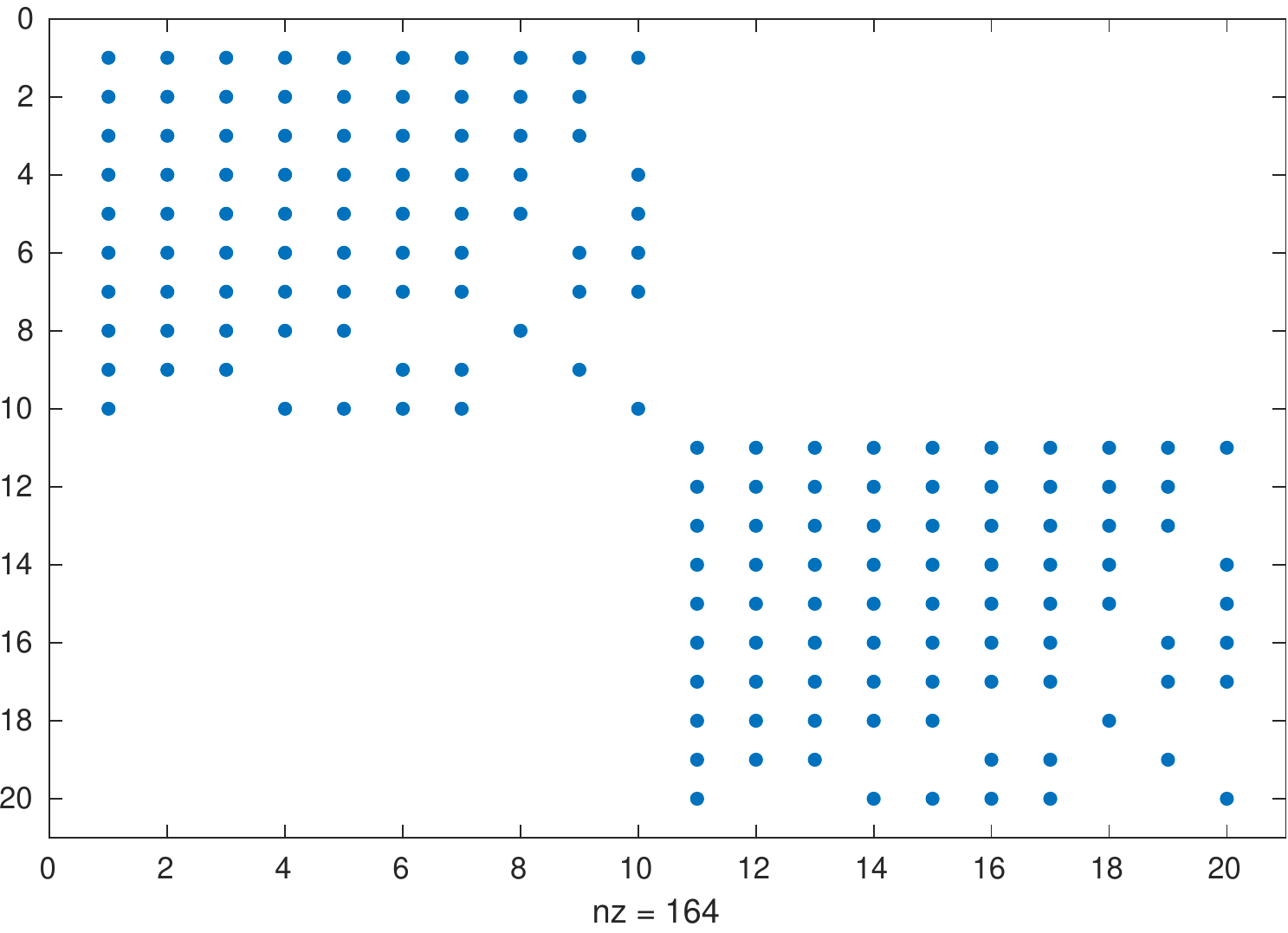}\\
    \centering
    \begin{minipage}{0.1\textwidth}
        \vspace{-3cm}
        $p=3$
    \end{minipage}
    \begin{minipage}{0.3\textwidth}
    \vspace{-3cm}
    \includegraphics[width=\textwidth]{spymesh1.pdf}
    \end{minipage}
    \includegraphics[width=0.25\textwidth,height=0.25\textwidth]{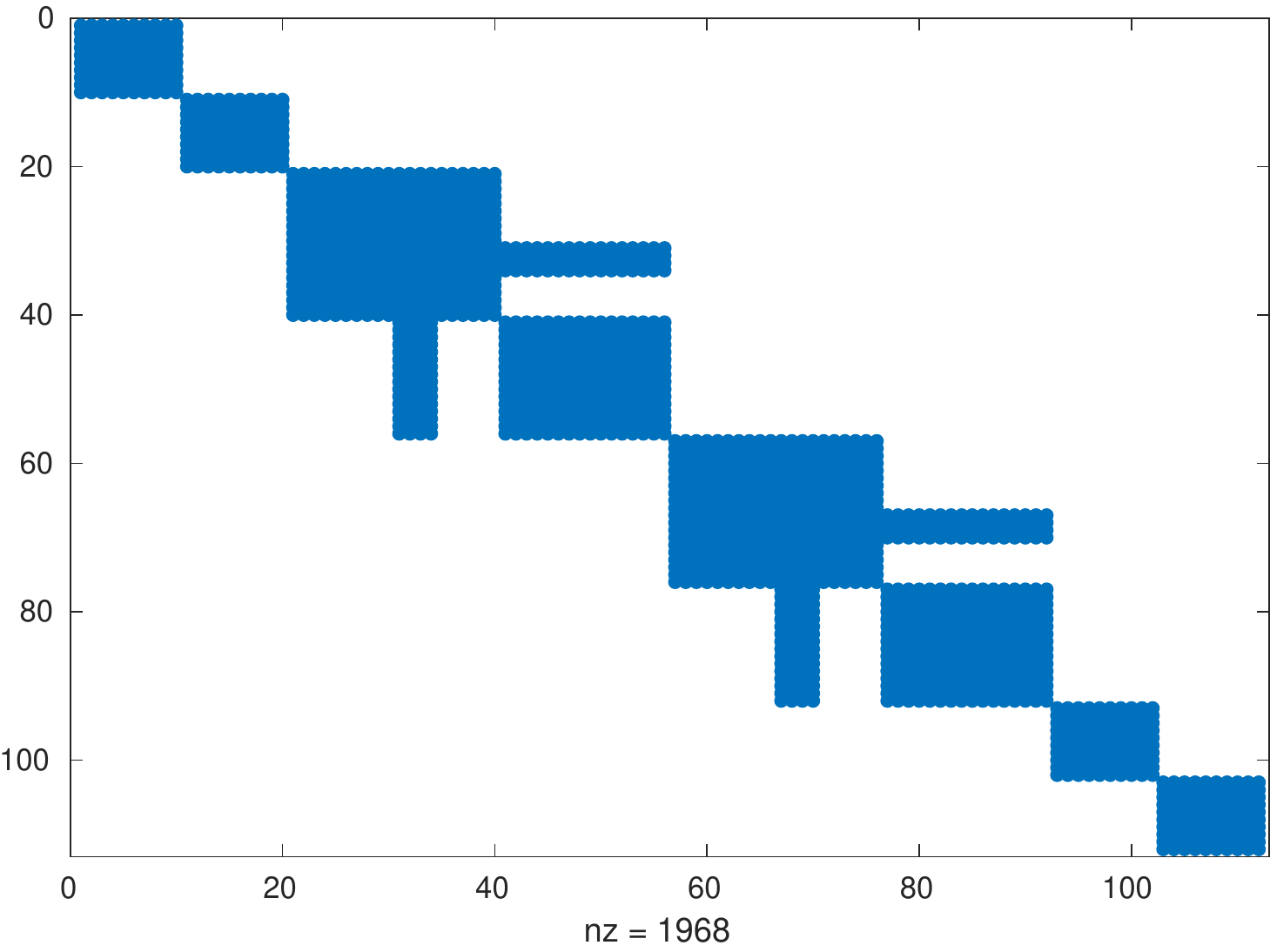}
    \includegraphics[width=0.25\textwidth,height=0.25\textwidth]{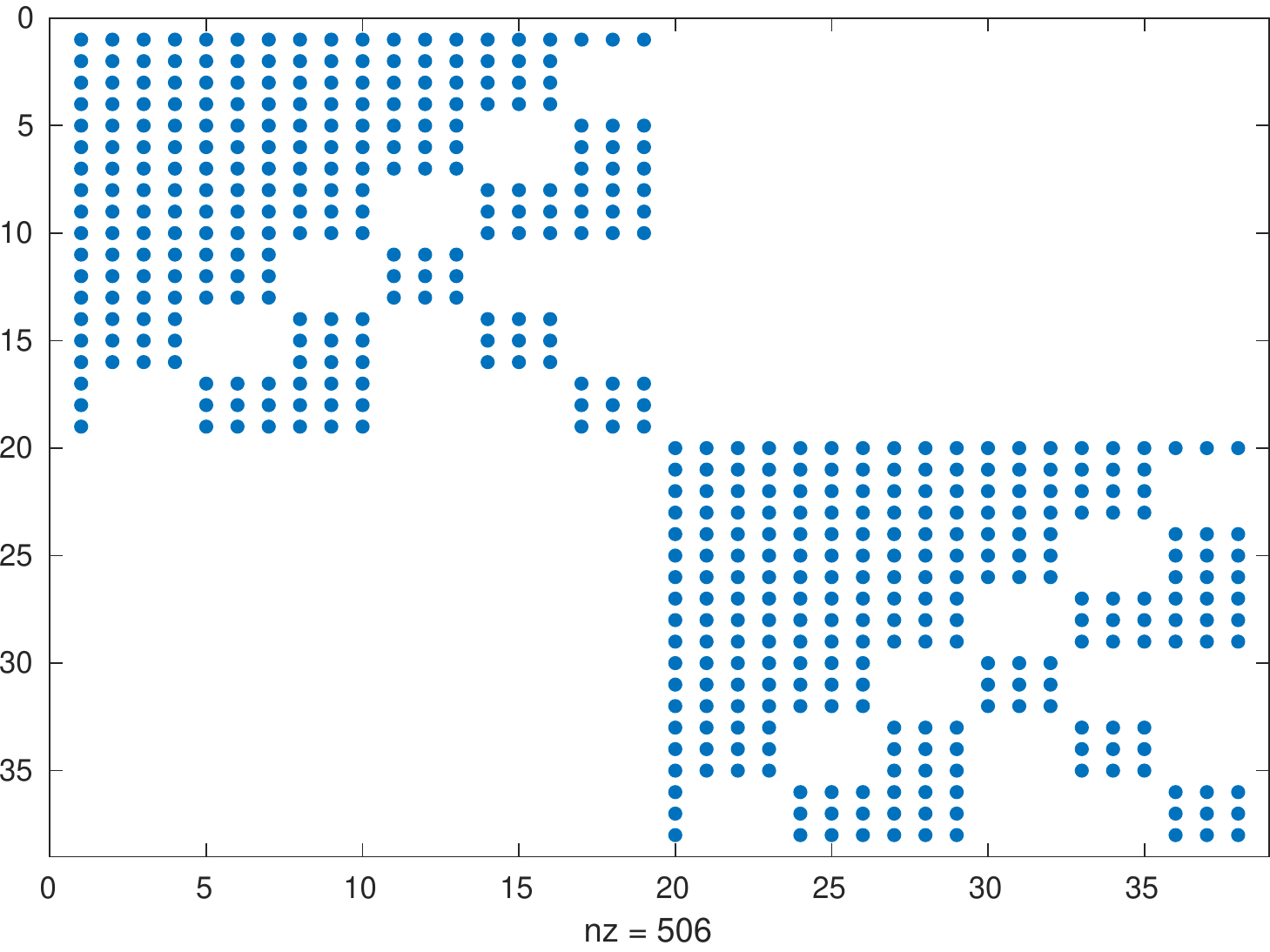}
    \caption{The sparsity pattern for the mass-matrices $\tilde{\mathbf{M}}_p^\varepsilon$ (second column) and $\mathbf{M}_p^\mu$ (third column) can also be studied under uniform $p$-refinement, i.e. the meshes remain unchanged in size, but the polynomial order is increased, namely we have $p=1,2,3$. The label $\mathrm{nz}$ again denotes the number of non-zero entries.}
    \label{fig:sparsity_p_check}
\end{figure}

The formula in (\ref{eq:weak_curl_algebra}) is arguably more impressive than (\ref{eq:mass_matrix_algebra}), since no dependence on the geometry of the mesh is left after all algebraic manipulations. More in detail, all non-zero entries in the $\mathbf{C}_p$ matrix are copies\footnote{up to orientation of edges and triangles, from which a very small number of equivalence classes for $\hat{\mathbf{C}}_p$ can be derived.} of entries of a local, entirely topological, template $\hat{M}\times\hat{N}$ matrix $\hat{\mathbf{C}}_p$, where
\begin{align*}
  & \hat{M} = \begin{pmatrix*}[c] p+2 \\2\end{pmatrix*} = \frac{(p+2)(p+1)}{2}, \;\;
    \hat{N} = 2\hat{M},
\end{align*}
\noindent are the dimensions of local scalar and vector shape-function spaces. As a consequence, with limited additional bookkeeping effort, no sparse and huge discrete curl-matrix needs to be stored in memory, and the ultra-weak curl operator can be applied efficiently, via its pre-computed low-storage representation, throughout the time integration of the problem. We remark that a similar result is achievable in more conventional DG formulations, as demonstrated in \cite{christophs, warburton_matrix_free, bk_js}, yet the fact that this feat can be achieved even when using the presented novel formulation on barycentric-dual complexes was highly non-trivial.

A final important hint on implementation of the proposed discrete formulation is motivated by the fact that we found, by direct computation, the following remarkable identity to hold:
\begin{align}
  & \int_{\hat{K}} \xi_1^r\xi_2^s\,\mathrm{d}\hat{\bm{r}} =
  \frac{B_{\frac{1}{3}}(r+1,s+1)}{2^{s+1}(r+s+2)} + \frac{B_{\frac{1}{3}}(s+1,r+1)}{2^{r+1}(r+s+2)},\label{eq:closed_form_ints}
\end{align}
\noindent for arbitrary non-negative integers $r$, $s$, where $B_{\alpha}(a,b)$ is the incomplete $\beta$-function (also known as the Euler integral of the first kind), defined as
\begin{align*}
  B_{\alpha}(a,b) = \int_0^\alpha z^{a-1}(1-z)^{b-1}\,\mathrm{d}z,
\end{align*}
\noindent the values of which can be computed to arbitrary precision and stored for all needed positive integers values of $a$, $b$ and for the particular value $\alpha=1/3$. Since all double integrals in the discrete formulation can be shown to reduce to linear combinations of terms equivalent to the l.h.s. of (\ref{eq:closed_form_ints}), no need for numerical integration arises (as long as the material parameters are piecewise-constant).

\subsection{The lowest order element and the cell method}
It is known that, both for conforming discretisations relying on finite elements of the Nedelec \cite{nedelec1980} type and for DG formulations based on central fluxes, the requirement is to have basis-functions which are piecewise-polynomials of degree $p$ and $p+1$ for the magnetic and electric field respectively (or vice-versa). This leads to sub-optimal convergence rates in the electromagnetic energy norm and also implies that, for the lowest admissible order, we need basis-functions which are piecewise-affine for one of the two unknown fields. For the proposed method instead, the two unknowns are approximated up to the same polynomial degree, and the lowest admissible one is $p=0$, i.e. piecewise-constant fields.
\begin{figure}[!h]
  \label{fig:low_order_incidence}
    \centering
    \begin{tikzpicture}[thick,scale=7.0, every node/.style={scale=7.0}]
      \coordinate (a) at (0.0,0.0);
      \coordinate (b) at (1/2,1/2);
      \coordinate (c) at (0.0,1/2);
      \coordinate (d) at (0.0,1/4);
      \coordinate (e) at (1/4,1/4);
      \coordinate (f) at (1/4,1/2);
      \coordinate (g) at (1/6,1/3);
      \coordinate (i) at (1/2,0.0);
      \coordinate (j) at (1/3,1/6);
      \coordinate (k) at (1/4,0.0);
      \coordinate (l) at (1/2,1/4);
      \node[scale=0.2] at (-1/20,-1/20) {$\bm{v}_1$};
      \node[scale=0.2] at (1/2+1/20,-1/20) {$\bm{v}_2$};
      \node[scale=0.2] at (1/2+1/18,1/2+1/18) {$\bm{v}_4$};
      \node[scale=0.2] at (-1/18,1/2+1/18) {$\bm{v}_3$};
      \node[scale=0.2,blue] at (1/8,-1/18) {$u_1$};
      \node[scale=0.2,blue] at (1/8+1/18,1/8) {$u_2$};
      \node[scale=0.2,blue] at (-1/18,1/8) {$u_3$};
      \node[scale=0.2,blue] at (3/8,-1/18) {$u_4$};
      \node[scale=0.2,blue] at (1/2+1/18,1/8) {$u_5$};
      \node[scale=0.2,blue] at (3/8+1/18,3/8) {$u_8$};
      \node[scale=0.2,blue] at (1/2+1/18,3/8) {$u_9$};
      \node[scale=0.2,blue] at (3/8,1/2+1/18) {$u_{10}$};
      \node[scale=0.2,blue] at (1/8,1/2+1/18) {$u_7$};
      \node[scale=0.2,blue] at (-1/18,3/8) {$u_6$};
      \node[scale=0.2,red] at (1/3+1/36,1/6+1/18) {$f_1$};
      \node[scale=0.2,red] at (1/6+1/18,1/3+1/36) {$f_2$};
      \draw[color=black,thick] (a) -- (b) -- (c) -- (a);
      \draw[color=black,thick] (a) -- (i);
      \draw[color=black,thick] (i) -- (b);
      \draw[color=black,dashed] (d) -- (g);
      \draw[color=black,dashed] (e) -- (g);
      \draw[color=black,dashed] (f) -- (g);
      \draw[color=black,dashed] (j) -- (e);
      \draw[color=black,dashed] (j) -- (k);
      \draw[color=black,dashed] (j) -- (l);
      \draw[-latex,color=red] (2/7,1/8) arc (-150:180:1mm) node[near start,left] {};
      \draw[-latex,color=red] (1/12,2/7) arc (-150:180:1mm) node[near start,left] {};
      \draw[draw,blue,->] (a) -- (1/8,0.0);
      \draw[draw,blue,->] (k) -- (3/8,0.0);
      \draw[draw,blue,->] (i) -- (1/2,1/8);
      \draw[draw,blue,->] (a) -- (1/8,1/8);
      \draw[draw,blue,->] (a) -- (0.0,1/8);
      \draw[draw,blue,->] (e) -- (3/8,3/8);
      \draw[draw,blue,->] (f) -- (3/8,1/2);
      \draw[draw,blue,->] (d) -- (0.0,3/8);
      \draw[draw,blue,->] (c) -- (1/8,1/2);
      \draw[draw,blue,->] (l) -- (1/2,3/8);
    %% \draw[->] (a) to node {}  (c);
  \end{tikzpicture}
    \caption{The $p=0$ method at work on a mesh consisting of two triangles $\mathcal{T}_1$ (with vertices $\bm{v}_1$,$\bm{v}_2$,$\bm{v}_4$) and $\mathcal{T}_2$ (with vertices $\bm{v}_1$, $\bm{v}_3$, $\bm{v}_4$).}
\end{figure}
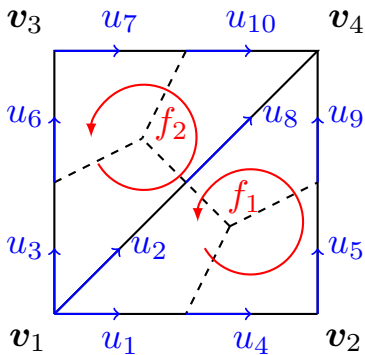

We may in fact take a closer look at the dimensions of the spaces in (\ref{eq:dimvec}) and (\ref{eq:dimscal}). We notice that, if we set $p=0$, we still have a non-empty basis. %, as also the definition of shape-functions requires.
Specifically, we are left with one degree of freedom (DoF) per triangle for the pseudo-vector ${H}^{h,0}(\bm{r},t)$ and two DoFs per each triangle edge in the primal complex for the vector field $\bm{E}^{h,0}(\bm{r},t)$.
A simple illustrating example is given in Fig.~\ref{fig:low_order_incidence} for a mesh consisting of two triangles $\mathcal{T}_1$, and $\mathcal{T}_2$: we have there ten DoFs for the electric field $u_{n=1,2,\dots,10}$ and two DoFs (one per triangle) for the magnetic field $f_{m=1,2}$. Their indexing is induced strictly from the ordering of vertices in the primal complex:  it is easy to prove that the $\bm{u}(t)$ DoFs are line-integrals of the electric fields along edges $e\in\mathcal{K}_1^\Omega\cap\mathcal{S}(\mathcal{C}^\Omega)$ while the $\bm{f}(t)$ DoFs are fluxes of the $B^{h,0}(\bm{r},t) := \mu {H}^{h,0}(\bm{r},t)$ pseudo-vector field across the triangles $\mathcal{T}\in\mathcal{K}_2^\Omega$. Given an edge $e\in\{\mathcal{K}_1^\Omega\cap\mathcal{S}(\mathcal{C}^\Omega)\}$, it holds in fact:
\begin{align*}
\int_{e} \bm{E}^{h,0} (\bm{r},t)\cdot \hat{\bm{t}}(\ell)\,\mathrm{d}\ell &= 
\int_{e} u_e(t)\bm{w}_{e}^0(\bm{r})\cdot \hat{\bm{t}}(\ell)\,\mathrm{d}\ell = \nonumber\\
&= u_e(t) \int_{e}  C_{00l} \left(\mathbf{A}_K^{-\mathrm{T}} \hat{\nabla}{\xi}_{l} \right)\cdot \hat{\bm{t}}(\ell)\,\mathrm{d}\ell=\nonumber\\
&= u_e(t),
\end{align*}
\noindent where we abuse the notation again by using $e$ both as an index and as the integration domain, and $K \in\mathcal{K}_2^\Omega$ s.t. ${e} \subset \{\partial{K}\}, l = l(e,K)$. The fact that only one basis-function has non-null tangential component on the edge $e$ was also exploited, by setting $C_{00l}=1/|e|$. Very similar steps are easily computable for the magnetic field approximation ${H}^{h,0}$.

The discrete operator $\mathbf{C}_0$ and its transpose are instead exactly incidence matrices: this can be proved by direct computation of (\ref{eq:weak_curl_algebra}) where, as the reader may notice, the double integrals on the kite vanish identically (since $p=0$) but the line-integrals do not. 
The left-multiplication of DoFs vectors with the incidence matrix $\mathbf{C}_0 \bm{u}$ equates to the Stokes theorem: 
\begin{align*}\int_{\mathcal{T}_m} \partial_t(\mu{H}^{h,0})\,\mathrm{d}\bm{r} = \oint_{\partial\mathcal{T}_m} \bm{E}^{h,0}\cdot\hat{\bm{t}}(\ell)\,\mathrm{d}\ell, & \;\;\;m=1,2.\end{align*}

If PEC boundary conditions are enforced, we note that the number of unconstrained DoFs becomes equal (to two) for both unknowns $\bm{E}^{h,0}$ and ${H}^{h,0}$ (as predicted, for example, in \cite{he_tex_dofs}). We finally again stress that the $\bm{E}^{h,0}(\bm{r},t)$ field is allowed to be fully discontinuous on the dashed dual edges.

We remark that this is exactly the Cell Method's framework advocated by Tonti\cite{tonti,marrone}, while also being a generalization of the Yee algorithm to unstructured meshes. The peculiarity of having to split each $E\in\mathcal{C}_1^\Omega$ into two segments (while still preserving the physical interpretation of DoFs) is also not new, but was instead studied by some of the authors in the most general 3D setting in \cite{codecasa_politi,dgatap,dgagpu}, where tetrahedral meshes are used. In fact the (covariantly mapped) function 
\begin{align*}\mathbf{A}_K^{-\mathrm{T}} \hat{\nabla}{\xi}_{l}\end{align*}
is a one-form which coincides exactly with the basis-functions introduced in \cite{codecasa_politi} directly in the global coordinates.
As anticipated in the introduction, a recent equivalence proof (in \cite{dga_as_dg}) between the lowest order 3D CM and a DG approach was a leading cause for the present developments. 

%% We finally remark that the construction set forth in the present section is in fact not new, since it is a two dimensional restriction of the bases of Codecasa et al. (see \cite{codecasa_politi,coddy_bases,dgatap}). The cited works are invariably set in $\mathbb{R}^3$, but the equivalence of the dofs definition and the topologic nature of the curl operator are easily seen by taking the simplicial complex $\mathcal{C}_\Omega$ of the present paper as a $2$D restriction of a 3D mesh comprising prisms with a triangular basis and a uniform height $h_z=1$ in the $\hat{\mathbf{z}}$ direction.
%% We can now define $2\!|\mathcal{C}_1^\Omega\!|$ 

\section{Numerical Results}\label{sec:num}
We shall here validate the method through numerical experiments. All computations are in natural units, i.e. physical units have been rescaled such that the speed of light in a vacuum is normalized to one, which means in practice $\varepsilon = \varepsilon_r$ and $\mu = \mu_r$.
\subsection{MIBVP results}
To test the transient behaviour of the method we use a manufactured time domain problem, with solution already available in closed form in \cite{dgatap}, where the computational domain is a waveguide $\Omega = [0,1]\!\times\![0,2]$. The fundamental propagating mode\footnote{the waveguide-mode with the lowest cut-off frequency.} is enforced as a time-dependent boundary condition at the entrance $y=0$ of the waveguide, while all other segments in $\partial\Omega$ are set to PEC.
To simulate an invariant structure in the $z$ direction (a needed assumption for true 2D problems) we need a transverse-electric (TE) mode, i.e. only one component of the electric field is not identically zero. It is very convenient for the purpose to swap the field approximation spaces with respect to the theory and make the $\bm{E}$ field a pseudo-vector. This poses no real hardships, as the input field can be injected as an equivalent magnetic current by projecting it on the vector-valued trial-space (which requires 1D Gauss integration on mesh edges at $y=0$).
\begin{figure}[!h]
  \centering
  \includegraphics[width=1.0\textwidth]{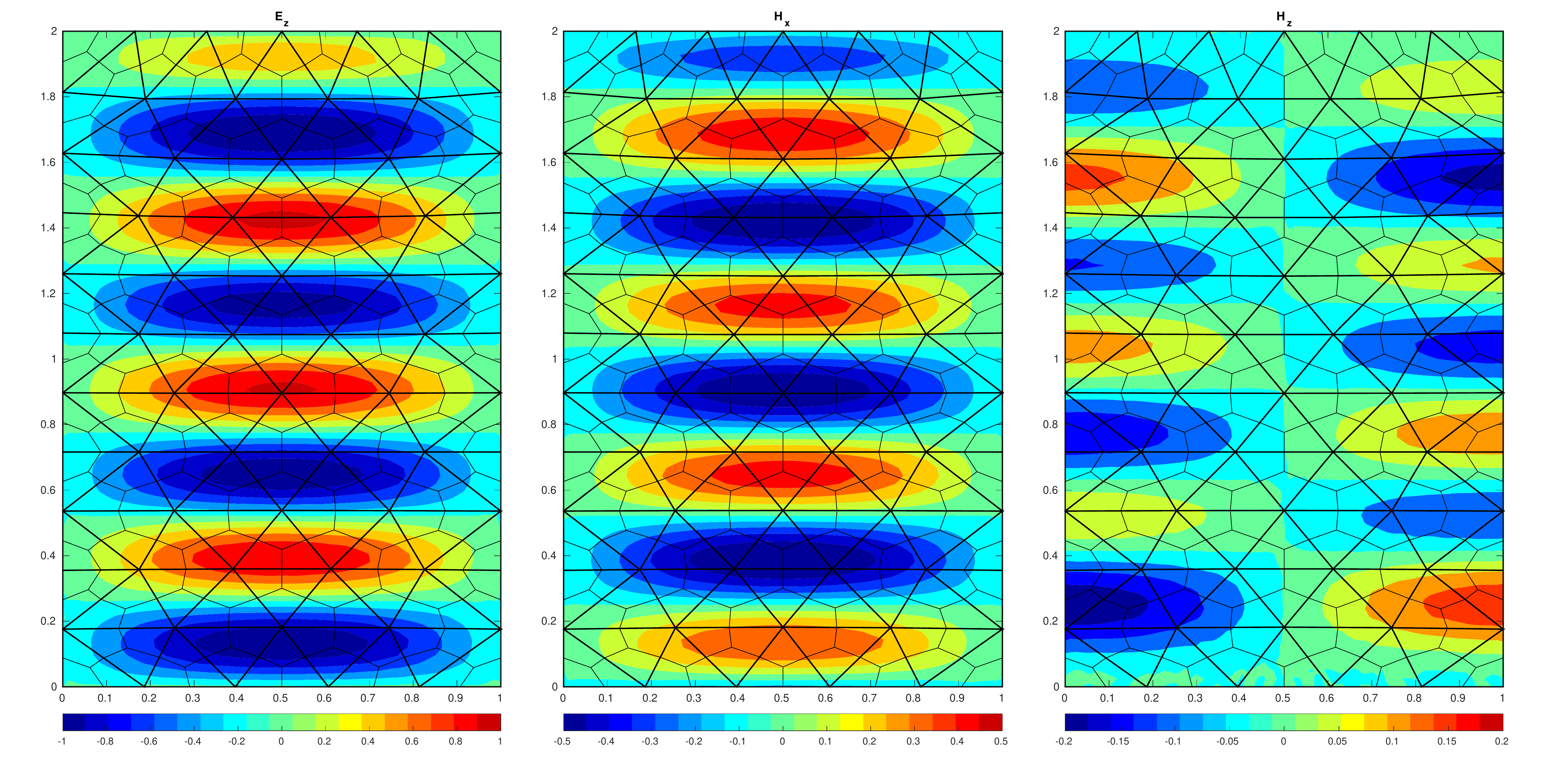}
  \caption{The transient field in the waveguide at $t=2$.}
  \label{fig:contourf}
\end{figure}

The behaviour in the whole waveguide for the three non-zero components of the electromagnetic field is shown in Fig.~\ref{fig:contourf}, for polynomial degree $p=5$, average mesh size $h=0.2$ and at time $t=2$ (again in natural units).
Due to the reflections at $y=2$, the $z$-aligned field is not everywhere continuously differentiable in time. The presence of critical points in the temporal behaviour is visible in Fig.~\ref{fig:td_full}, which shows how the various polynomial orders behave for the same mesh, chosen to be rather coarse with a maximum mesh size of $h=0.2$. All polynomial degrees in the bases are tested with the leap-frog time-stepping scheme using $\tau$ equal to the upper limit for stability (the usual practical choice). The qualitatively better approximation properties of the higher order versions of the method are clearly visible.

We do not make any claim to have programmed the fastest possible version of the method, yet it is useful to remark that for $p=6$ and the given mesh size $h$, the method requires $21\,472$ DoFs ($6\,464$ for the scalar-valued unknown, $15\,008$ for the vector-valued one), the maximum allowed time-step is $\tau=3.833\times 10^{-3}$, and the computation reaches a yield of $129.633$ time-steps per second (in wall-time, averaged over simulations with $10^5$ time-steps) on a modest laptop computer (Intel Core i7-6500U CPU, clocked at 2.50GHz with 4 physical cores, 8 GB of RAM), which amounts to roughly $2.783\times 10^6$ DoFs/second of average performance.

%\begin{minipage}{0.49\textwidth}
\begin{figure}[h!]
  \centering
    \centering
    \includegraphics[width=0.49\textwidth]{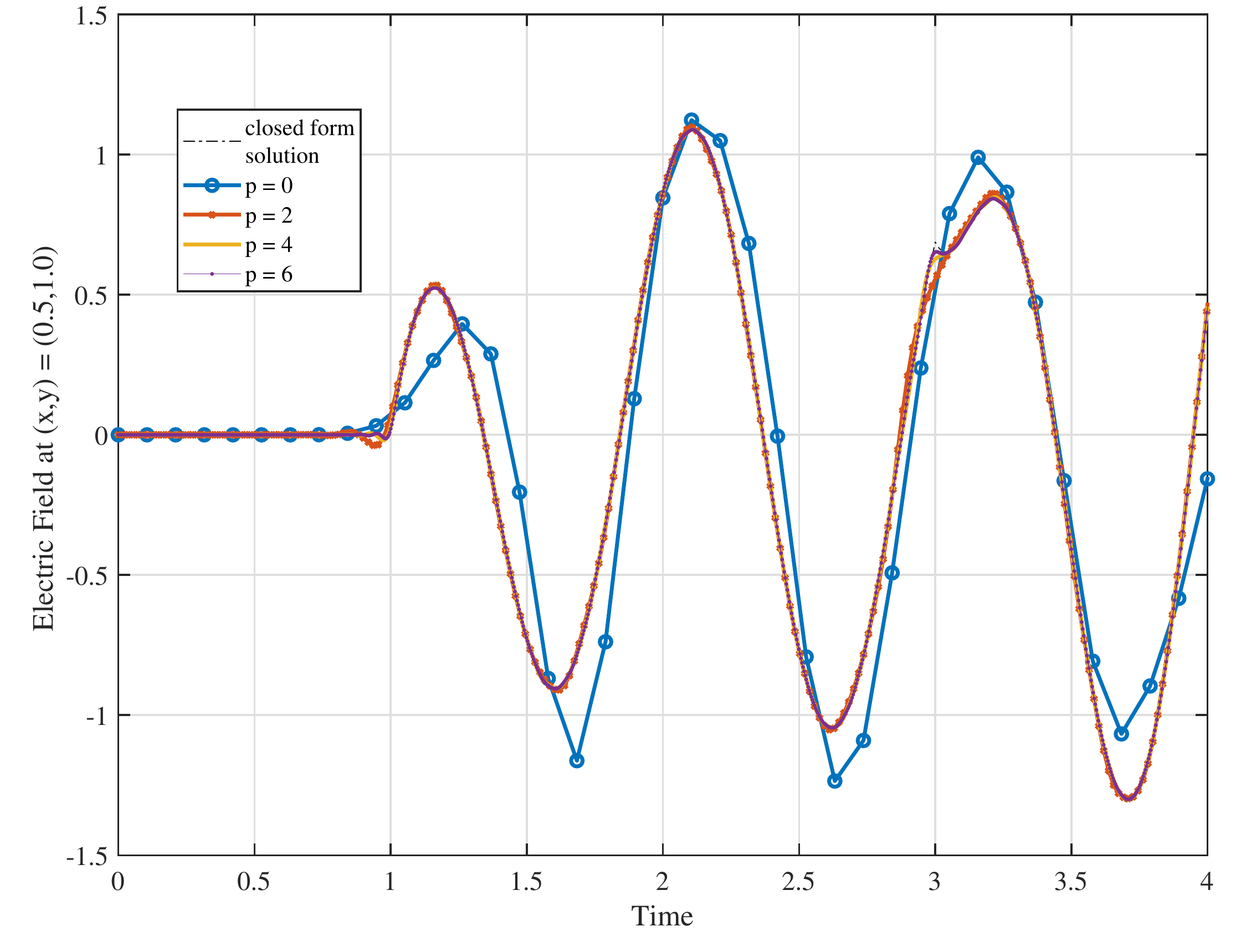}
    \includegraphics[width=0.49\textwidth]{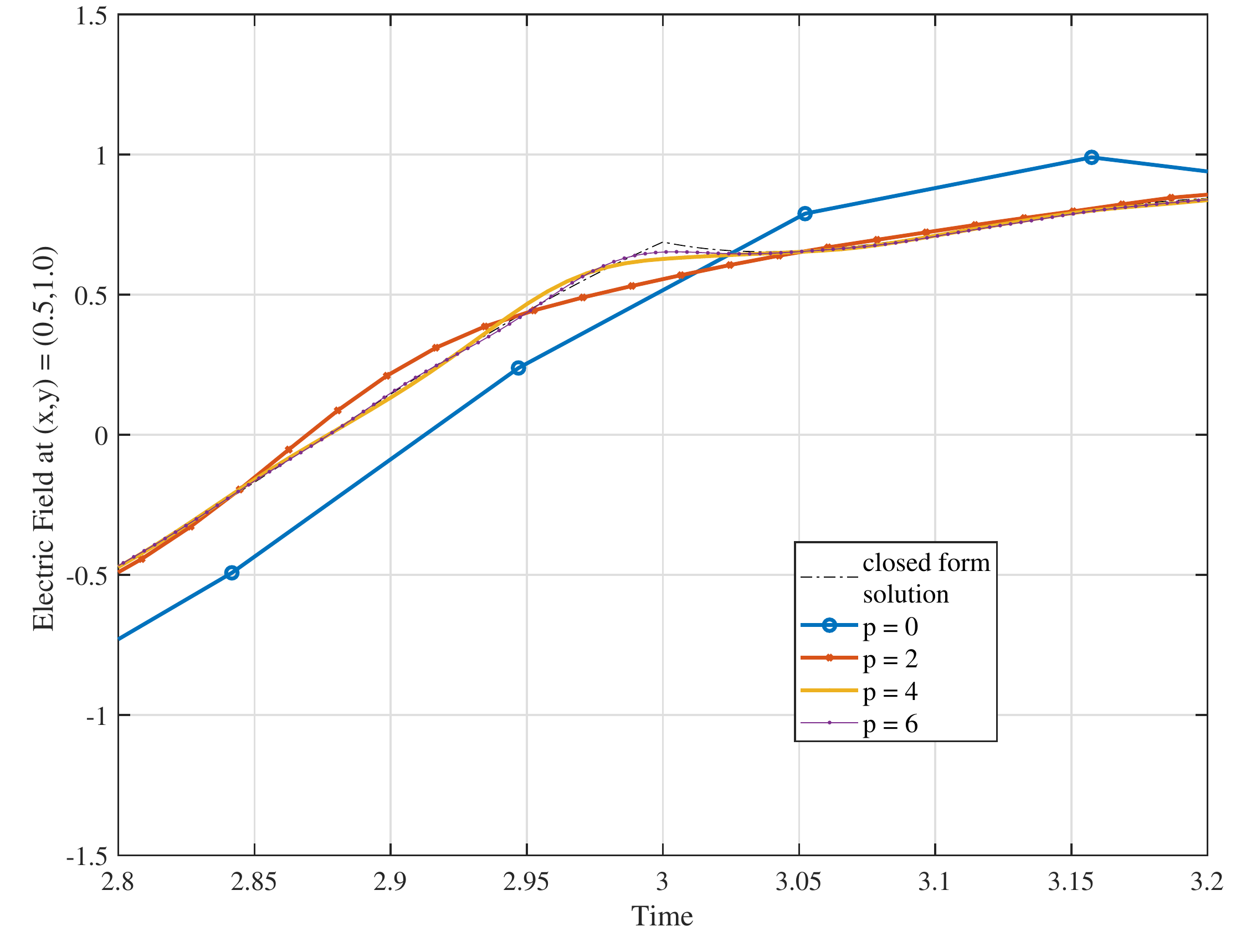}
  \caption{The left panel shows the time-dependent solution to the waveguide problem. The field is measured at the centroid of the domain, as can be inferred by the delay in the propagation at the start. In the right panel a blow-up of a small interval around $t = 3.0$ is shown, where a critical point in the solution must be approximated.}
  \label{fig:td_full}
\end{figure}
\subsection{Spectral accuracy}
Due to low regularity of the true solution for the transient waveguide problem, we cannot expect to observe the theoretical order of convergence for the method.
A good way to assess the superiority in terms of approximation properties when using higher order basis-functions is to use the proposed method to solve an associated generalized eigenvalue problem. In fact, since we are using a kind of discontinuous Galerkin approach, the spectral accuracy of the method is interesting in its own right, and not just as a mean to study convergence, since we have no formal guarantee for the absence of spurious modes, which would tarnish the appeal of any new numerical method. By acting directly on the semi-discrete system of (\ref{eq:semidiscrete_cmp}) and making it time-harmonic ($\partial_t \mapsto -i\omega$, where now $i=\sqrt{-1}$) we arrive at the two following ``dual'' formulations:
\begin{align}
  & \mathbf{C}_p (\mathbf{M}_p^\varepsilon)^{-1} \mathbf{C}_p^{\mathrm{T}} \hat{\mathbf{f}} = \lambda \mathbf{M}_p^\mu \hat{\mathbf{f}},\label{eq:mevp}\\
  & \mathbf{C}_p^{\mathrm{T}} (\mathbf{M}_p^\mu)^{-1} \mathbf{C}_p \hat{\mathbf{u}} = \lambda^* \mathbf{M}_p^\varepsilon \hat{\mathbf{u}},\label{eq:eevp}
\end{align}
\noindent where the hat super-script denotes the time-harmonic solutions and $\lambda$, $\lambda^*$ are the squared eigenfrequencies. Depending on boundary conditions, (\ref{eq:mevp})--(\ref{eq:eevp}) approximate either the Dirichlet MEP or the Neumann one. 
We choose to work with (\ref{eq:mevp}) since the $H^{\bm{curl}}(\Omega)$ space, as defined in Section \ref{sec:funspaces}, coincides in two dimensions with the standard Sobolev space $H^1(\Omega)$, i.e. we are basically approximating the Laplace operator with the matrix $\mathbf{C}_p (\mathbf{M}_p^\varepsilon)^{-1} \mathbf{C}_p^{\mathrm{T}}$.

As a first example we take the unit square domain $\Omega=[0,1]\times[0,1]$ with uniform material coefficients $\mu_r=1$ and $\varepsilon_r=1$. In this case the eigenvalues are of the form $\lambda = (a^2+b^2)\pi^2$ where $a,b \in \mathbb{N}^+$ for Dirichlet boundary conditions on the $H$ field and $a,b\in\mathbb{N}_0$ for Neumann ones.

\begin{figure}[!h]
  \centering
  \includegraphics[width=0.65\textwidth]{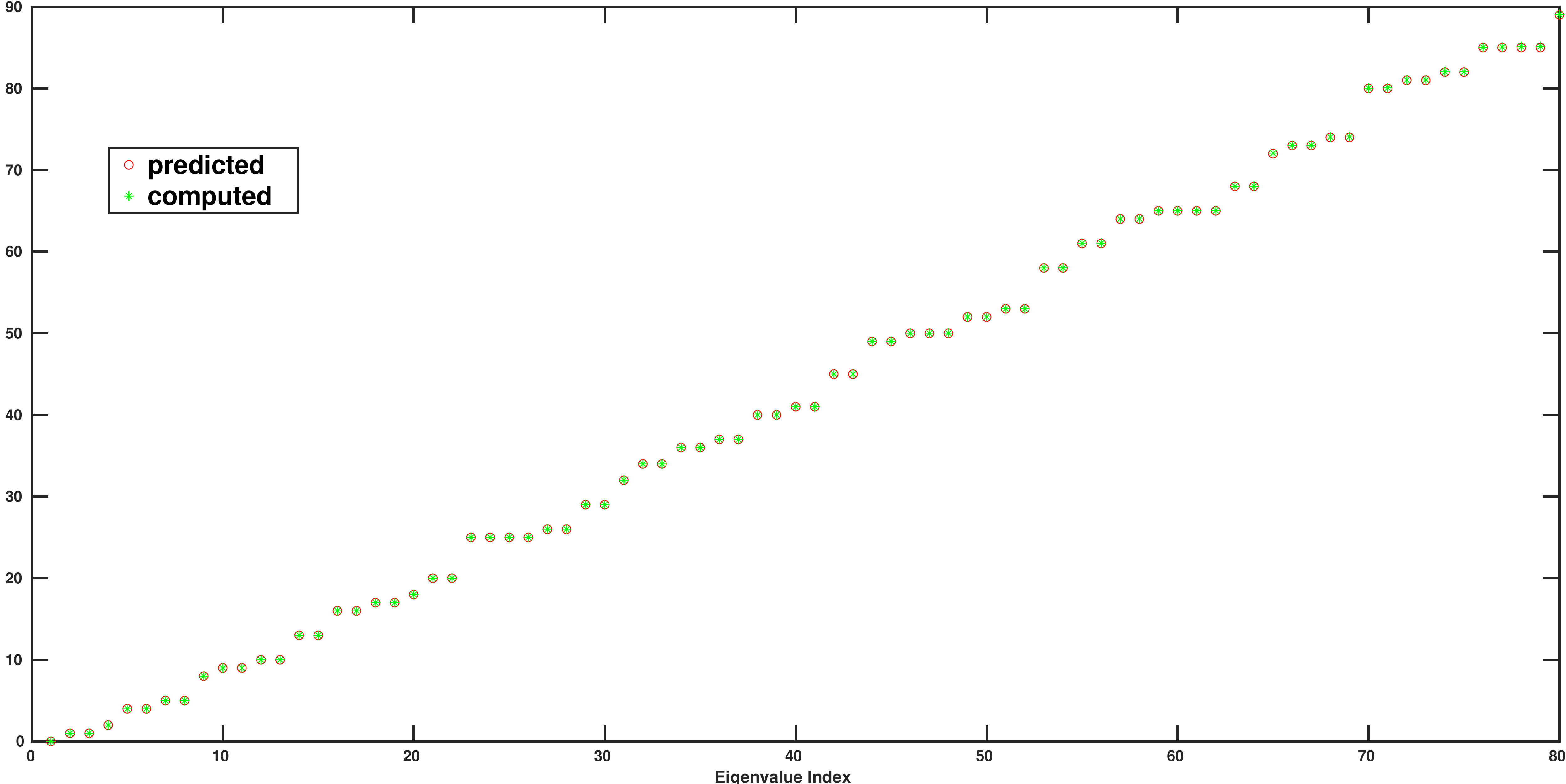}\\
  \includegraphics[width=0.65\textwidth]{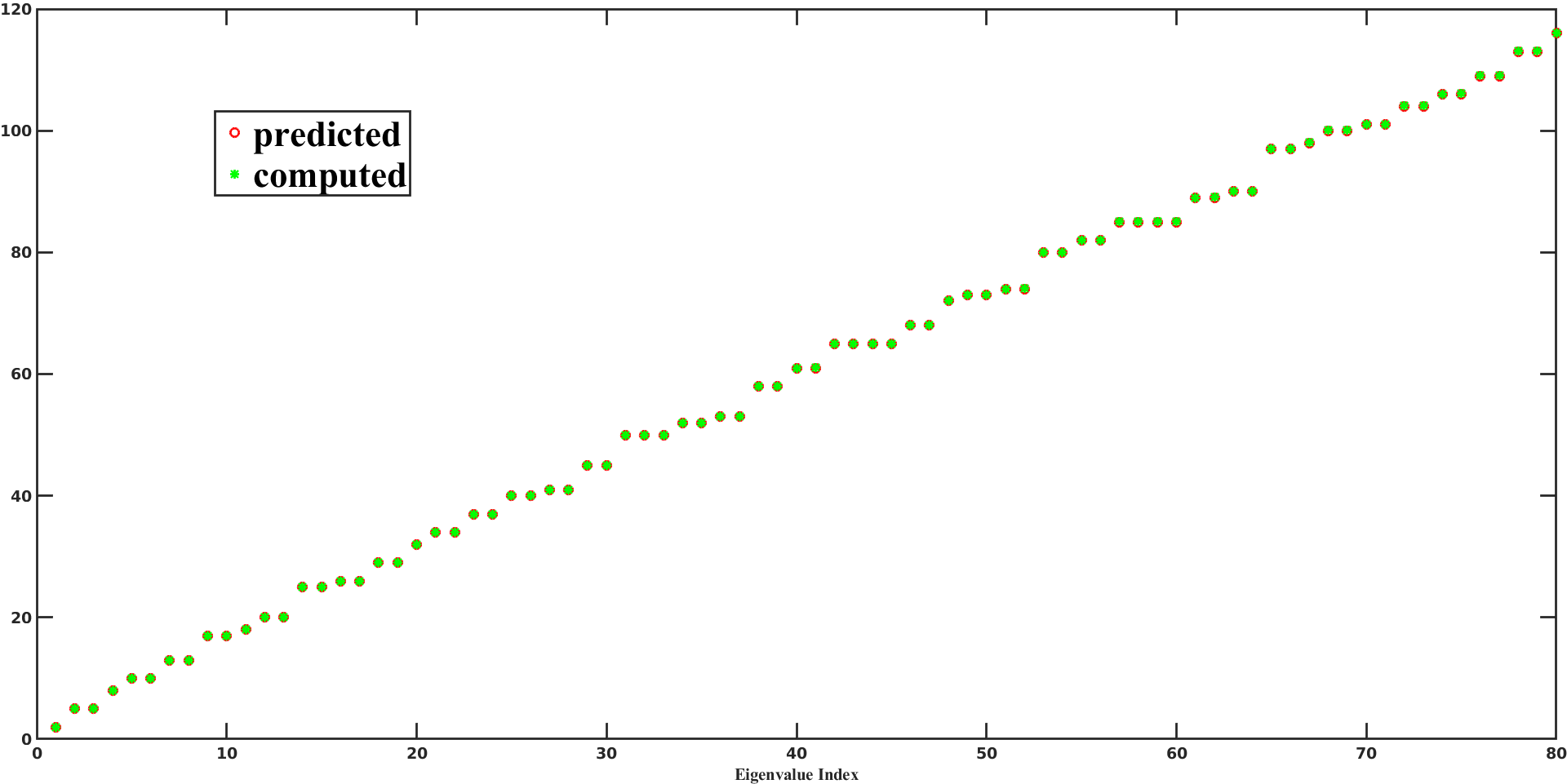}
  \caption{We show spectral correctness of the proposed method when solving the generalized eigenvalue problem in (\ref{eq:mevp}) for Neumann (note the one zero eigenvalue in the upper panel), and Dirichlet boundary conditions in lower panel. Here $\mu_r=\varepsilon_r=1$ holds on the whole domain $\Omega$.}
  \label{fig:eigvals_uniform_material}
\end{figure}
%\begin{sidewaysfigure}
\noindent Fig.~\ref{fig:eigvals_uniform_material} shows the first 80 eigenvalues (all scaled by $\pi^2$) for both cases, computed with $p=4$ and $h=0.2$ using the \textbf{eigs} function \cite{matlab_eigs} in MATLAB. No spurious eigenvalues appear (we note there is exactly one zero eigenvalue for the Neumann problem). Thorough testing for all $p<8$ and various mesh sizes confirms the absence of spurious eigenmodes due to the method. The accuracy is quite impressive for the shown test, for which we also present the first ten computed eigenfunctions in Fig.~\ref{fig:eigfuns_uniform_material}, where we note that, when the associated eigenvalue has algebraic multiplicity bigger than one, we cannot easily force the chosen solver to yield the appropriate mutually orthogonal eigenfunctions instead of a pair of their linear combinations.
\begin{figure}[h!]
    \centering
  \includegraphics[width=0.48\textwidth]{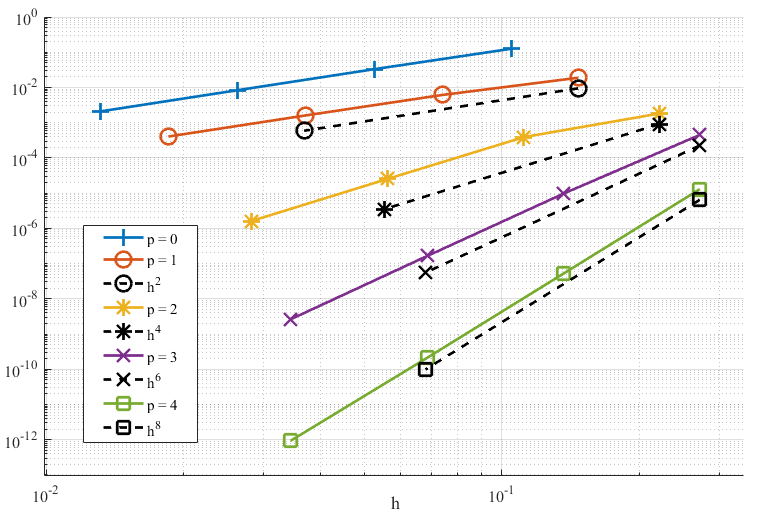}
  \includegraphics[width=0.48\textwidth]{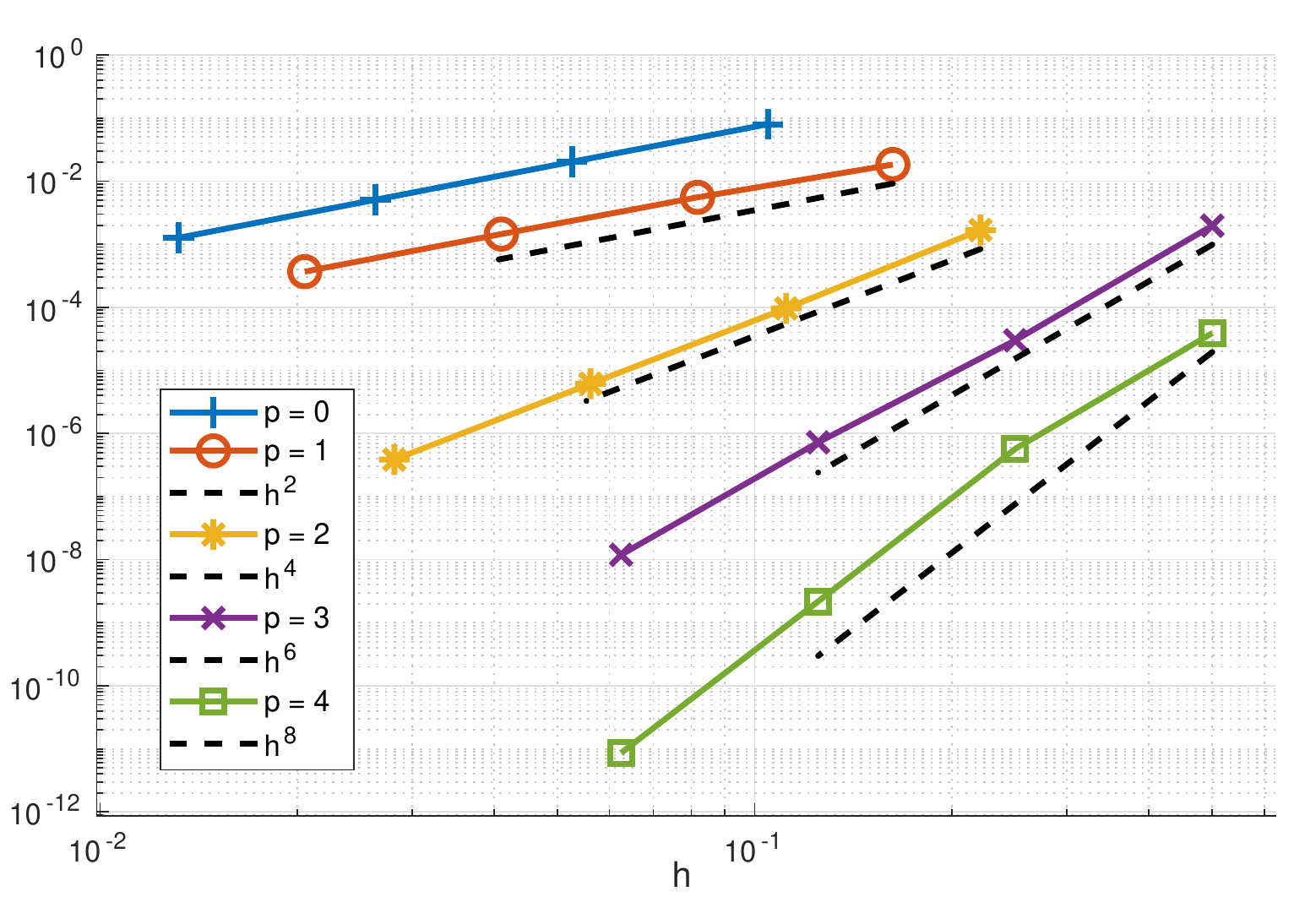}
  \caption{The error in approximating the 16th eigenvalue of the Dirichlet (left panel) and Neumann (right panel) problem with respect to the mesh size $h$ vanishes with the expected rate for the various tested polynomial orders $p$ for the case of a uniformly filled cavity.}
  \label{fig:h_conv_uniform_material}
\end{figure}
A more formal study of convergence is shown in Fig.~\ref{fig:h_conv_uniform_material}, which reveals $\mathcal{O}(h^{2p})$ convergence when polynomial degree $p$ is used and the mesh-size $h$ vanishes. This has been found to hold for the eigenvalues of both generalized problems (\ref{eq:mevp})--(\ref{eq:eevp}). The obtained rates are in perfect agreement with the theoretical studies of Buffa \& Perugia in \cite{perugia_buffa} for DG methods. We nevertheless stress that the analysis therein relies on the introduction of (mesh and polynomial degree dependent) penalty parameters, which should be big enough to ensure coercivity of the bilinear form on the l.h.s. of the weak formulation of the eigenvalue problem. No free parameters are instead present in the herein proposed formulation. 
\begin{figure}[h!]
  \centering
  \includegraphics[width=0.9\textwidth]{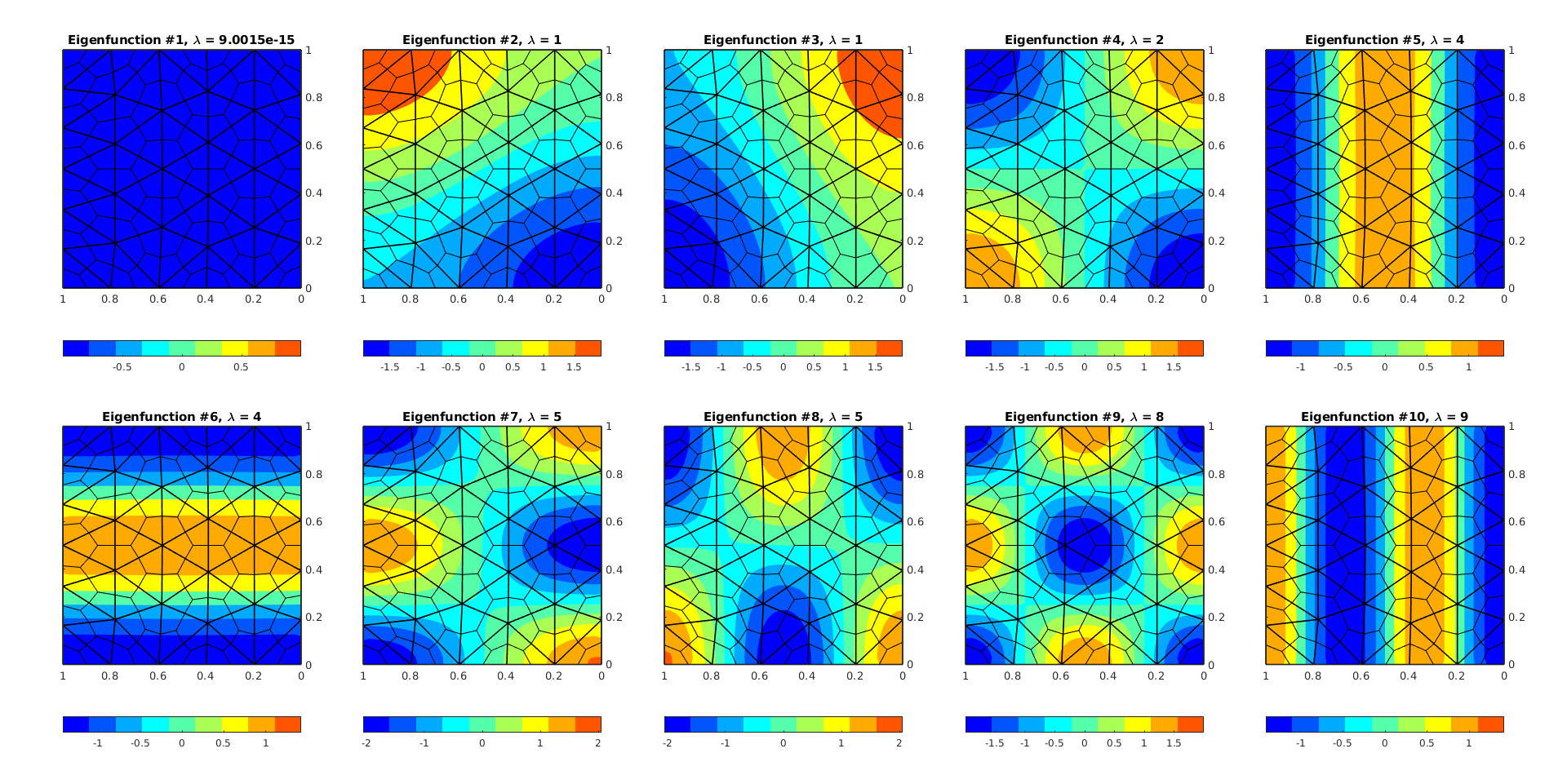}\\
  \includegraphics[width=0.9\textwidth]{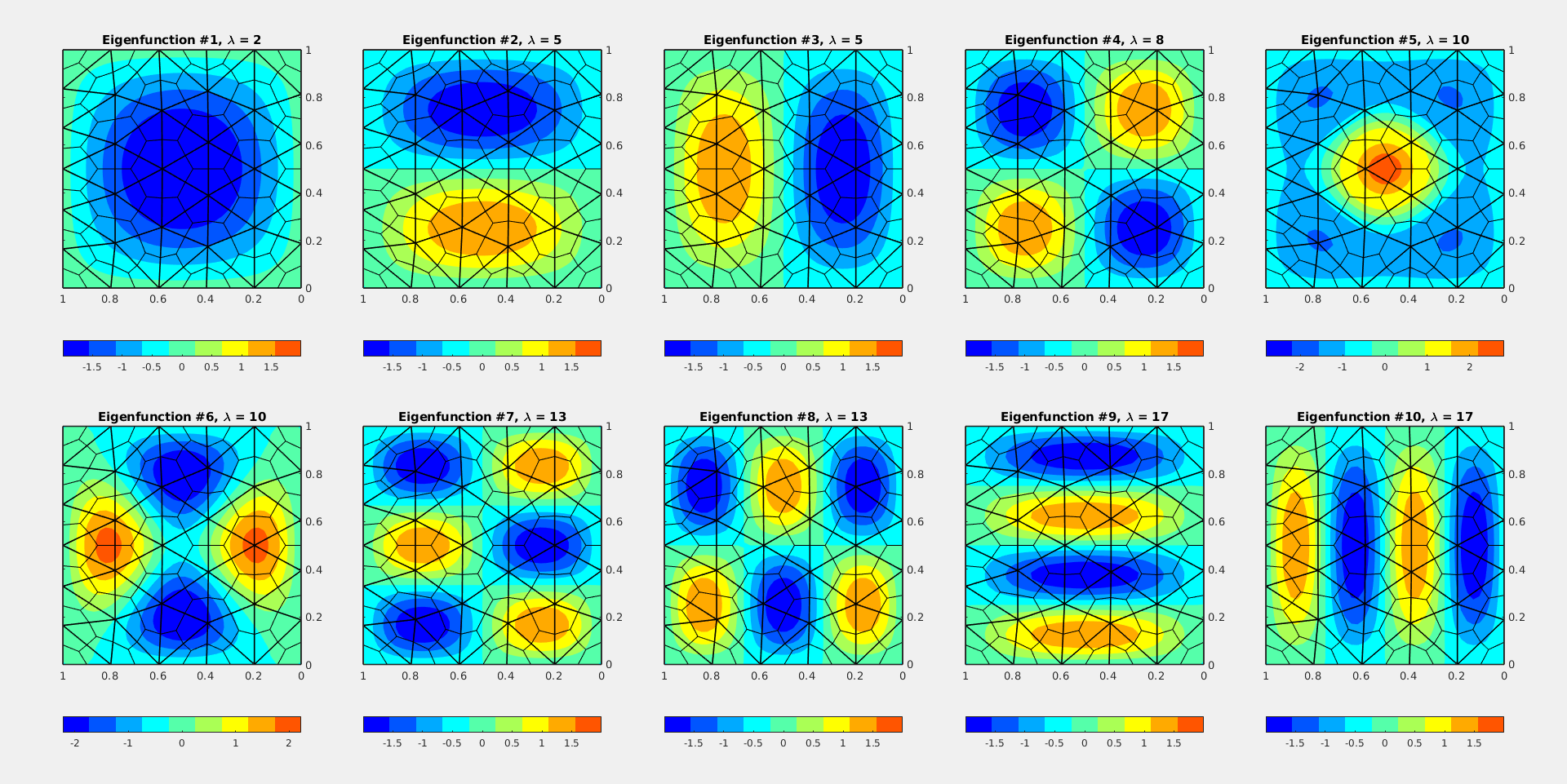}
  \caption{The first ten computed eigenfunctions for the case of the uniformly filled cavity: Neumann and Dirichlet case.}
  \label{fig:eigfuns_uniform_material}
\end{figure}
We furthermore remark that the $p=0$ version of the method shows $\mathcal{O}(h^{2p+2})$ convergence rate, but this super-convergence phenomenon is not translated to higher polynomial degrees, at least for the proposed sets of basis-functions. This fact clearly begs for further theoretical investigation.

As a more testing setup, we split our square $\Omega$ exactly into two halves, with a discontinuity aligned with the $y$ axis. We fill the left half of the cavity $\Omega_1 =[0,1/2]\!\times\![0,1]$ with a higher index material $\varepsilon_1=4$, which corresponds to halving the speed of light with respect to the vacuum parameters, which we keep intact on $\Omega_2 = \Omega\setminus\Omega_1$.
The exact values of the Neumann eigenvalues are not easily computable with pen and paper any more, as one needs to solve a transcendental equation (see \cite{pincherle}) involving hyperbolic functions. Yet, using any symbolic mathematics toolbox, we can estimate their values with arbitrary precision.
\begin{figure}[h!]
  \centering
  \includegraphics[width=\textwidth]{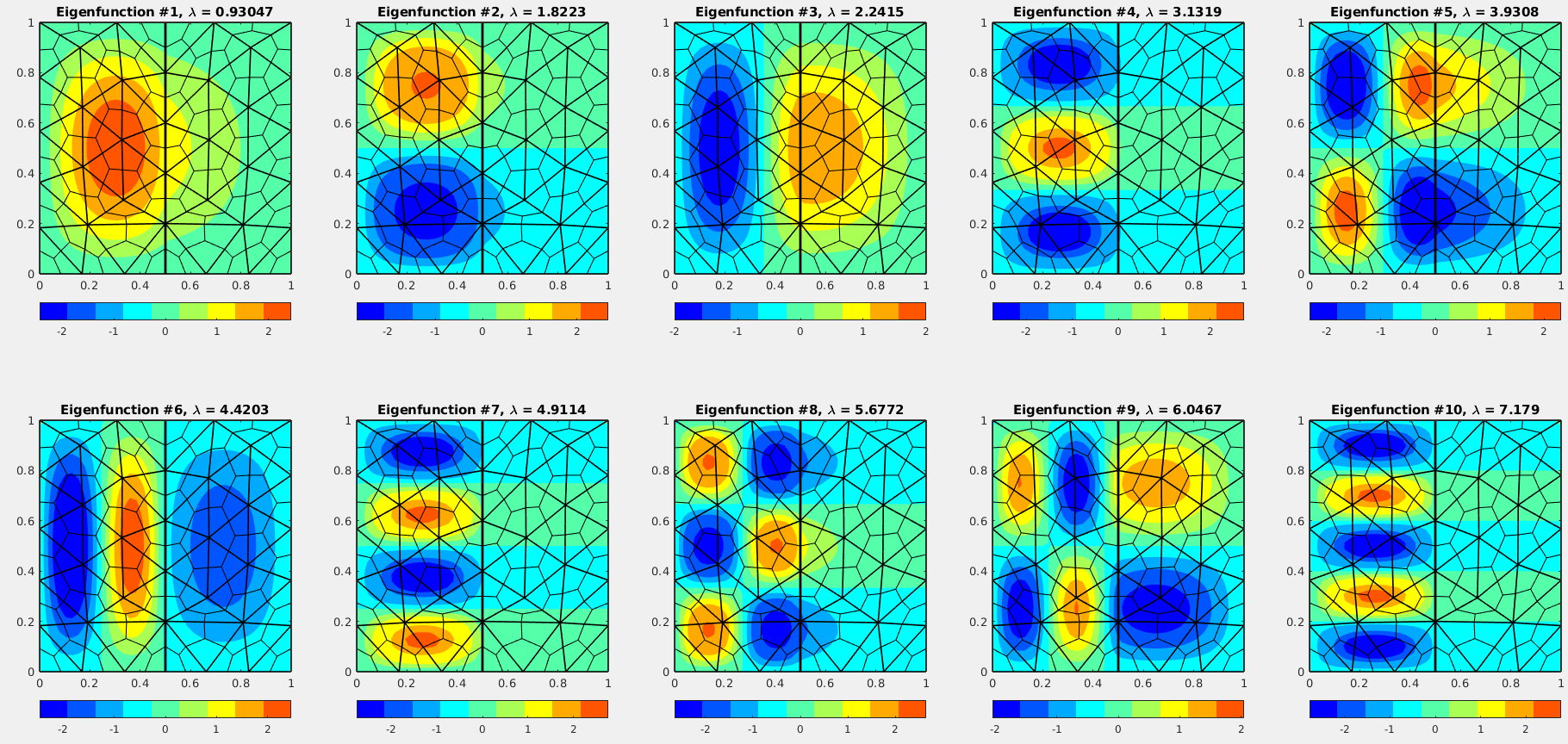}
  \caption{Eigenfunctions for the cavity with discontinuous permittivity.}
  \label{fig:eigfuns_2materials}
\end{figure}
We show the first ten eigenfunctions we computed (again with $p=4$ and $h=0.2$) as a reference in Fig.~\ref{fig:eigfuns_2materials}, where we stress the fact that ``partially evanescent'' modes are clearly visible: more formally these are modes with real wave-number $k = \left(k_x^2 + k_y^2\right)^{\frac{1}{2}}$ (due to the positive-definiteness property of the Laplace operator) but imaginary $k_x$.
\begin{figure}[h!]
  \centering
  \includegraphics[width=0.45\textwidth]{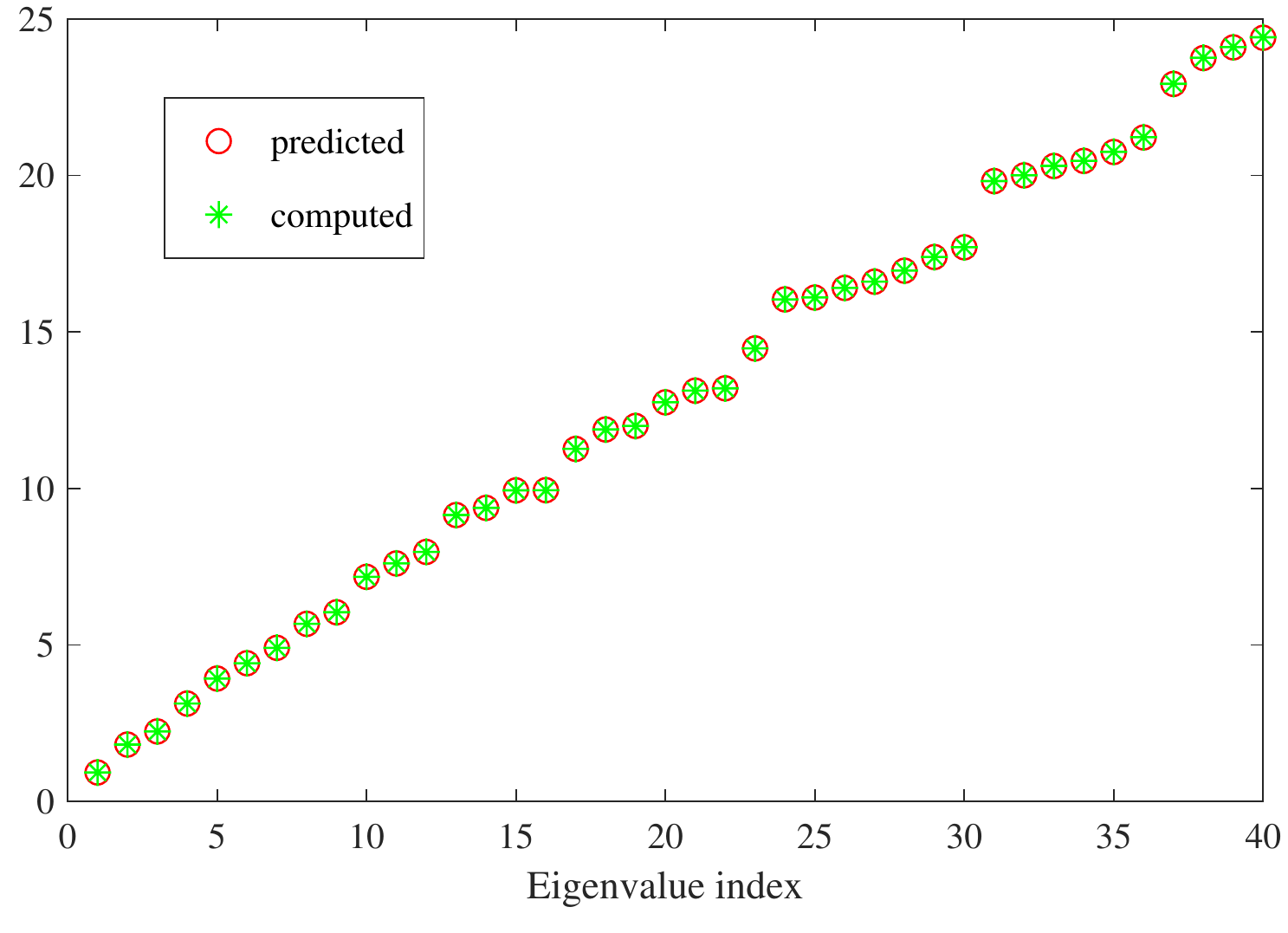}
  \includegraphics[width=0.50\textwidth]{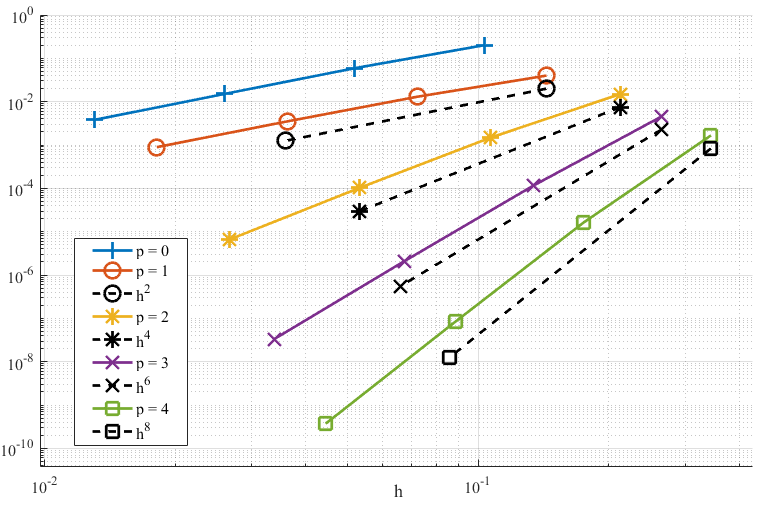}
  \caption{The method remains spectrally correct when approximating the Neumann problem for the inhomogenously filled square. The error in approximating the 20th eigenvalue is shown on the right to also still vanish with the optimal rate, with respect to the mesh size $h$, for the various tested polynomial orders $p$.}
  \label{fig:h_conv_2materials}
\end{figure}
This behaviour is confirmed by the distribution of eigenvalues in Fig.~\ref{fig:h_conv_2materials} (leftmost panel, which again shows no spurious solutions), where the eigenvalues are shown to be perturbed closer together towards zero. The optimal order of convergence with varying polynomial degree is also again confirmed for the discontinuous material case in Fig.~\ref{fig:h_conv_2materials} (right panel). 
\begin{figure}[h!]
  \centering
  \includegraphics[width=0.75\textwidth]{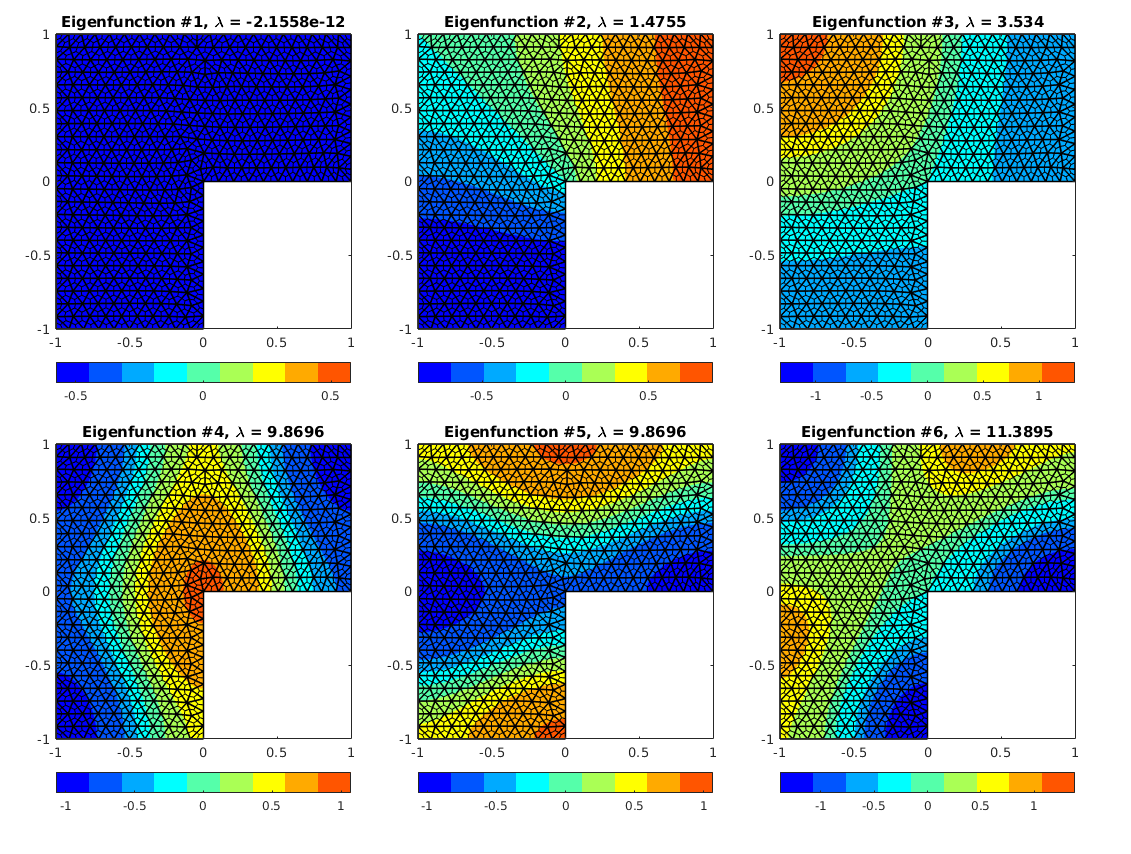}
  \caption{The first six eigenfunctions for the $L$-shape domain.}
  \label{fig:eigfuns_lshape}
\end{figure}
\begin{figure}[!h]
  \centering
  \includegraphics[width=0.48\textwidth]{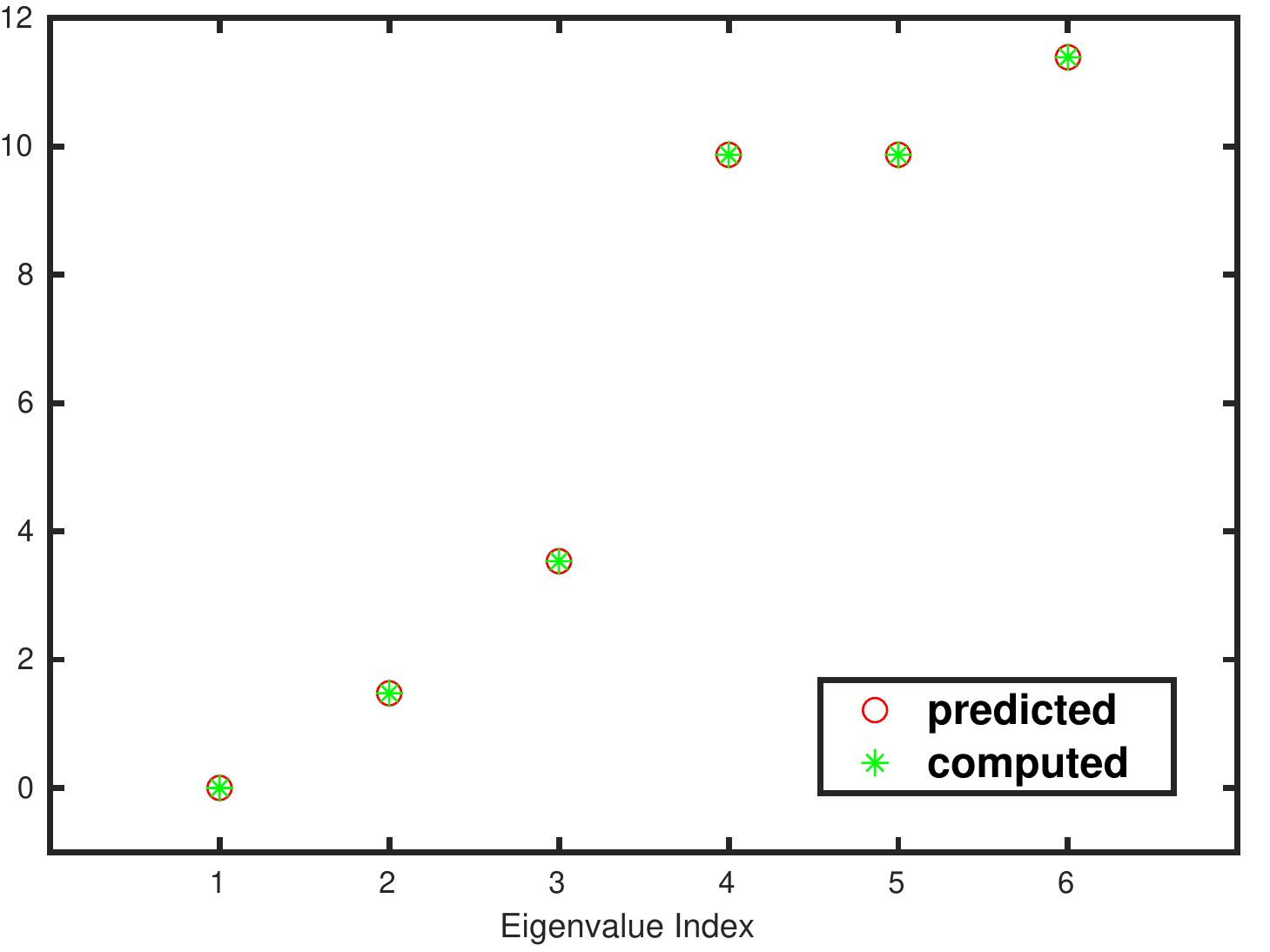}
  \includegraphics[width=0.48\textwidth]{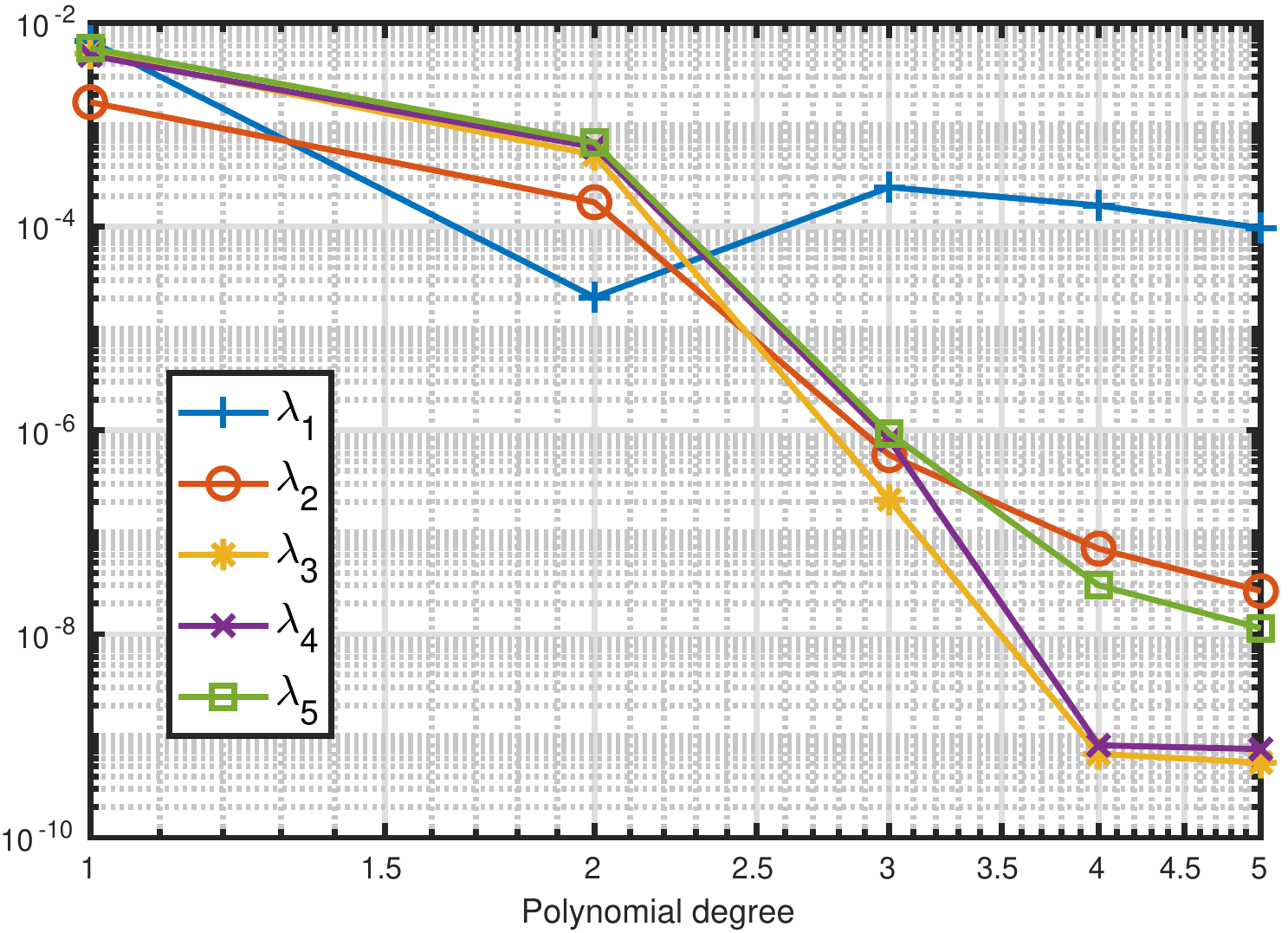}
  \caption{Spectral correctness check for the L-shape domain: first six eigenvalues (on the left). On the right we show convergence under $p$ refinement for the approximation of the first five non-zero eigenvalues.}
  \label{fig:eigvals_lshape}
\end{figure}

As a final test we show how the method behaves when singular solutions are expected. To this end we use the celebrated $L$-shaped domain: $\Omega = \{[-1,1]\times[-1,1]\}\setminus \{[0,1]\times[-1,0]\}$, for which the first six eigenfunctions when solving the Neumann problem are shown in Fig.~\ref{fig:eigfuns_lshape}, computed with a fine mesh. We show the six associated eigenvalues in \ref{fig:eigvals_lshape}, where values from \cite{dauge} (numerically estimated with the standard FEM, with eleven digits expected to be correct) are taken as a reference solution. Again no spurious solutions are observed. Naturally, optimal convergence cannot be expected (at least not with a naive mesh-refinement strategy) for the second and for the sixth eigenvalue, as the associated eigenfunctions have a strong unbounded singularity at the origin. Restoring optimal convergence by appropriate $hp$--refinement goes outside of the scope of the present contribution, while again providing an obvious research direction for future work.
\section{Conclusions}\label{sec:conclusio}
The proposed method presents very promising approximation properties, as shown both by theoretical and numerical investigations. Its potential for high-performance is preserved by the block-diagonal structure of the mass-matrices. Furthermore, the arbitrary order version also preserves the explicit splitting of the involved discrete operators into topological and geometric ones.
It has also not escaped our notice that, with slight modifications in the definitions of shape-functions and transformation rules, a method applicable to the acoustic wave equation (in the velocity--pressure first order formulation) instead of the Maxwell system can be obtained.
In lack of a complete theory, we hope that its extension to three spatial dimensions, which is currently being carried out and will be the topic of a subsequent submission, will further show its effectiveness as a fast solver for the time-dependent Maxwell equations.
Nevertheless, a more thorough theoretical analysis for the introduced functional spaces and the development of a spectral theory for the involved operators is a mandatory question to be investigated by researchers.

On a more critical note, we remark that the proposed local shape-functions, although in principle of arbitrary degree and hierarchical, are not practical for polynomial degrees $p>5$, since bases consisting of scaled monomials on a subset of the unit square will quickly yield ill-conditioned mass-matrix blocks (see also \cite{poptimal_basis}). This can be mended by partial orthonormalization techniques which do not pose any drastic theoretical hardships.

We finally remark that a reduction of the present high-order method to Cartesian-orthogonal meshes is straightforward (via the same barycentric subdivision procedure), and the resulting scheme degenerates to Yee's algorithm when piecewise-constant bases are used.

\section*{Acknowledgements}
Author Bernard Kapidani was financially supported by the Austrian Science Fund (FWF) under grant number F65: \emph{Taming Complexity in Partial Differential Equations}.

\appendix
\section{Explicit construction of the barycentric-dual cellular complex}\label{sec:app1}
For a the primal complex to be a conforming triangulation of $\Omega$, the following axioms are required to hold:
\begin{align}
& \forall \sigma_k \!\in\!\mathcal{C}_k^\Omega, \;\; \exists \sigma_{k+1} \!\in\! \mathcal{C}_{k+1}^\Omega \; s.t. \; \sigma\!\subset\!\partial \sigma_{k+1}, \;
( k\!\in\!\{0,1\} ), \label{eq:bnd_operator}\\
& \sigma' \cap \sigma'' \subset \mathcal{S}^{k-1}(\mathcal{C}^\Omega), \;\;\; \forall \sigma',\sigma''\!\in\! \mathcal{C}_k^\Omega,\, \sigma'\neq\sigma'', \label{eq:non_intersection}\\
& \bar{\Omega} = \bigcup_{k=0}^{k=2}\! \mathcal{C}_k^\Omega, \label{eq:closure_of_omega}
\end{align}
\noindent where $\bar{\Omega}$ denotes the closure of $\Omega$. In plain words: no hanging edges and nodes are allowed, and no overlap of equal dimensional simplexes.
The setup for the process of barycentric subdivision requires the definition of centroid (or barycenter) for a $k$-simplex:
\begin{align*}
&\overline{\mathbf{r}}(\sigma) =
\frac{1}{k\!+\!1} \hspace{-5mm} \sum\limits_{\mathbf{v} \in \{\mathcal{C}_0^\Omega \cap\partial\sigma\} } \hspace{-5mm}\!\mathbf{v},\;\;
\forall \sigma \!\in\!\mathcal{C}_k^\Omega,
\end{align*}
\noindent where we explicitly identify 0-simplexes with the associated Euclidean position vectors. With some more harmless abuse of notation we also define the half-open oriented line segment from point $\mathbf{r}_1$ to point $\mathbf{r}_2$ as 
\begin{align*}
& ]\mathbf{r}_1,\mathbf{r}_2] = \{ (1-\ell)\mathbf{r}_1+\ell\mathbf{r}_2 \; s.t. \; \ell\!\in ]0,1]\},
\end{align*}
\noindent and we remark that the extension of the above definition to closed straight segments $[\mathbf{r}_1,\mathbf{r}_2]$ and open straight segments $]\mathbf{r}_1,\mathbf{r}_2[$ is trivially deduced. We can now introduce the sets involved in the barycentric-dual cellular complex. We define first the duality mapping:
\begin{align*}
D_2\,:\,\mathcal{C}_2^\Omega &\mapsto \tilde{\mathcal{C}}_0^\Omega,\\
\mathcal{T} &\mapsto \overline{\mathbf{r}}(T),
\end{align*}
\noindent i.e. $\tilde{\mathcal{C}}_0^\Omega$ is the set of centroids of triangles, followed by
\begin{align*}
D_1\,:\,\mathcal{C}_1^\Omega  &\mapsto \tilde{\mathcal{C}}_1^\Omega,\\
E &\mapsto \bigcup_{\substack{\mathcal{T} \in {\mathcal{C}}_2^\Omega, \\  E\subset\partial\mathcal{T}} } ]\overline{\mathbf{r}}(T),\overline{\mathbf{r}}(E)],
\end{align*}
\noindent where some more ingenuity was needed: for edges $E\in\mathcal{C}_1^\Omega$ which lie in the interior of $\Omega$, we will always find a pair of triangles $\mathcal{T}$, $\mathcal{T}'$ which share $E$ in their boundaries, but on $\partial\Omega$ we are left with \emph{halved} dual edges (se for example \cite{dgatap}). The above definition of $\tilde{\mathcal{C}}_{\Omega}^1$ clearly accommodates both cases by ``looping'' first over edges and then over triangles with the given edge in their boundary.
With a bit more involved notation we finally introduce
\begin{align*}
D_0\,:\,\mathcal{C}_0^\Omega  &\mapsto \tilde{\mathcal{C}}_2^\Omega,\\
\mathbf{v} &\mapsto \bigcup_{\substack{E \in {\mathcal{C}}_1^\Omega, \\  \mathbf{v}\subset\partial{e}} }
\bigcup_{\substack{\mathcal{T} \in {\mathcal{C}}_2^\Omega, \\  E\subset\partial\mathcal{T}} }
\text{Conv}\{\overline{\mathbf{r}}(T),\overline{\mathbf{r}}(E),\mathbf{v}\},%\label{eq:dual_cells}
\end{align*}
\noindent where, again, $\text{Conv}\{\cdot,\dots,\cdot\}$ denotes the convex hull of its arguments.
Proving that the $D_k$ mappings we have introduced are one-to-one is a matter of straightforward counting. 
%A drawing exemplifying the construction of a simplicial and its dual complex is shown in Fig.\ref{fig:staggered_grids}. 
It is straight-forward to see that (\ref{eq:bnd_operator}), (\ref{eq:non_intersection}) and (\ref{eq:closure_of_omega}) hold also for $\tilde{\mathcal{C}}^\Omega$. What is trickier is the fact that $\partial\Omega$ becomes part of the boundary of dual 2-cells of $\tilde{\mathcal{C}}_2^\Omega$ without actually being the image under $D_1$ of any edge in $\mathcal{C}_1^\Omega$. This is nevertheless not inconsistent with the definition of a cellular complex, but has practical implications for the definition of natural and essential boundary conditions, as further discussed in Section \ref{sec:funspaces} (the reader may also see \cite{refundation} for treatments of the subject for zero-order versions of the CM).

%\section*{References}
%\printbibheading
%\printbibliography
%\bibliographystyle{unsrt}
\bibliographystyle{abbrvnat}
\bibliography{mybibfile}

\end{document}